\def\arxiv{1}    

\let\accentvec\vec 
\documentclass{llncs}
\let\vec\accentvec %

\ifnum\arxiv=1
\fi

\usepackage[utf8]{inputenc}
\usepackage[T1]{fontenc} 

\usepackage{amsmath}
\usepackage{amsfonts}
\usepackage{amssymb}
\usepackage{amsbsy}

\usepackage{graphicx}
\usepackage{capt-of}
\usepackage{xcolor}
\definecolor{darkblue}{rgb}{0,0,.5}

\usepackage{url}
\usepackage{hyperref}
\hypersetup{colorlinks=true, breaklinks=true, linkcolor=darkblue, menucolor=darkblue, urlcolor=darkblue, citecolor=darkblue}

\usepackage{colortbl}
\usepackage{paralist}
\def\Item{\item~\vspace{-2\normalbaselineskip}}

\usepackage[ruled,vlined,linesnumberedhidden]{algorithm2e}
\DontPrintSemicolon

\usepackage{framed}
\renewenvironment{proof}[1][\proofname]{\begin{framed}{\noindent\bfseries #1. }}{\qed\end{framed}}
\newenvironment{proof_sketch}{\textsc{Sketch of Proof.}}{\qed}
\spnewtheorem*{quest}{Question}{\itshape}{\upshape}
\spnewtheorem*{main_theo}{Main Theorem}{\bfseries\upshape}{\itshape}

\ifnum\arxiv=1
\newcommand{\figPath}{.}
\else
\newcommand{\figPath}{plots}
\fi
\newcommand{\imgScale}{0.47}
\renewcommand{\vec}[1]{{\boldsymbol{#1}}}
\newcommand{\hpkn}{H^\dg_{n,p}}
\newcommand{\vhpkn}{H^\vec{\dg} _{n,\vec{p}}}
\newcommand{\vhknm}{H^\vec{\dg}_{n,m,\vec{\alpha}}}

\newcommand{\tvhknm}{\tilde{H}^\vec{\dg}_{n,m,\vec{\alpha}}}

\newcommand{\lo}{a}
\newcommand{\hi}{b}
\newcommand{\de}{\mathrm{d}}

\newcommand{\eps}{\varepsilon}
\newcommand{\Po}{\mathrm{Po}}

\newcommand{\dg}{\ensuremath{k}}
\newcommand{\Adg}{\bar{\dg}}
\newcommand{\uAdg}{\bar{K}}

\newcommand{\zup}{z_{\mathrm{up}}}
\newcommand{\zlo}{z_{\mathrm{lo}}}

\title{\texorpdfstring{On Thresholds for the Appearance of\\2-cores in Mixed Hypergraphs}{On Thresholds for the Appearance of 2-cores in Mixed Hypergraphs}}
\subtitle{(Extended Abstract)\vspace{-0.5cm}}

\author{Michael Rink\thanks{Research supported by DFG grant DI 412/10-2.}}
\institute{Fakultät für Informatik und Automatisierung, Technische Universität Ilmenau \email{michael.rink@tu-ilmenau.de}}

\renewcommand{\subparagraph}{}
\usepackage[noindentafter]{titlesec}
\titlespacing{\section}{0pt}{0.6em}{0.3em}
\titlespacing{\subsection}{0pt}{0.6em}{0.3em}

\begin{document}
\maketitle
\vspace{-0.5cm}
\begin{abstract}
We study thresholds for the appearance of a $2$-core in random hypergraphs 
that are a mixture of a constant number of random uniform hypergraphs each
with a linear number of edges but with different edge sizes.
For the case of two overlapping hypergraphs we give a solution for the optimal 
(expected) number of edges of each size such that the $2$-core threshold for the
resulting mixed hypergraph is maximized. We show that for adequate edge sizes
this threshold exceeds the maximum $2$-core threshold for any random uniform hypergraph,
which can be used to improve the space utilization of several data structures that
rely on this parameter.\vspace{-0.2cm}
\end{abstract}

\section{Introduction}
The $2$-core of a hypergraph $H$ is the largest induced sub-hypergraph (possibly empty),
that has minimum degree at least $2$. It can be obtained via
a simple peeling procedure (Algorithm~\ref{algo:peeling}) that successively removes nodes of degree $1$ together with
their incident edge.
\vspace{-0.5cm}
\begin{algorithm}
 \KwIn{Hypergraph $H$}
 \KwOut{Maximum induced sub-hypergraph with minimum degree $2$.}

  \While{$H$ has a node $v$ of degree $\leq 1$}
 {
   \lIf{$v$ is incident to an edge $e$}{remove $e$ from $H$}\;
   remove $v$ from $H$\;
 }
 \KwRet{$H$}
 \caption{\label{algo:peeling}Peeling}
\end{algorithm}
\vspace{-0.5cm}

Let $\hpkn$ be a random $\dg$-uniform hypergraph with $n$ nodes where each of the possible $\binom{n}{\dg}$ edges
is present with probability $p$ independent of the other edges. In the case that the expected number of edges equals $c\cdot n$ for some
constant $c>0$, the following theorem (conjectured e.g. in~\cite{MWHCFamily1996}, rigorously proved in~\cite{Molloy2004} and independently in~\cite{KimPoisson2006})
gives the threshold for the appearance of a $2$-core in $\hpkn$. Let 
\begin{equation}
 t(\lambda,\dg)= \frac{\lambda}{ \dg \cdot \Big(\Pr\left( \Po\left[\lambda \right] \geq 1\right) \Big)^{\dg-1} } \ ,
\end{equation}
where $\Po[\lambda]$ denotes a Poisson random variable with mean $\lambda$.
\begin{theorem}[{\cite[Theorem 1.2]{Molloy2004}}]
\label{theo:2-core_Molloy}
Let $\dg\geq 3$ be constant, and let $c^*(\dg)=\min_{\lambda>0} t(\lambda,\dg)$.
Then for $p={c\cdot n}/{\binom{n}{\dg}}$ with probability $1-o(1)$ for $n\to \infty$ the following holds:
\begin{compactenum}[$(i)$]
 \item if $c<c^*$ then $\hpkn$ has an empty $2$-core,
 \item if $c>c^*$ then $\hpkn$ has a non-empty $2$-core.
\end{compactenum}
\end{theorem}
\begin{remark}
Actually this is only a special case of \cite[Theorem 1.2]{Molloy2004} which covers $\ell$-cores for $\dg$-uniform hypergraphs for all $\ell \geq 2$, $\dg$ and $\ell$ not both equal to 2. 
\end{remark}
Now consider a mixture of graphs $\hpkn$ on $n$ nodes for different values of $p$ and $\dg$.
Let $\vhpkn$ be a random hypergraph with $n$ nodes where each of the possible $\binom{n}{\dg_i}$ edges
is present with probability $p_i$, given via the vectors $\vec{\dg}=(\dg_1,\dg_2,\ldots,\dg_s)$ and $\vec{p}=(p_1,p_2,\ldots,p_s)$.
While studying cores of hypergraphs in the context of cuckoo hashing the authors of \cite{DGMMPR2010}
described how to extend the analysis of uniform hypergraphs to mixed hypergraphs, which directly leads to 
the following theorem. For
$\vec{\alpha}=(\alpha_1,\alpha_2,\ldots,\alpha_s)\in[0,1]^s$ with $\sum_{i=1}^s\alpha_i=1$ let
\begin{equation}
\label{eq:general_threshold_function}
 t(\lambda,\vec{\dg},\vec{\alpha})= \frac{\lambda}{ \sum\limits_{i=1}^s \alpha_i \cdot \dg_i \cdot \Big(\Pr\left( \Po\left[\lambda \right] \geq 1\right) \Big)^{\dg_i-1} } \ . 
\end{equation}
\begin{theorem}[{generalization of Theorem \ref{theo:2-core_Molloy}, implied by \cite{DGMMPR2010}}]
\label{theo:2-core_general}
Let $s\geq 1$ be constant.
For each $1\leq i\leq s$ let $\dg_i\geq 3$ be constant, and let $\alpha_i\in[0,1]$ be constant, where $\sum_{i=1}^s \alpha_i=1$.
Furthermore let
$c^*(\vec{\dg},\vec{\alpha})=\min_{\lambda>0} t(\lambda,\vec{\dg},\vec{\alpha})$.
Then for $p_i=\alpha_i \cdot {c\cdot n}/{\binom{n}{\dg_i}}$ with probability $1-o(1)$ for $n\to \infty$ the following holds:
\begin{compactenum}[$(i)$]
 \item if $c<c^*$ then $\vhpkn$ has an empty $2$-core,
 \item if $c>c^*$ then $\vhpkn$ has a non-empty $2$-core.
\end{compactenum}
\end{theorem}
Using ideas from \cite[Section 4]{DGMMPR2010} this theorem can be proved along the lines of \cite[Theorem 1.2]{Molloy2004} utilizing that
$\vhpkn$ is a mixture of a constant number of independent hypergraphs.
\begin{remark}
Analogous to Theorem~\ref{theo:2-core_Molloy}, Theorem~\ref{theo:2-core_general} can also be generalized such that it covers $\ell$-cores for all $\ell\geq 2$. 
\end{remark}
Now consider hypergraphs $\vhpkn$ with edge probabilities $p_i=\alpha_i \cdot {c\cdot n}/{\binom{n}{\dg_i}}$ as in Theorem~\ref{theo:2-core_general}.
One can ask the following questions.
\begin{compactenum}
 \item Assume $\vec{\dg}$ is given. What is the optimal vector $\vec{\alpha}^*$ such the that threshold $c^*(\vec{\dg},\vec{\alpha}^*)=:c^*(\vec{\dg})$ is maximal
among all thresholds $c^*(\vec{\dg},\vec{\alpha})$? In other words, we want to solve the following optimization problem
\begin{equation}
\label{eq:optimization_problem}
c^*(\vec{\dg})=\min\limits_{\lambda>0} t(\lambda,\vec{\dg},\vec{\alpha}^*)
=\max\limits_{\vec{\alpha}}\min\limits_{\lambda>0} t(\lambda,\vec{\dg},\vec{\alpha})
 \ .
\end{equation}
 \item Is there a $\vec{\dg}$ such that $\vec{\alpha}^*$ gives some $c^*(\vec{\dg})$
that exceeds the maximum \mbox{$2$-core} threshold $c^*(\dg)$ among all $\dg$-uniform hypergraphs (not mixed),
which is known to be about $0.818$ for $\dg=3$, see e.g.~\cite[conjecture]{HMWCGraphs1993,MWHCFamily1996}, \cite[proof]{CooperCores2004}.
\begin{remark}
Often the $2$-core threshold is given for hypergraph models slightly different from $\hpkn$. The justification that some ``common'' hypergraph models
are equivalent in terms of this threshold is given in Section~\ref{sec:different_hypergraph_modesl}.
\end{remark}

\end{compactenum}

\subsection{Results}
We give the solution for the non-linear optimization problem 
\eqref{eq:optimization_problem} for $s=2$.
That is for each $\vec{\dg}=(\dg_1,\dg_2)$ we
either give optimal solutions $\vec{\alpha}^*=(\alpha^*,1-\alpha^*)$ and
$c^*(\vec{\dg})$ in analytical form
or identify a subset of the interval $(0,1]$ where we can use binary search to determine $\alpha^*$ and therefore 
$c^*(\vec{\dg})$ numerically with arbitrary precision.
Interestingly, it turns out that for adequate edge sizes $\dg_1$ and $\dg_2$ the maximum 2-core threshold $c^*(\vec{\dg})$
exceeds the maximum $2$-core threshold $c^*(k)$ for $\dg$-uniform hypergraphs. The following table lists some values.
\begin{center}
\small
\begin{tabular}{l||c|ccccccc|c}
$(\dg_1,\dg_2)$  & $(3,3)$  & $(3,4)$ & $(3,6)$ & $(3,8)$ & $(3,10)$ & $(3,12)$ & $(3,14)$ & $(3,16)$ & $(3,21)$                       \\\hline\hline
$c^*$            & 0.81847  & 0.82151 & 0.83520 & 0.85138 & 0.86752  & 0.88298  & 0.89761  & 0.91089  & \cellcolor{lightgray}0.92004   \\
$\alpha^*$       & 1        & 0.83596 & 0.85419 & 0.86512 & 0.87315  & 0.87946  & 0.88464  & 0.88684  & 0.88743                        \\
$\Adg$           & 3        & 3.16404 & 3.43744 & 3.67439 & 3.88795  & 4.08482  & 4.26898  & 4.47102  & 5.02626
\end{tabular}
\vspace{-0.1cm}
\captionof{table}{\label{tab:optimal_values}Optimal $2$-core thresholds $c^*(\vec{\dg})$, $\vec{\dg}=(\dg_1,\dg_2)$, and
$\vec{\alpha}^*=(\alpha^*,1-\alpha^*)$, and $\Adg=\alpha^*\cdot \dg_1 +(1-\alpha^*)\cdot\dg_2$.
The values are rounded to the nearest multiple of $10^{-5}$.}
\end{center}
More comprehensive tables for parameters $3\leq \dg_1\leq 6$ and $\dg_1\leq \dg_2\leq 50$ are given in Appendix~\ref{app:optimal_values}.
The maximum threshold found is about $0.92$ for $\vec{\dg}=(3,21)$.
\begin{remark}
So why does it help to use edges of different sizes? Consider a $\dg$-uniform hypergraph $\hpkn$ 
that has a non-empty $2$-core with node set $V$ and edge set $E$.
Let $V'$ be the set of nodes outside $V$, that is $V'\cap V=\emptyset$.
Assume that $c$ is just above $c^*(\dg)$.
Then there are many small sets $E_1,E_2,\ldots\subset E$, such that if one removes all
edges of any of these sets, the $2$-core of the remaining hypergraph would be empty.
Now randomly replace a constant fraction $\beta=1-\alpha$ of the edges of $\hpkn$ by edges of larger size.
If $\beta$ is large enough, then it is likely that there exists a set $E_i$, where all of the edges are substituted 
and all of the corresponding larger edges are incident with nodes from $V'$.
Consider an arbitrary large edge $e$ with $e\cap V'\neq \emptyset$. If $\beta$ is small enough, then it is likely that
there is at least one node $v$ from $e\cap V'$
that is incident to only one large edge (namely $e$).
It follows that $e$ will be removed by the standard peeling algorithm (Algorithm~\ref{algo:peeling}).
Hence, if $\beta$ is not too small and not too large, then it is likely that there exists a set $E_i$
whose edges are substituted by larger edges that will
be removed by the peeling algorithm, which results in an empty $2$-core.
\end{remark}
\subsection{Extensions to Other Hypergraph Models}
\label{sec:different_hypergraph_modesl}
While Theorems~\ref{theo:2-core_Molloy}~and~\ref{theo:2-core_general} are stated for hypergraphs $\vhpkn$,
one often considers slightly different hypergraphs, e.g. in the analysis of data structures.

Let $\vhknm$ and $\tvhknm$ be random hypergraphs with $n$ nodes and $m$ edges,
where for each $1\leq i\leq s$, a fraction of $\alpha_i$ of the edges are fully randomly chosen from the set of all possible edges of size $\dg_i$.
In the case of $\vhknm$ the random edge choices are made \emph{without} replacement and in the case of
$\tvhknm$ the random edge choices are made \emph{with} replacement.
Using standard arguments, one sees that if $m=c\cdot n, \dg_i\geq 3,$ and $p_i=\alpha_i \cdot {c\cdot n}/{\binom{n}{\dg_i}}$
as in the situation of Theorem~\ref{theo:2-core_general},
the $2$-core threshold of $\vhknm$  is the same as for $\vhpkn$ (see e.g. \cite[analogous to Proposition 2]{FPOrientability2010}),
and  the $2$-core threshold of $\tvhknm$ is the same as for $\vhknm$ (see e.g. \cite[analogous to Proposition 1]{FPOrientability2010}).


\subsection{Related Work}
\label{sec:related_work}
Non-uniform hypergraphs have proven very useful in the design of erasure correcting codes,
such as Tornado codes~\cite{LMSSSLoss-Resilient1997,LMSSTornado2001},
LT codes~\cite{LubyLTCodes2002}, Online codes~\cite{MaymounkovOnlineCodes2002}, and Raptor codes~\cite{ShokrollahiRaptorCodes2006}.
Each of these codes heavily rely on one or more hypergraphs where the hyperedges
correspond to variables (input/message symbols) and the nodes correspond to 
constraints on these variables (encoding/check symbols).
An essential part of the decoding process of an encoded message 
is the application of a procedure that can be interpreted as
peeling the hypergraph (see Algorithm~\ref{algo:peeling}) associated with the recovery process,
where it is required that the result is an empty 2-core.
Given $m$ message symbols, carefully designed non-uniform hypergraphs
allow, in contrast to uniform ones, to gain codes where
in the example of Tornado, Online, and Raptor codes a
random set of $(1+\eps)\cdot m$ encoding symbols are necessary to
decode the whole message in linear time (with high probability),
and in the case of LT codes a random set of $m+o(m)$ encoding symbols
are necessary to decode the whole message in time proportional to $m\cdot\ln(m)$ (with high probability).
Tornado codes use explicit underlying hypergraphs designed
for a given fixed code rate, whereas LT codes 
and its improvements, Online and Raptor codes, use
implicit graph constructions to generate essentially 
infinite hypergraphs resulting in so called rateless codes.
In the case of Torndado codes the size of the
hyperedges as well as the degree of the nodes
follow precalculated sequences that are optimized to obtain the desired properties.
In the case of LT codes, as well as in the last stage of Raptor and Online codes
each node chooses its degree at random according to some fixed distribution, and then selects its incident hyperedges uniformly at random.
(For Online codes also a skewed selection of the hyperedges is discussed, see \cite[Section 7]{MaymounkovOnlineCodes2002}.)
While the construction of the non-uniform hypergraph
used for these codes is not quite the same as
for $\tvhknm$ (or $\vhpkn$, $\vhknm$), since, among other reasons, the degree of the nodes is part of the design,
they are similar enough to seemingly make the
optimization methods / heuristics of \cite{LMSSTornado2001} applicable, see footnote~\cite[page 10]{GoodrichM2011}.
Having said that, compared to e.g. \cite{LMSSTornado2001},
our optimization problem is easier in the sense that it has fewer free parameters
and harder in the sense that we are seeking a global optimum.
\subsection{Overview of the Paper}
In the next section we discuss the effect of our results on three succinct data structures.
Afterwards, we give our main theorem that shows how to determine optimal 2-core thresholds
for mixed hypergraphs with two different edge sizes.
It follows a section with experimental evaluation of the appearance of 2-cores for
a few selected mixed hypergraphs, which underpins our theoretical results.
We conclude with a short summary and an open question.

\section{Some Applications to Succinct Data Structures}
\label{sec:applications}
Several succinct data structures are closely related to the $2$-core threshold of $\dg$-uniform hypergraphs
(often one considers $\tvhknm$ for $s=\alpha_1=1$ and $\dg_1=3$).
More precisely, the space usage of these data structures is inversely proportional to the $2$-core threshold $c^*(\dg)$,
while the evaluation time is proportional to the edge size $\dg$. By showing that the value of $c^*(3)$ can be
improved using mixed hypergraphs instead of uniform ones, our result opens a new possibility for a space–time tradeoff
regarding these data structures, allowing to further reduce their space needs at the cost of a constant increase in the evaluation time.
Below we briefly sketch three data structures and discuss possible improvements, where
we make use of the following definitions.

Let $\vec{a}=(a_1,a_2,\ldots,a_n)$ be a vector with $n$ cells each of size $r$ bits.
Let $S=\{x_1,x_2,\ldots,x_m\}$ be a set of $m$ keys, where $S$ is subset of some universe $U$
and it holds $m=c\cdot n$ for some constant $c<1$.
The vector cells correspond to nodes of a hypergraph
and the keys from $S$ are mapped via some function $\varphi$ to a sequence of vector cells and therefore correspond to hyperedges.
We identify cells (and nodes) via their index $i$, $1\leq i\leq n$, whereas $a_i$ stands for the value of cell $i$.
The following three data structures essentially consist of a vector $\vec{a}$ and a mapping $\varphi$.
For each data structure we compare their performance, depending if $\varphi$ realizes a uniform or a mixed hypergraph.
In the case of a uniform hypergraph each key $x_j$ is mapped to $\dg=3$ random nodes $\varphi(x_j)=( g_1(x_j),g_2(x_j),g_3(x_j) )$ 
via functions $g_1,g_2,g_3: U \to \{1,2,\ldots,n\}$.
In the case of a mixed hypergraph, as an example, a fraction of $\alpha^*=0.88684$ keys are mapped to $3$ random nodes 
using functions $g_1,g_2,g_3$ and a fraction of $1-\alpha^*$ keys are mapped to $16$ random nodes via functions $g'_1,g'_2,\ldots,g'_{16}: U \to \{1,2,\ldots,n\}$.
We fix $c$ below the 2-core threshold to $c=c^*-0.005$, which gives $c=0.813$ in the uniform case and $c=0.906$ in the mixed case, cf. Table~\ref{tab:optimal_values}.
The reason why we use a rather small distance of $0.005$ is that for large $m$ one observes a fairly sharp phase transition from
``empty 2-core'' to ``non-empty 2-core'' in experiments, cf. Section~\ref{sec:experiments}.

\subsection{Invertible Bloom Lookup Table}
The invertible Bloom Lookup Table~\cite{GoodrichM2011} (IBLT) is a Bloom filter data structure that,
amongst others, supports a complete listing of the inserted elements (with high probability).
We restrict ourselves to the case where the IBLT is optimized for the listing operation
and we assume without loss of generality that the keys from $S$ are integers.
Each vector cell contains a summation counter and a quantity counter, initialized with $0$.
The keys arrive one by one and are inserted into the IBLT.
Inserting a key $x_j$ adds its value to the summation counter
and increments the quantity counter at each of the cells given via $\varphi(x_j)$.
To list the inserted elements of the IBLT one essentially uses the standard peeling process
for finding the 2-core of the underlying hypergraph (see Algorithm~\ref{algo:peeling}).
While there exists a cell where the quantity counter has value 1,
extract the value of the summation counter of this cell which gives some element $x_j$.
Determine the summation counters and quantity counters associated with $x_j$ via evaluating $\varphi(x_j)$
and subtract $x_j$ from the summation counters and decrement the quantity counters.
With this method a complete listing of the inserted elements is possible if the 2-core of the hypergraph is empty.
Therefore in the case of uniform hypergraphs we get a space usage of $n/m\cdot r \approx 1.23 \cdot r$ bits per key.
As already pointed out by the authors of~\cite{GoodrichM2011}, who highlight parallels to erasure correcting codes (see Section~\ref{sec:related_work}), a
non-uniform version of the IBLT where keys have a different number of associated cells could improve the
maximum fraction $c=m/n$ where a complete listing is successful with high probability.
Using our example of mixed hypergraphs leads to such an improved space usage of about $1.10 \cdot r$ bits per key.

\subsection{Retrieval Data Structure}
Given a set of key-value pairs $\{ (x_j,v_j) \mid x_j\in S, v_j\in R, j\in[m]\}$,
the \emph{retrieval problem} is the problem of building a
function $f: U\to R$ such that for all $x_j$ from $S$ it holds $f(x_j)=v_j$;
for any $y$ from $U\setminus S$ the value $f(y)$ can be an arbitrary element from $R$.
Chazelle et al.~\cite{CKRTBloomier2004} gave a simple and practical construction of a retrieval data structure,
consisting of a vector $\vec{a}$ and some mapping $\varphi$ that has
constant evaluation time, via simply calculating $f(x_j)= \bigoplus\nolimits_{i \in \varphi(x_j)} a_i$.
The construction is based on the following observation, which is stated more explicitly in~\cite{BotelhoPZ2007}.
Let $\vec{v}=(v_1,v_2,\ldots,v_{m})$ be the vector of the function values and
let $\vec{M}$ be the $m\times n$ incidence matrix of the underlying hypergraph,
where the characteristic vector of each hyperedge is a row vector of $\vec{M}$.
If the hypergraph has an empty 2-core then the linear system $\vec{M}\cdot\vec{a} = \vec{v}$ can be solved in linear time.
For appropriate $c$ this gives expected linear construction time.
As before, in the case of uniform hypergraphs the space usage is about $1.23\cdot r$ bits per key, assuming that the values $v_j$ are bit strings of length $r$.
And in our example of mixed hypergraphs the space usage is about $1.10 \cdot r$ bits per key at the cost of
a slight increase of the evaluation time of $f$.

In \cite{DPSuccinct2008} it is shown how to obtain a retrieval data structure with space usage of $(1 + \eps)\cdot r$ bits per key,
for any fixed $\eps>0$, evaluation time $O(\log(1/\eps))$, and linear expected construction time, while using essentially the same construction as above.
The central idea is to transfer the problem of solving one large linear system into the problem of solving
many small linear systems, where each system fits into a single memory word and can be solved via precomputed pseudoinverses.
As shown in~\cite{ADRExperimental2009} this approach is limited in its practicability but can be adapted to build
retrieval data structures with $1.10\cdot r$ bits per key (and fewer) for realistic key set sizes. But this modified 
construction could possibly be outperformed by our direct approach of solving one large linear system in expected linear time.

\subsection{Perfect Hash Function}
Given a set of keys $S$, the problem of \emph{perfect hashing} is to build a function $h:U\to \{1,2,\ldots,n\}$ that is 1-to-1 on $S$.
The construction from \cite{BotelhoPZ2007} and \cite{CKRTBloomier2004} gives a data structure consisting of a vector $\vec{a}$
and some mapping $\varphi$ that has constant evaluation time.
Formulated in the context of retrieval, one builds a vector $\vec{v}=(v_1,v_2,\ldots,v_m)$ such that
each key $x_j$ is associated with a value $f(x_j)=v_j$ that is the index $\iota$ of the position of a node in the sequence $\varphi(x_j)$.
This node must have the property that if one applies the peeling process to the underlying hypergraph (Algorithm~\ref{algo:peeling})
it will be selected and removed because it gets degree 1.
If $c$ is below the 2-core threshold then with high probability for each $x_j$ there exists such an index $\iota$, and the
linear system $\vec{M}\cdot\vec{a} = \vec{v}$ can be solved in linear time.
Given the vector $\vec{a}$ the evaluation of $h$ is done via $h(x_j)={\varphi(x_j)_\iota}$ where $\iota= \bigoplus\nolimits_{i \in \varphi(x_j)} a_i$.

In the case of a $3$-uniform hypergraph one gets a space usage of about $1.23\cdot 2$ bits per key,
since there are at most 3 different entries in $\vec{a}$. If one applies a simple compression method that stores
every $5$ consecutive elements from $\vec{a}$ in one byte,
one gets a space usage of about $1.23\cdot 8/5\approx 1.97$ bits per key. The range of $h$ is $n=1.23\cdot m$.

In contrast to the examples above, improving this data structure by simply using
a mixed hypergraph is not necessarily successful, since the increase of the load $c$ is compensated by the increase of the
maximum index in the sequence $\varphi(x_j)$, which in our example would lead to a space usage of about $1.10\cdot 4$ bits per key
for uncompressed $\vec{a}$, since we use up to $16$ functions for $\varphi(x_j)$.
However, this can be circumvented by modifying the construction of the vector $\vec{v}$ as follows.
Let $G=(S \cup \{1,2,\ldots,n\}, E)$ be a bipartite graph with edge set $E=\{ \{x,g_\iota(x_j)\} \mid x_j \in S, \iota\in\{1,2,3\} \}$.
According to the results on $3$-ary cuckoo hashing, see e.g.~\cite{FPOrientability2010,DGMMPR2010},
it follows that for $c<0.917$ (as in our case) the graph
$G$ has a left-perfect matching with high probability. 
Given such a matching one stores in $\vec{v}$ for each key $x_j$ the index $\iota$ of $g_\iota$ that has the property that $\{x_j,g_\iota(x_j)\}$ is a matching edge.
Now given the solution of $\vec{M}\cdot\vec{a} = \vec{v}$, the function $h$ is evaluated via $h(x_j)={g_\iota(x_j)}$ where $\iota= \bigoplus\nolimits_{i \in \varphi(x_j)} a_i$.
Since $\vec{a}$ has at most three different entries it follows that the 
space usage in our mixed hypergraph case is about $1.10\cdot 2$ bits per key.
Using the same compression as before, the space usage can be reduced to about $1.10\cdot 8/5=1.76$ bits per key.
Now the range of $h$ is $n=1.10\cdot m$.
Solving the linear system can be done in expected linear time. It is conjectured that if $G$ has a matching then it is found by
the ($\dg,1$)-generalized selfless algorithm from~\cite[Section 5]{DGMMPR2009}; this algorithm can be implemented to work in expected linear time.

A more flexible trade-off between space usage and range 
yields the CHD algorithm from~\cite{BBDHash2009}. This algorithm 
allows to gain ranges $n=(1+\eps)\cdot m$ for arbitrary $\eps > 0$ in combination with a adjustable compression rate that
depends on some parameter $\lambda$. For example, using a range of about $1.11 \cdot m$, a space usage of $1.65$ bits per key is achievable,
see \cite[Fig. 1(b), $\lambda=5$]{BBDHash2009}. But since the expected construction time of the CHD algorithm is $O(m\cdot (2^\lambda+(1/\eps)^\lambda))$
\cite[Theorem 2]{BBDHash2009}, our approach could be faster for a comparable space usage and range.

\section{Maximum Thresholds for the Case $s=2$}
In this section we state our main theorem that gives a solution for the
non-linear optimization problem \eqref{eq:optimization_problem} 
for the case $s=2$, that is given two edge sizes 
we show how to compute the optimal (expected) fraction of edges of each size 
such that the threshold of the appearance of a $2$-core
of a random hypergraph using this configuration is maximal.

Let $\vec{\dg}=(\lo,\hi)$ with $\lo\geq 3$, and $\hi>\lo$.
Furthermore, let
$\vec{\alpha}=(\alpha,1-\alpha)$
and\footnote{We can exclude the case $\alpha=0$, since 
if $3\leq \lo<\hi$, then it holds that $c^*(\lo)>c^*(\hi)$.}
$\alpha\in(0,1]$,
as well as $\lambda\in(0,+\infty)$. 
Consider the following threshold function as a special case of \eqref{eq:general_threshold_function}
\begin{equation}
 t(\lambda,\lo,\hi,\alpha)=\frac{\lambda}{ \alpha\cdot \lo \cdot (1-e^{-\lambda})^{\lo-1}+(1-\alpha)\cdot \hi \cdot(1-e^{-\lambda})^{\hi-1} } \ .
\end{equation}
We transform $t(\lambda,\lo,\hi,\alpha)$ in a more manageable function using a monotonic and bijective domain mapping via  
$z=1-e^{-\lambda}$ and $\lambda=-\ln(1-z)$.
Hence the transformed threshold function is
\begin{equation}
\label{eq:transformed_threshold_function}
 T(z,\lo,\hi,\alpha)=\frac{-\ln(1-z)}{\alpha \cdot \lo \cdot z^{\lo-1}+ (1-\alpha)\cdot \hi\cdot z^{\hi-1} } \ ,
\end{equation}
where $z\in(0,1)$.
According to \eqref{eq:optimization_problem} and using $T(z,\lo,\hi,\alpha)$ instead of $t(\lambda,\lo,\hi,\alpha)$
the optimization problem is defined as
\begin{equation}
\label{eq:transformed_optimization_problem}
 \max_{\alpha\in(0,1]}\min_{z\in(0,1)} T(z,\lo,\hi,\alpha) \ .
\end{equation}
For a short formulation of our results we make use of the following three auxiliary functions.
\begin{align}
 f(z)&=\frac{-\ln(1-z)\cdot (1-z)}{z} \\
 g(z,\lo,\hi)&=f(z)\cdot (\hi-1)\cdot(\lo-1) + \frac{1}{1-z} +2 -\hi -\lo \\
 h(z,\lo,\hi)&=\frac{ \lo\cdot z^{\lo -\hi} -\hi -f(z)\cdot (\lo\cdot(\lo-1)\cdot z^{\lo-\hi}-\hi\cdot(\hi-1)) }{\hi\cdot ((\hi-1)\cdot f(z)-1)} \ .
\end{align}
Furthermore we need to define some ``special'' points.
\begin{align}
 z' &= \left( \frac{\lo}{\hi} \right)^{\frac{1}{\hi-\lo}} & &
 z_l = f^{-1}\left(\frac{1}{\lo-1}\right) & & z_r = f^{-1}\left(\frac{1}{\hi-1}\right) \\ 
 z_1 &= \min \{z \mid g(z)=0\} & & z_2 = \max \{z \mid g(z)=0\} \ .
\end{align}
It can be shown that if $z_1$ and $z_2$ exist, then it holds $z'\neq z_1$ and $z'\neq z_2$. Now we can state our main theorem\footnote{For any function $\phi=\phi(\cdot,x)$ we will use $\phi(\cdot)$ and $\phi(\cdot,x)$ synonymously, if $x$ is considered to be fixed.}.
\begin{theorem}
\label{theo:main}
Let $\lo,\hi$ be fixed and let $T({z}^*,{\alpha}^*)=\max\limits_{\alpha\in(0,1]}\min\limits_{z\in(0,1)} T(z,\alpha)$.
Then the following holds:
\begin{compactenum}[1.]
 \item Let $\min_z g(z)\geq  0$.
\begin{compactenum}[$(i)$]
\item If $h(z')\leq 1$ then the optimal point is
$({z}^*,{\alpha}^*) =(z_l,1) $
and the maximum threshold is given by
\begin{displaymath}
 T({z}^*,{\alpha^*})= \frac{-\ln( 1-z_l)}{\lo\cdot z_l^{\lo-1}} \ .
\end{displaymath}
\item If $h(z')>1$ then the optimal point is the saddle point
\begin{displaymath}
({z}^*,{\alpha}^*) =\left(\left(\tfrac{\lo}{\hi}\right)^{\tfrac{1}{\hi-\lo}} ,\tfrac{\hi-1}{\hi-\lo}-\tfrac{1}{f({z}^*)\cdot (\hi-\lo)}\right)
\end{displaymath}
and the maximum threshold is given by
\begin{displaymath}
 T({z}^*,{\alpha^*})= -\ln\left( 1-\left(\frac{\lo}{\hi} \right)^\frac{1}{\hi-\lo}  \right)\cdot \left(\frac{\hi^{\lo-1}}{\lo^{\hi-1}}\right)^{\frac{1}{\hi-\lo}} \ .
\end{displaymath}
\end{compactenum}
\item Let $\min_z g(z)<0$.
\begin{compactenum}[$(i)$]
\item If $h(z')\leq 1$              then the optimum is the same as in case $1(i)$.
\item If $h(z')\in (1,h(z_2)]$      then the optimum is the same as in case $1 (ii)$.
\item If $h(z')\in (h(z_2),h(z_1))$ then there are two optimal points
$(z^*,\alpha^*)$ and $(z^{**},\alpha^*)$. It holds $1/\alpha^*=h(z^*)=h(z^{**})$ and $T(z^*,\alpha^*)=T(z^{**},\alpha^*)$.

The optimal points can be determined numerically using binary search for
the value $\alpha$ that gives
$T(\tilde{z}_1,\alpha)=T(\tilde{z}_2,\alpha)$,
where $\alpha$ is from the interval $[1/h(\zup),1/h(\zlo)]$ and it holds
$h(\tilde{z}_1)=h(\tilde{z}_2)=1/\alpha$, with $\tilde{z}_1$ from $(z_l,\zup)$,
and $\tilde{z}_2$ from $(\zlo,z_r)$.
The (initial) interval for $\alpha$ is:
\begin{compactitem}[$\bullet$]
\item $[1/h(z_1),1/h(z_2)]$, if $z_1<z'<z_2$,
\item $[1/h(z'),1/h(z_2)]$,  if $z'<z_1$,
\item $[1/h(z_1),1/h(z')]$,  if $z'>z_2$. 
\end{compactitem}
\item If $h(z') \in [h(z_1),\infty) $ then the optimum is the same as in case $1 (ii)$.
\end{compactenum}
\end{compactenum}
\end{theorem}
\begin{proof_sketch}
Assume first that $\alpha\in(0,1]$ is arbitrary but fixed, that is we are looking for a
global minimum of \eqref{eq:transformed_threshold_function} in $z$-direction.
Since 
$\lim_{z\to 0} T(z)=\lim_{z\to 1} T(z)=+\infty$ 
and $T(z)$ is continuous for $z\in(0,1)$, a global minimum must be a
point where the first derivative of $T(z)$ is zero, that is a critical point.
Let $\tilde{z}$ be a critical point of $T(z)$
then it must hold $\tilde{z} \in[z_l,z_r)$ and $\alpha=1/h(\tilde{z})$.

Consider the case $\min g(z)> 0$. (The case $\min g(z)=0$ can be handled analogously).
Since $\frac{\partial h(z)}{\partial z}>0 \Leftrightarrow g(z)>0$, the function
$h(z)$ is monotonically increasing in $\in[z_l,z_r)$. Furthermore it holds, if $g(\tilde{z})>0$ then $\tilde{z}$ is a local minimum point of $T(z)$.
It follows that for each $\alpha$ there is only one critical point $\tilde{z}$ and according to
the monotonicity of $T(z)$ this must be a global minimum point. Now consider the function of critical points $\tilde{T}(z):=T(z,1/h(z))$ of $T(z,\alpha)$.
It holds that
\begin{equation*}
 \forall z<z' \colon \frac{\partial \tilde{T}(z)}{\partial z} > 0 \Leftrightarrow g(z)>0 \text{ and }
 \forall z>z' \colon \frac{\partial \tilde{T}(z)}{\partial z} < 0 \Leftrightarrow g(z)>0 \ .
\end{equation*}
It follows that the function of critical points has a global maximum at $z'=\left( \frac{\lo}{\hi} \right)^{\frac{1}{\hi-\lo}}$,
where $z'$ is at the same time a global minimum of $T(z,\alpha)$ in $z$-direction.
If $h(z')>1$ then $\alpha=1/h(z')\in(0,1)$ and the optimum point $(z^*,\alpha^*)$ is $(z',1/h(z'))$, which is the only saddle point of
$T(z,\alpha)$. If $h(z')\leq 1$ then because of the monotonicity of $\tilde{T}(z)$ 
the solution for $\alpha^*$ is $1$ (degenerated solution). Since $h(z_l)=1$ it follows that that $(z^*,\alpha^*)=(z_l,1)$.

Consider the case $\min g(z)< 0$. The function $g(z)$ has exactly two roots, $z_1$ and $z_2$,
and for $z\in(z_l,z_r)$ the function $h(z)$ is strictly increasing to a local maximum at $z_1$,
is then strictly decreasing to a local minimum at $z_2$, and is strictly increasing afterwards.
Now for fixed $\alpha$ there can be more than one critical point and one has to do a case-by-case analysis.
A complete proof of the theorem is given in Appendix~\ref{app:proof_main_theorem}.
\end{proof_sketch}
\smallskip
The distinction between case 1 and case 2 of Theorem~\ref{theo:main}
can be done via solving $\frac{\partial g(z)}{\partial z}=0$, for $z\in(0,1)$,
since the function $g(z)$ has only one critical point and this point is a global minimum point. 
Hence, Theorem~\ref{theo:main} can be easily transferred into an algorithm that
determines $\alpha^*, z^*$ and $T(z^*,\alpha^*)$ for given $\vec{\dg}=(\lo,\hi)$.
(The pseudocode of such an algorithm is given at the end of Appendix~\ref{app:proof_main_theorem}.)
Some results for $c^*(\vec{\dg})=t(\lambda^*,\lo,\hi,\alpha^*)=T(z^*,\lo,\hi,\alpha^*)$
for selected $\vec{\dg}=(\lo,\hi)$ are given in Table~\ref{tab:optimal_values}
and Appendix~\ref{app:optimal_values}. 
They show that the optimal $2$-core threshold of mixed hypergraphs can be above the $2$-core threshold for $3$-uniform hypergraphs.
 
\section{Experiments}
\label{sec:experiments}
In this section we consider mixed hypergraphs $\tvhknm$ as described in Section~\ref{sec:different_hypergraph_modesl}.
For the parameters $\vec{\dg}=(\dg_1,\dg_2)\in \{(3,4), (3,8), (3,16), (3,21) \}$ and the corresponding optimal fractions of edge size $\vec{\alpha^*}$ 
we experimentally approximated the point $c^*(\vec{\dg})$ of the phase transition from empty to non-empty $2$-core.

For each fixed tuple $(\vec{\dg},\vec{\alpha^*})$ we performed the following experiments.
We fixed the number of nodes to $n=10^7$ and considered growing equidistant edge densities $c=m/n$. 
The densities covered an interval of size $0.008$ with the theoretical $2$-core threshold $c^*(\vec{\dg})$ in its center.
For each quintuple $(\dg_1,\dg_2,\alpha^*,n,c)$ we constructed $10^2$ random hypergraphs $\tvhknm$
with nodes $\{1,2,\ldots, n\}$ and  $c\cdot \alpha^*\cdot n$ edges of size $\dg_1$ and $c\cdot (1-\alpha^*)\cdot n$ edges of size $\dg_2$.
For the random choices of each edge we used the pseudo random number generator MT19937 ``Mersenne Twister'' of the
GNU Scientific Library~\cite{GNUScientific2011}. Given a concrete hypergraph we applied Algorithm~\ref{algo:peeling}
to determine if the $2$-core is empty. 

A non-empty $2$-core was considered as failure, an empty $2$-core was considered as success.
We measured the failure rate and determined
an approximation of the $2$-core threshold, via fitting the sigmoid function
\begin{equation*}
 \sigma(c;x,y) = (1+\exp( -(c-x)/y ) )^{-1}
\end{equation*}
to the measured failure rate using the ``least squares fit'' of \texttt{gnuplot}~\cite{gnuplot}.
The resulting fit parameter $x=x(\vec{\dg})$ is our approximation of the theoretical threshold~$c^*(\vec{\dg})$.
Table~\ref{tab:approximations} compares $c^*(\vec{\dg})$ and $x(\vec{\dg})$.
The quality of the approximation is quantified in terms of the sum of squares of residuals $\sum_{\mathrm{res}}$.
The results show a difference of theoretical and experimentally estimated threshold of less than $2 \cdot 10^{-4}$.
The corresponding plots of the measured failure rates and the fit function
are shown in Figures~\ref{fig:thresh_3-4},~\ref{fig:thresh_3-8},~\ref{fig:thresh_3-16} and \ref{fig:thresh_3-21}.
\begin{center}
\begin{minipage}{0.5\textwidth}
 \centering
 \scalebox{\imgScale}{
\begingroup
  \makeatletter
  \providecommand\color[2][]{%
    \GenericError{(gnuplot) \space\space\space\@spaces}{%
      Package color not loaded in conjunction with
      terminal option `colourtext'%
    }{See the gnuplot documentation for explanation.%
    }{Either use 'blacktext' in gnuplot or load the package
      color.sty in LaTeX.}%
    \renewcommand\color[2][]{}%
  }%
  \providecommand\includegraphics[2][]{%
    \GenericError{(gnuplot) \space\space\space\@spaces}{%
      Package graphicx or graphics not loaded%
    }{See the gnuplot documentation for explanation.%
    }{The gnuplot epslatex terminal needs graphicx.sty or graphics.sty.}%
    \renewcommand\includegraphics[2][]{}%
  }%
  \providecommand\rotatebox[2]{#2}%
  \@ifundefined{ifGPcolor}{%
    \newif\ifGPcolor
    \GPcolorfalse
  }{}%
  \@ifundefined{ifGPblacktext}{%
    \newif\ifGPblacktext
    \GPblacktexttrue
  }{}%
  \let\gplgaddtomacro\g@addto@macro
  \gdef\gplbacktext{}%
  \gdef\gplfronttext{}%
  \makeatother
  \ifGPblacktext
    \def\colorrgb#1{}%
    \def\colorgray#1{}%
  \else
    \ifGPcolor
      \def\colorrgb#1{\color[rgb]{#1}}%
      \def\colorgray#1{\color[gray]{#1}}%
      \expandafter\def\csname LTw\endcsname{\color{white}}%
      \expandafter\def\csname LTb\endcsname{\color{black}}%
      \expandafter\def\csname LTa\endcsname{\color{black}}%
      \expandafter\def\csname LT0\endcsname{\color[rgb]{1,0,0}}%
      \expandafter\def\csname LT1\endcsname{\color[rgb]{0,1,0}}%
      \expandafter\def\csname LT2\endcsname{\color[rgb]{0,0,1}}%
      \expandafter\def\csname LT3\endcsname{\color[rgb]{1,0,1}}%
      \expandafter\def\csname LT4\endcsname{\color[rgb]{0,1,1}}%
      \expandafter\def\csname LT5\endcsname{\color[rgb]{1,1,0}}%
      \expandafter\def\csname LT6\endcsname{\color[rgb]{0,0,0}}%
      \expandafter\def\csname LT7\endcsname{\color[rgb]{1,0.3,0}}%
      \expandafter\def\csname LT8\endcsname{\color[rgb]{0.5,0.5,0.5}}%
    \else
      \def\colorrgb#1{\color{black}}%
      \def\colorgray#1{\color[gray]{#1}}%
      \expandafter\def\csname LTw\endcsname{\color{white}}%
      \expandafter\def\csname LTb\endcsname{\color{black}}%
      \expandafter\def\csname LTa\endcsname{\color{black}}%
      \expandafter\def\csname LT0\endcsname{\color{black}}%
      \expandafter\def\csname LT1\endcsname{\color{black}}%
      \expandafter\def\csname LT2\endcsname{\color{black}}%
      \expandafter\def\csname LT3\endcsname{\color{black}}%
      \expandafter\def\csname LT4\endcsname{\color{black}}%
      \expandafter\def\csname LT5\endcsname{\color{black}}%
      \expandafter\def\csname LT6\endcsname{\color{black}}%
      \expandafter\def\csname LT7\endcsname{\color{black}}%
      \expandafter\def\csname LT8\endcsname{\color{black}}%
    \fi
  \fi
  \setlength{\unitlength}{0.0500bp}%
  \begin{picture}(7200.00,5040.00)%
    \gplgaddtomacro\gplbacktext{%
      \csname LTb\endcsname%
      \put(1078,767){\makebox(0,0)[r]{\strut{} 0}}%
      \put(1078,1131){\makebox(0,0)[r]{\strut{} 0.1}}%
      \put(1078,1496){\makebox(0,0)[r]{\strut{} 0.2}}%
      \put(1078,1860){\makebox(0,0)[r]{\strut{} 0.3}}%
      \put(1078,2224){\makebox(0,0)[r]{\strut{} 0.4}}%
      \put(1078,2589){\makebox(0,0)[r]{\strut{} 0.5}}%
      \put(1078,2953){\makebox(0,0)[r]{\strut{} 0.6}}%
      \put(1078,3318){\makebox(0,0)[r]{\strut{} 0.7}}%
      \put(1078,3682){\makebox(0,0)[r]{\strut{} 0.8}}%
      \put(1078,4046){\makebox(0,0)[r]{\strut{} 0.9}}%
      \put(1078,4411){\makebox(0,0)[r]{\strut{} 1}}%
      \put(1623,484){\makebox(0,0){\strut{} 0.818}}%
      \put(2322,484){\makebox(0,0){\strut{} 0.819}}%
      \put(3022,484){\makebox(0,0){\strut{} 0.82}}%
      \put(3721,484){\makebox(0,0){\strut{} 0.821}}%
      \put(4421,484){\makebox(0,0){\strut{} 0.822}}%
      \put(5120,484){\makebox(0,0){\strut{} 0.823}}%
      \put(5820,484){\makebox(0,0){\strut{} 0.824}}%
      \put(6519,484){\makebox(0,0){\strut{} 0.825}}%
      \put(308,2771){\rotatebox{-270}{\makebox(0,0){\strut{}failure rate among $100$ random hypergraphs}}}%
      \put(4071,154){\makebox(0,0){\strut{}$c$}}%
      \put(4117,967){\makebox(0,0)[l]{\strut{}$x=0.821466$}}%
      \put(4117,1368){\makebox(0,0)[l]{\strut{}$\sum_{\mathrm{res}}=0.00536087$}}%
    }%
    \gplgaddtomacro\gplfronttext{%
      \csname LTb\endcsname%
      \put(2260,4602){\makebox(0,0)[l]{\strut{}measured data}}%
      \csname LTb\endcsname%
      \put(2260,4382){\makebox(0,0)[l]{\strut{}$(1+e^{-(c-x)/y)})^{-1}$}}%
    }%
    \gplbacktext
    \put(0,0){\includegraphics{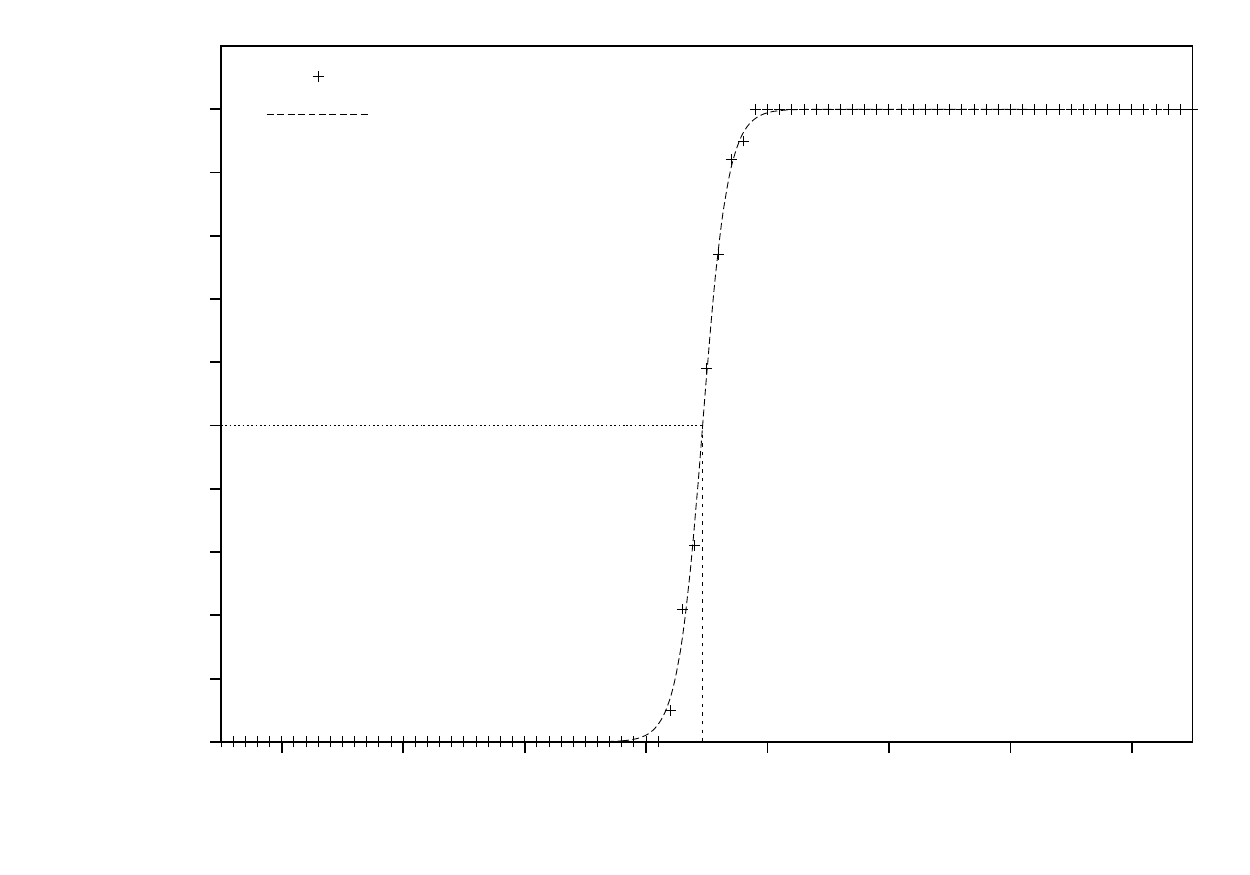}}%
    \gplfronttext
  \end{picture}%
\endgroup
}
 \parbox{0.8\textwidth}{
 \captionof{figure}{\label{fig:thresh_3-4}$(\dg_1,\dg_2)=(3,4)$}}
\end{minipage}\hfill
\begin{minipage}{0.5\textwidth}
 \centering
 \scalebox{\imgScale}{
\begingroup
  \makeatletter
  \providecommand\color[2][]{%
    \GenericError{(gnuplot) \space\space\space\@spaces}{%
      Package color not loaded in conjunction with
      terminal option `colourtext'%
    }{See the gnuplot documentation for explanation.%
    }{Either use 'blacktext' in gnuplot or load the package
      color.sty in LaTeX.}%
    \renewcommand\color[2][]{}%
  }%
  \providecommand\includegraphics[2][]{%
    \GenericError{(gnuplot) \space\space\space\@spaces}{%
      Package graphicx or graphics not loaded%
    }{See the gnuplot documentation for explanation.%
    }{The gnuplot epslatex terminal needs graphicx.sty or graphics.sty.}%
    \renewcommand\includegraphics[2][]{}%
  }%
  \providecommand\rotatebox[2]{#2}%
  \@ifundefined{ifGPcolor}{%
    \newif\ifGPcolor
    \GPcolorfalse
  }{}%
  \@ifundefined{ifGPblacktext}{%
    \newif\ifGPblacktext
    \GPblacktexttrue
  }{}%
  \let\gplgaddtomacro\g@addto@macro
  \gdef\gplbacktext{}%
  \gdef\gplfronttext{}%
  \makeatother
  \ifGPblacktext
    \def\colorrgb#1{}%
    \def\colorgray#1{}%
  \else
    \ifGPcolor
      \def\colorrgb#1{\color[rgb]{#1}}%
      \def\colorgray#1{\color[gray]{#1}}%
      \expandafter\def\csname LTw\endcsname{\color{white}}%
      \expandafter\def\csname LTb\endcsname{\color{black}}%
      \expandafter\def\csname LTa\endcsname{\color{black}}%
      \expandafter\def\csname LT0\endcsname{\color[rgb]{1,0,0}}%
      \expandafter\def\csname LT1\endcsname{\color[rgb]{0,1,0}}%
      \expandafter\def\csname LT2\endcsname{\color[rgb]{0,0,1}}%
      \expandafter\def\csname LT3\endcsname{\color[rgb]{1,0,1}}%
      \expandafter\def\csname LT4\endcsname{\color[rgb]{0,1,1}}%
      \expandafter\def\csname LT5\endcsname{\color[rgb]{1,1,0}}%
      \expandafter\def\csname LT6\endcsname{\color[rgb]{0,0,0}}%
      \expandafter\def\csname LT7\endcsname{\color[rgb]{1,0.3,0}}%
      \expandafter\def\csname LT8\endcsname{\color[rgb]{0.5,0.5,0.5}}%
    \else
      \def\colorrgb#1{\color{black}}%
      \def\colorgray#1{\color[gray]{#1}}%
      \expandafter\def\csname LTw\endcsname{\color{white}}%
      \expandafter\def\csname LTb\endcsname{\color{black}}%
      \expandafter\def\csname LTa\endcsname{\color{black}}%
      \expandafter\def\csname LT0\endcsname{\color{black}}%
      \expandafter\def\csname LT1\endcsname{\color{black}}%
      \expandafter\def\csname LT2\endcsname{\color{black}}%
      \expandafter\def\csname LT3\endcsname{\color{black}}%
      \expandafter\def\csname LT4\endcsname{\color{black}}%
      \expandafter\def\csname LT5\endcsname{\color{black}}%
      \expandafter\def\csname LT6\endcsname{\color{black}}%
      \expandafter\def\csname LT7\endcsname{\color{black}}%
      \expandafter\def\csname LT8\endcsname{\color{black}}%
    \fi
  \fi
  \setlength{\unitlength}{0.0500bp}%
  \begin{picture}(7200.00,5040.00)%
    \gplgaddtomacro\gplbacktext{%
      \csname LTb\endcsname%
      \put(1078,767){\makebox(0,0)[r]{\strut{} 0}}%
      \put(1078,1131){\makebox(0,0)[r]{\strut{} 0.1}}%
      \put(1078,1496){\makebox(0,0)[r]{\strut{} 0.2}}%
      \put(1078,1860){\makebox(0,0)[r]{\strut{} 0.3}}%
      \put(1078,2224){\makebox(0,0)[r]{\strut{} 0.4}}%
      \put(1078,2589){\makebox(0,0)[r]{\strut{} 0.5}}%
      \put(1078,2953){\makebox(0,0)[r]{\strut{} 0.6}}%
      \put(1078,3318){\makebox(0,0)[r]{\strut{} 0.7}}%
      \put(1078,3682){\makebox(0,0)[r]{\strut{} 0.8}}%
      \put(1078,4046){\makebox(0,0)[r]{\strut{} 0.9}}%
      \put(1078,4411){\makebox(0,0)[r]{\strut{} 1}}%
      \put(1693,484){\makebox(0,0){\strut{} 0.848}}%
      \put(2392,484){\makebox(0,0){\strut{} 0.849}}%
      \put(3092,484){\makebox(0,0){\strut{} 0.85}}%
      \put(3791,484){\makebox(0,0){\strut{} 0.851}}%
      \put(4491,484){\makebox(0,0){\strut{} 0.852}}%
      \put(5190,484){\makebox(0,0){\strut{} 0.853}}%
      \put(5890,484){\makebox(0,0){\strut{} 0.854}}%
      \put(6589,484){\makebox(0,0){\strut{} 0.855}}%
      \put(308,2771){\rotatebox{-270}{\makebox(0,0){\strut{}failure rate among $100$ random hypergraphs}}}%
      \put(4071,154){\makebox(0,0){\strut{}$c$}}%
      \put(4108,967){\makebox(0,0)[l]{\strut{}$x=0.851353$}}%
      \put(4108,1368){\makebox(0,0)[l]{\strut{}$\sum_{\mathrm{res}}=0.00175451$}}%
    }%
    \gplgaddtomacro\gplfronttext{%
      \csname LTb\endcsname%
      \put(2260,4602){\makebox(0,0)[l]{\strut{}measured data}}%
      \csname LTb\endcsname%
      \put(2260,4382){\makebox(0,0)[l]{\strut{}$(1+e^{-(c-x)/y)})^{-1}$}}%
    }%
    \gplbacktext
    \put(0,0){\includegraphics{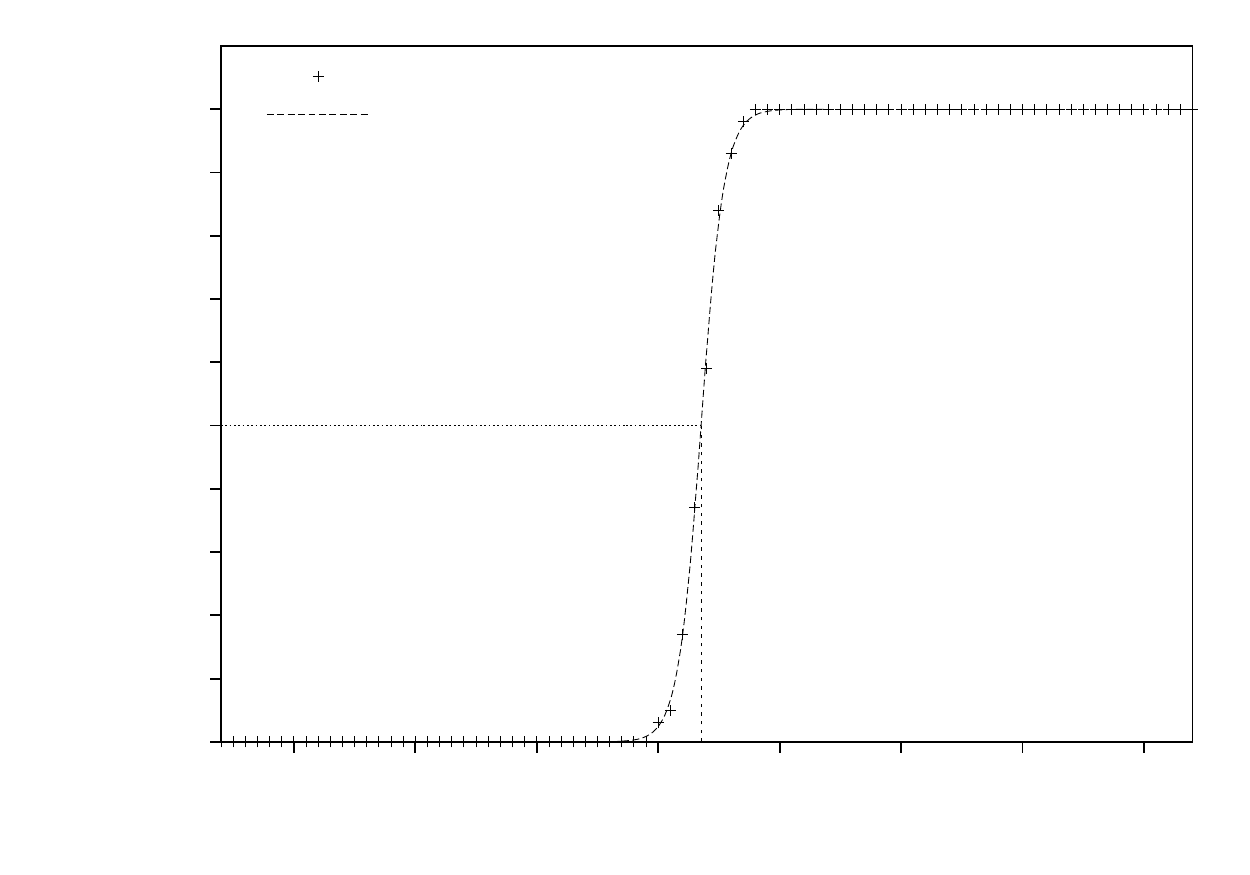}}%
    \gplfronttext
  \end{picture}%
\endgroup
}
 \parbox{0.8\textwidth}{
 \captionof{figure}{\label{fig:thresh_3-8}$(\dg_1,\dg_2)=(3,8)$}}
\end{minipage}\hfill
\vspace{0.2cm}
\noindent
\begin{minipage}{0.5\textwidth}
 \centering
 \scalebox{\imgScale}{
\begingroup
  \makeatletter
  \providecommand\color[2][]{%
    \GenericError{(gnuplot) \space\space\space\@spaces}{%
      Package color not loaded in conjunction with
      terminal option `colourtext'%
    }{See the gnuplot documentation for explanation.%
    }{Either use 'blacktext' in gnuplot or load the package
      color.sty in LaTeX.}%
    \renewcommand\color[2][]{}%
  }%
  \providecommand\includegraphics[2][]{%
    \GenericError{(gnuplot) \space\space\space\@spaces}{%
      Package graphicx or graphics not loaded%
    }{See the gnuplot documentation for explanation.%
    }{The gnuplot epslatex terminal needs graphicx.sty or graphics.sty.}%
    \renewcommand\includegraphics[2][]{}%
  }%
  \providecommand\rotatebox[2]{#2}%
  \@ifundefined{ifGPcolor}{%
    \newif\ifGPcolor
    \GPcolorfalse
  }{}%
  \@ifundefined{ifGPblacktext}{%
    \newif\ifGPblacktext
    \GPblacktexttrue
  }{}%
  \let\gplgaddtomacro\g@addto@macro
  \gdef\gplbacktext{}%
  \gdef\gplfronttext{}%
  \makeatother
  \ifGPblacktext
    \def\colorrgb#1{}%
    \def\colorgray#1{}%
  \else
    \ifGPcolor
      \def\colorrgb#1{\color[rgb]{#1}}%
      \def\colorgray#1{\color[gray]{#1}}%
      \expandafter\def\csname LTw\endcsname{\color{white}}%
      \expandafter\def\csname LTb\endcsname{\color{black}}%
      \expandafter\def\csname LTa\endcsname{\color{black}}%
      \expandafter\def\csname LT0\endcsname{\color[rgb]{1,0,0}}%
      \expandafter\def\csname LT1\endcsname{\color[rgb]{0,1,0}}%
      \expandafter\def\csname LT2\endcsname{\color[rgb]{0,0,1}}%
      \expandafter\def\csname LT3\endcsname{\color[rgb]{1,0,1}}%
      \expandafter\def\csname LT4\endcsname{\color[rgb]{0,1,1}}%
      \expandafter\def\csname LT5\endcsname{\color[rgb]{1,1,0}}%
      \expandafter\def\csname LT6\endcsname{\color[rgb]{0,0,0}}%
      \expandafter\def\csname LT7\endcsname{\color[rgb]{1,0.3,0}}%
      \expandafter\def\csname LT8\endcsname{\color[rgb]{0.5,0.5,0.5}}%
    \else
      \def\colorrgb#1{\color{black}}%
      \def\colorgray#1{\color[gray]{#1}}%
      \expandafter\def\csname LTw\endcsname{\color{white}}%
      \expandafter\def\csname LTb\endcsname{\color{black}}%
      \expandafter\def\csname LTa\endcsname{\color{black}}%
      \expandafter\def\csname LT0\endcsname{\color{black}}%
      \expandafter\def\csname LT1\endcsname{\color{black}}%
      \expandafter\def\csname LT2\endcsname{\color{black}}%
      \expandafter\def\csname LT3\endcsname{\color{black}}%
      \expandafter\def\csname LT4\endcsname{\color{black}}%
      \expandafter\def\csname LT5\endcsname{\color{black}}%
      \expandafter\def\csname LT6\endcsname{\color{black}}%
      \expandafter\def\csname LT7\endcsname{\color{black}}%
      \expandafter\def\csname LT8\endcsname{\color{black}}%
    \fi
  \fi
  \setlength{\unitlength}{0.0500bp}%
  \begin{picture}(7200.00,5040.00)%
    \gplgaddtomacro\gplbacktext{%
      \csname LTb\endcsname%
      \put(1078,767){\makebox(0,0)[r]{\strut{} 0}}%
      \put(1078,1131){\makebox(0,0)[r]{\strut{} 0.1}}%
      \put(1078,1496){\makebox(0,0)[r]{\strut{} 0.2}}%
      \put(1078,1860){\makebox(0,0)[r]{\strut{} 0.3}}%
      \put(1078,2224){\makebox(0,0)[r]{\strut{} 0.4}}%
      \put(1078,2589){\makebox(0,0)[r]{\strut{} 0.5}}%
      \put(1078,2953){\makebox(0,0)[r]{\strut{} 0.6}}%
      \put(1078,3318){\makebox(0,0)[r]{\strut{} 0.7}}%
      \put(1078,3682){\makebox(0,0)[r]{\strut{} 0.8}}%
      \put(1078,4046){\makebox(0,0)[r]{\strut{} 0.9}}%
      \put(1078,4411){\makebox(0,0)[r]{\strut{} 1}}%
      \put(1343,484){\makebox(0,0){\strut{} 0.907}}%
      \put(2042,484){\makebox(0,0){\strut{} 0.908}}%
      \put(2742,484){\makebox(0,0){\strut{} 0.909}}%
      \put(3441,484){\makebox(0,0){\strut{} 0.91}}%
      \put(4141,484){\makebox(0,0){\strut{} 0.911}}%
      \put(4840,484){\makebox(0,0){\strut{} 0.912}}%
      \put(5540,484){\makebox(0,0){\strut{} 0.913}}%
      \put(6239,484){\makebox(0,0){\strut{} 0.914}}%
      \put(308,2771){\rotatebox{-270}{\makebox(0,0){\strut{}failure rate among $100$ random hypergraphs}}}%
      \put(4071,154){\makebox(0,0){\strut{}$c$}}%
      \put(4004,967){\makebox(0,0)[l]{\strut{}$x=0.910704$}}%
      \put(4004,1368){\makebox(0,0)[l]{\strut{}$\sum_{\mathrm{res}}=0.00347648$}}%
    }%
    \gplgaddtomacro\gplfronttext{%
      \csname LTb\endcsname%
      \put(2260,4602){\makebox(0,0)[l]{\strut{}measured data}}%
      \csname LTb\endcsname%
      \put(2260,4382){\makebox(0,0)[l]{\strut{}$(1+e^{-(c-x)/y)})^{-1}$}}%
    }%
    \gplbacktext
    \put(0,0){\includegraphics{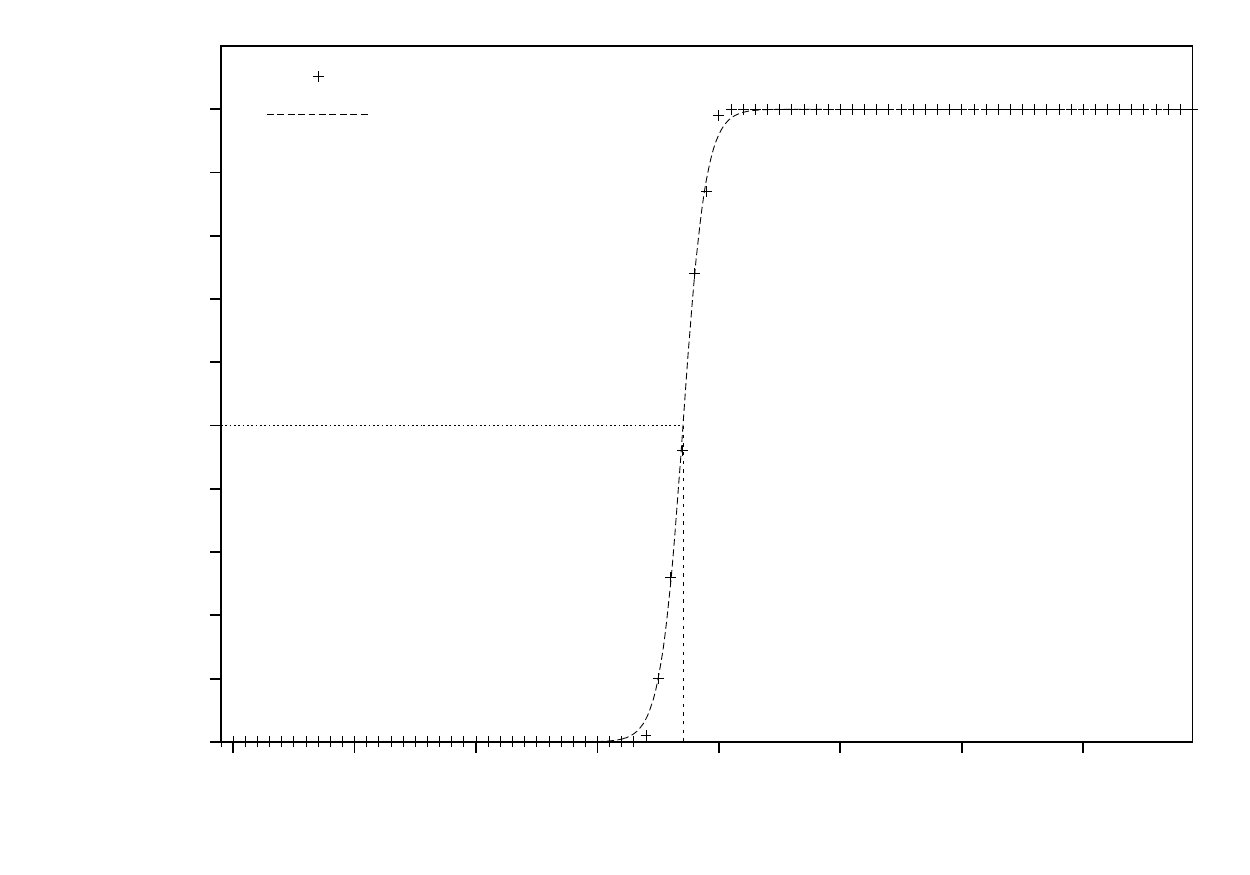}}%
    \gplfronttext
  \end{picture}%
\endgroup
}
 \parbox{0.8\textwidth}{
 \captionof{figure}{\label{fig:thresh_3-16}$(\dg_1,\dg_2)=(3,16)$}}
\end{minipage}\hfill
\begin{minipage}{0.5\textwidth}
 \centering
 \scalebox{\imgScale}{
\begingroup
  \makeatletter
  \providecommand\color[2][]{%
    \GenericError{(gnuplot) \space\space\space\@spaces}{%
      Package color not loaded in conjunction with
      terminal option `colourtext'%
    }{See the gnuplot documentation for explanation.%
    }{Either use 'blacktext' in gnuplot or load the package
      color.sty in LaTeX.}%
    \renewcommand\color[2][]{}%
  }%
  \providecommand\includegraphics[2][]{%
    \GenericError{(gnuplot) \space\space\space\@spaces}{%
      Package graphicx or graphics not loaded%
    }{See the gnuplot documentation for explanation.%
    }{The gnuplot epslatex terminal needs graphicx.sty or graphics.sty.}%
    \renewcommand\includegraphics[2][]{}%
  }%
  \providecommand\rotatebox[2]{#2}%
  \@ifundefined{ifGPcolor}{%
    \newif\ifGPcolor
    \GPcolorfalse
  }{}%
  \@ifundefined{ifGPblacktext}{%
    \newif\ifGPblacktext
    \GPblacktexttrue
  }{}%
  \let\gplgaddtomacro\g@addto@macro
  \gdef\gplbacktext{}%
  \gdef\gplfronttext{}%
  \makeatother
  \ifGPblacktext
    \def\colorrgb#1{}%
    \def\colorgray#1{}%
  \else
    \ifGPcolor
      \def\colorrgb#1{\color[rgb]{#1}}%
      \def\colorgray#1{\color[gray]{#1}}%
      \expandafter\def\csname LTw\endcsname{\color{white}}%
      \expandafter\def\csname LTb\endcsname{\color{black}}%
      \expandafter\def\csname LTa\endcsname{\color{black}}%
      \expandafter\def\csname LT0\endcsname{\color[rgb]{1,0,0}}%
      \expandafter\def\csname LT1\endcsname{\color[rgb]{0,1,0}}%
      \expandafter\def\csname LT2\endcsname{\color[rgb]{0,0,1}}%
      \expandafter\def\csname LT3\endcsname{\color[rgb]{1,0,1}}%
      \expandafter\def\csname LT4\endcsname{\color[rgb]{0,1,1}}%
      \expandafter\def\csname LT5\endcsname{\color[rgb]{1,1,0}}%
      \expandafter\def\csname LT6\endcsname{\color[rgb]{0,0,0}}%
      \expandafter\def\csname LT7\endcsname{\color[rgb]{1,0.3,0}}%
      \expandafter\def\csname LT8\endcsname{\color[rgb]{0.5,0.5,0.5}}%
    \else
      \def\colorrgb#1{\color{black}}%
      \def\colorgray#1{\color[gray]{#1}}%
      \expandafter\def\csname LTw\endcsname{\color{white}}%
      \expandafter\def\csname LTb\endcsname{\color{black}}%
      \expandafter\def\csname LTa\endcsname{\color{black}}%
      \expandafter\def\csname LT0\endcsname{\color{black}}%
      \expandafter\def\csname LT1\endcsname{\color{black}}%
      \expandafter\def\csname LT2\endcsname{\color{black}}%
      \expandafter\def\csname LT3\endcsname{\color{black}}%
      \expandafter\def\csname LT4\endcsname{\color{black}}%
      \expandafter\def\csname LT5\endcsname{\color{black}}%
      \expandafter\def\csname LT6\endcsname{\color{black}}%
      \expandafter\def\csname LT7\endcsname{\color{black}}%
      \expandafter\def\csname LT8\endcsname{\color{black}}%
    \fi
  \fi
  \setlength{\unitlength}{0.0500bp}%
  \begin{picture}(7200.00,5040.00)%
    \gplgaddtomacro\gplbacktext{%
      \csname LTb\endcsname%
      \put(1078,767){\makebox(0,0)[r]{\strut{} 0}}%
      \put(1078,1131){\makebox(0,0)[r]{\strut{} 0.1}}%
      \put(1078,1496){\makebox(0,0)[r]{\strut{} 0.2}}%
      \put(1078,1860){\makebox(0,0)[r]{\strut{} 0.3}}%
      \put(1078,2224){\makebox(0,0)[r]{\strut{} 0.4}}%
      \put(1078,2589){\makebox(0,0)[r]{\strut{} 0.5}}%
      \put(1078,2953){\makebox(0,0)[r]{\strut{} 0.6}}%
      \put(1078,3318){\makebox(0,0)[r]{\strut{} 0.7}}%
      \put(1078,3682){\makebox(0,0)[r]{\strut{} 0.8}}%
      \put(1078,4046){\makebox(0,0)[r]{\strut{} 0.9}}%
      \put(1078,4411){\makebox(0,0)[r]{\strut{} 1}}%
      \put(1273,484){\makebox(0,0){\strut{} 0.916}}%
      \put(1973,484){\makebox(0,0){\strut{} 0.917}}%
      \put(2672,484){\makebox(0,0){\strut{} 0.918}}%
      \put(3372,484){\makebox(0,0){\strut{} 0.919}}%
      \put(4071,484){\makebox(0,0){\strut{} 0.92}}%
      \put(4771,484){\makebox(0,0){\strut{} 0.921}}%
      \put(5470,484){\makebox(0,0){\strut{} 0.922}}%
      \put(6170,484){\makebox(0,0){\strut{} 0.923}}%
      \put(6869,484){\makebox(0,0){\strut{} 0.924}}%
      \put(308,2771){\rotatebox{-270}{\makebox(0,0){\strut{}failure rate among $100$ random hypergraphs}}}%
      \put(4071,154){\makebox(0,0){\strut{}$c$}}%
      \put(4035,967){\makebox(0,0)[l]{\strut{}$x=0.919848$}}%
      \put(4035,1368){\makebox(0,0)[l]{\strut{}$\sum_{\mathrm{res}}=0.0109112$}}%
    }%
    \gplgaddtomacro\gplfronttext{%
      \csname LTb\endcsname%
      \put(2260,4602){\makebox(0,0)[l]{\strut{}measured data}}%
      \csname LTb\endcsname%
      \put(2260,4382){\makebox(0,0)[l]{\strut{}$(1+e^{-(c-x)/y)})^{-1}$}}%
    }%
    \gplbacktext
    \put(0,0){\includegraphics{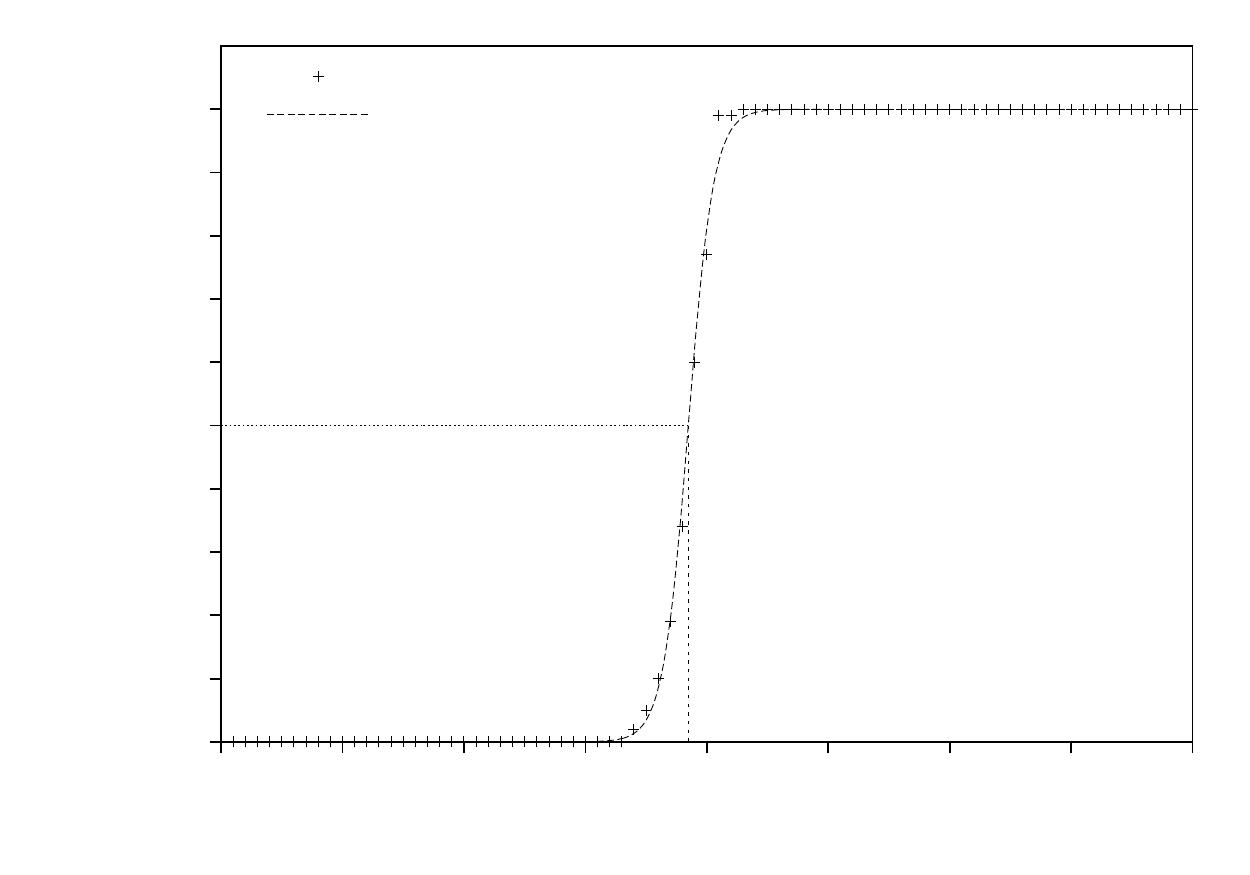}}%
    \gplfronttext
  \end{picture}%
\endgroup
}
 \parbox{0.8\textwidth}{
 \captionof{figure}{\label{fig:thresh_3-21}$(\dg_1,\dg_2)=(3,21)$}}
\end{minipage}
\end{center}
\begin{center}
{\small
\begin{tabular}{l|cccc}
$(\dg_1,\dg_2)$       &  $(3,4)$ &  $(3,8)$ &  $(3,16)$ & $(3,21)$  \\\hline\hline
$c^*$                 &  0.82151 &  0.85138 &  0.91089  & 0.92004   \\
$x$                   &  0.82147 &  0.85135 &  0.91070  & 0.91985   \\
$\sum_{\mathrm{res}}$ &  0.00536 &  0.00175 &  0.00348  & 0.01091  
\end{tabular}}
\captionof{table}{\label{tab:approximations}Comparison of experimentally approximated and theoretical 2-core thresholds.
The values are rounded to the nearest multiple of $10^{-5}$.}
\end{center}

\section{Summary and Future Work}
We have shown that the threshold for the appearance of a $2$-core in mixed hypergraphs
can be larger than the $2$-core threshold for $\dg$-uniform hypergraphs, for each $\dg\geq 3$.
Moreover, for hypergraphs with two given constant edges sizes
we showed how to determine the optimal (expected) fraction
of edges of each size, that maximizes the $2$-core threshold.
The maximum threshold found for $3\leq \dg_1 \leq 6$ and $\dg_1\leq \dg_2\leq 50$
is about $0.92$ for $\vec{\dg}=(3,21)$. We \emph{conjecture} that
this is the best possible for two edge sizes.

Based on the applications of mixed hypergraphs, as for example discussed in Section~\ref{sec:applications},
the following question seems natural to ask.
Consider the hypergraph $\tvhknm$ and some fixed upper bound $\uAdg$
on the average edge size $\Adg=\sum_{i=1}^s \alpha_i \cdot \dg_i$.
\begin{quest}
Which pair of vectors $\vec{\dg}$ and $\vec{\alpha}$ 
that gives an average edge size below $\uAdg$ maximizes the threshold for the appearance of a $2$-core? 
That means we are looking for the solution of
$\max\limits_{\vec{\dg},\vec{\alpha}}\min\limits_{\lambda>0} t(\lambda,\vec{\dg},\vec{\alpha})$ under the constraint that
$\Adg\leq \uAdg$.
\end{quest}
\vspace{-0.6cm}
\subsubsection{Acknowledgment}
The author would like to thank Martin Dietzfelbinger for
many helpful suggestions and the time
he spent in discussions on the topic.
He also would like to thank Udi Wieder for pointing
him to applications of mixed hypergraphs
in LT codes, Online codes and Raptor codes.

\ifnum\arxiv=0

\bibliographystyle{splncs03} 
\bibliography{literature.bib}

\else

\fi

\vfill
\pagebreak
\appendix

\section{Proof of the Main Theorem}
\label{app:proof_main_theorem}
In this section we give the full proof of Theorem~\ref{theo:main}, i.e. we
solve the (transformed) non-linear optimization problem~\eqref{eq:transformed_optimization_problem}.
As is to be expected, the proof mainly employs methods from calculus.
\subsection{Preliminaries}
\subsubsection{Derivatives.}
At first we want to determine the partial derivatives of $T(z,\lo,\hi,\alpha)$ with respect to $z$ and $\alpha$.
To shorten and simplify notation we use the following definitions. For all $j \in \mathbb{N}$ let
\begin{align*}
 D_j(z,\lo,\hi,\alpha)&= \alpha \cdot \lo\cdot (\lo-1)^j\cdot z^{\lo-1} +   (1-\alpha) \cdot \hi \cdot (\hi-1)^j \cdot z^{\hi-1} \text{\hspace{1cm}and} \nonumber\\
 Z_j(z,\lo,\hi)  &= \lo\cdot (\lo-1)^j\cdot        z^{\lo-1} - \hi \cdot (\hi-1)^j\cdot z^{\hi-1} \ .
\end{align*}
The first partial derivatives of $T(z,\alpha)$ are
\begin{align}
\label{eq:first_derivative_z}    \frac{\partial T(z,\alpha)}{\partial z}=& \frac{1}{1-z}\cdot \frac{1}{D_0(z,\alpha)} + \frac{\ln(1-z)}{z} \cdot \frac{D_1(z,\alpha)}{D_0(z,\alpha)^2}  \text{\hspace{1cm}and}\\
\label{eq:first_derivative_alpha}\frac{\partial T(z,\alpha)}{\partial \alpha}=& \frac{\ln(1-z)\cdot Z_0(z)}{D_0(z,\alpha)^2} \ .
\end{align}  
The second partial derivatives of $T(z,\alpha)$ are
\begin{align}
\label{eq:second_derivative_zz} \frac{\partial^2 T(z,\alpha)}{(\partial z)^2} =& \frac{1}{(1-z)^2}\cdot \frac{1}{D_0(z,\alpha)} -\frac{2}{z\cdot (1-z)}\cdot \frac{D_1(z,\alpha)}{D_0(z,\alpha)^2}\\
 &+ \frac{\ln(1-z)}{z^2}\cdot \frac{D_2(z,\alpha)-D_1(z,\alpha)}{D_0(z,\alpha)^2} -\frac{2\cdot \ln(1-z)}{z^2}\cdot \frac{D_1(z,\alpha)^2}{D_0(z,\alpha)^3} \nonumber \\
 \frac{\partial^2 T(z,\alpha)}{(\partial \alpha)^2} =& -\frac{2\cdot \ln(1-z) \cdot Z_0(z)^2 }{D_0(z,\alpha)^3} \\
 \frac{\partial}{\partial z} \left(\frac{\partial T(z,\alpha)}{\partial \alpha}\right)=&
-\frac{1}{1-z}\cdot \frac{Z_0(z)}{D_0(z,\alpha)^2}+\frac{\ln(1-z)}{z}\cdot \frac{Z_1(z)}{D_0(z,\alpha)^2}  \\ 
&-\frac{2\cdot \ln(1-z)}{z}\cdot \frac{Z_0(z)\cdot D_1(z,\alpha)}{D_0(z,\alpha)^3} \nonumber \ .
\end{align}
\subsubsection{Auxiliary Functions.}
Our analysis is heavily based on three functions, 
\begin{align}
 f(z)&=\frac{-\ln(1-z)\cdot (1-z)}{z} \\
 g(z,\lo,\hi)&=f(z)\cdot (\hi-1)\cdot(\lo-1) + \frac{1}{1-z} +2 -\hi -\lo \\
 h(z,\lo,\hi)&=\frac{ \lo\cdot z^{\lo -\hi} -\hi -f(z)\cdot (\lo\cdot(\lo-1)\cdot z^{\lo-\hi}-\hi\cdot(\hi-1)) }{\hi\cdot ((\hi-1)\cdot f(z)-1)}\\
     &=\frac{Z_0(z,\lo,\hi)-f(z) \cdot Z_1(z,\lo,\hi)}{\hi \cdot z^{\hi-1} \cdot \left((\hi-1)\cdot f(z)-1\right)}  \nonumber \ ,
\end{align}
which are shown in Figures~\ref{fig:f(z)}, \ref{fig:g(z)} and \ref{fig:h(z)} for some parameters $\lo$ and $\hi$.
Furthermore we make use of the following definitions,
\begin{align*}
 z' &= \left( \frac{\lo}{\hi} \right)^{\frac{1}{\hi-\lo}} & &
 z_l = f^{-1}\left(\frac{1}{\lo-1}\right) & & z_r = f^{-1}\left(\frac{1}{\hi-1}\right) \\ 
 z_1 &= \min \{z \mid g(z)=0\} & & z_2 = \max \{z \mid g(z)=0\} \ .
\end{align*}
Our line of argument will rely on essential properties of $f,g,h$ and $z_l,z_r,z_1,z_2$ and $z'$.
Proving these properties is standard calculus but unfortunately lengthy.
Therefore the proofs of the next four lemmas are only given in 
extra sections of the appendix. We start with the three auxiliary functions.
\begin{lemma}[Properties of $f(z)$]\\\hfill
\begin{minipage}{0.5\textwidth}
\label{lem:f(z)}
Let $z\in(0,1)$, then it holds
\begin{compactenum}[$(i)$]
 \item \label{prop:f(z)_1-z}$f(z)>1-z>0$.
 \item \label{prop:f(z)_lim_0}$\lim_{z\to 0} f(z)=1$.
 \item \label{prop:f(z)_lim_1}$\lim_{z\to 1} f(z)=0$.
 \item \label{prop:f(z)_decreasing}$f(z)$ is strictly decreasing.
 \item $f(z)$ is concave.
 \item \label{prop:f(z)_z'}$f(z')>f(z_r )=\frac{1}{\hi-1}$.
 \item \label{prop:f(z)_phi}$f(z)\neq -\frac{1}{1-z'}-2+\hi+\lo$.
\item[]
\item[]
\item[]
\item[]
\end{compactenum}
\end{minipage}\hfill
\begin{minipage}{0.5\textwidth}
\centering\upshape
\scalebox{0.48}{
\begingroup
  \makeatletter
  \providecommand\color[2][]{%
    \GenericError{(gnuplot) \space\space\space\@spaces}{%
      Package color not loaded in conjunction with
      terminal option `colourtext'%
    }{See the gnuplot documentation for explanation.%
    }{Either use 'blacktext' in gnuplot or load the package
      color.sty in LaTeX.}%
    \renewcommand\color[2][]{}%
  }%
  \providecommand\includegraphics[2][]{%
    \GenericError{(gnuplot) \space\space\space\@spaces}{%
      Package graphicx or graphics not loaded%
    }{See the gnuplot documentation for explanation.%
    }{The gnuplot epslatex terminal needs graphicx.sty or graphics.sty.}%
    \renewcommand\includegraphics[2][]{}%
  }%
  \providecommand\rotatebox[2]{#2}%
  \@ifundefined{ifGPcolor}{%
    \newif\ifGPcolor
    \GPcolortrue
  }{}%
  \@ifundefined{ifGPblacktext}{%
    \newif\ifGPblacktext
    \GPblacktexttrue
  }{}%
  \let\gplgaddtomacro\g@addto@macro
  \gdef\gplbacktext{}%
  \gdef\gplfronttext{}%
  \makeatother
  \ifGPblacktext
    \def\colorrgb#1{}%
    \def\colorgray#1{}%
  \else
    \ifGPcolor
      \def\colorrgb#1{\color[rgb]{#1}}%
      \def\colorgray#1{\color[gray]{#1}}%
      \expandafter\def\csname LTw\endcsname{\color{white}}%
      \expandafter\def\csname LTb\endcsname{\color{black}}%
      \expandafter\def\csname LTa\endcsname{\color{black}}%
      \expandafter\def\csname LT0\endcsname{\color[rgb]{1,0,0}}%
      \expandafter\def\csname LT1\endcsname{\color[rgb]{0,1,0}}%
      \expandafter\def\csname LT2\endcsname{\color[rgb]{0,0,1}}%
      \expandafter\def\csname LT3\endcsname{\color[rgb]{1,0,1}}%
      \expandafter\def\csname LT4\endcsname{\color[rgb]{0,1,1}}%
      \expandafter\def\csname LT5\endcsname{\color[rgb]{1,1,0}}%
      \expandafter\def\csname LT6\endcsname{\color[rgb]{0,0,0}}%
      \expandafter\def\csname LT7\endcsname{\color[rgb]{1,0.3,0}}%
      \expandafter\def\csname LT8\endcsname{\color[rgb]{0.5,0.5,0.5}}%
    \else
      \def\colorrgb#1{\color{black}}%
      \def\colorgray#1{\color[gray]{#1}}%
      \expandafter\def\csname LTw\endcsname{\color{white}}%
      \expandafter\def\csname LTb\endcsname{\color{black}}%
      \expandafter\def\csname LTa\endcsname{\color{black}}%
      \expandafter\def\csname LT0\endcsname{\color{black}}%
      \expandafter\def\csname LT1\endcsname{\color{black}}%
      \expandafter\def\csname LT2\endcsname{\color{black}}%
      \expandafter\def\csname LT3\endcsname{\color{black}}%
      \expandafter\def\csname LT4\endcsname{\color{black}}%
      \expandafter\def\csname LT5\endcsname{\color{black}}%
      \expandafter\def\csname LT6\endcsname{\color{black}}%
      \expandafter\def\csname LT7\endcsname{\color{black}}%
      \expandafter\def\csname LT8\endcsname{\color{black}}%
    \fi
  \fi
  \setlength{\unitlength}{0.0500bp}%
  \begin{picture}(7200.00,5040.00)%
    \gplgaddtomacro\gplbacktext{%
      \csname LTb\endcsname%
      \put(858,978){\makebox(0,0)[r]{\strut{} }}%
      \put(858,1168){\makebox(0,0)[r]{\strut{} 0.1}}%
      \put(858,1569){\makebox(0,0)[r]{\strut{} 0.2}}%
      \put(858,1969){\makebox(0,0)[r]{\strut{} 0.3}}%
      \put(858,2370){\makebox(0,0)[r]{\strut{} 0.4}}%
      \put(858,2771){\makebox(0,0)[r]{\strut{} }}%
      \put(858,3172){\makebox(0,0)[r]{\strut{} 0.6}}%
      \put(858,3573){\makebox(0,0)[r]{\strut{} 0.7}}%
      \put(858,3973){\makebox(0,0)[r]{\strut{} 0.8}}%
      \put(858,4374){\makebox(0,0)[r]{\strut{} 0.9}}%
      \put(858,4775){\makebox(0,0)[r]{\strut{} 1}}%
      \put(1635,484){\makebox(0,0){\strut{} 0.1}}%
      \put(2216,484){\makebox(0,0){\strut{} 0.2}}%
      \put(2798,484){\makebox(0,0){\strut{} 0.3}}%
      \put(3379,484){\makebox(0,0){\strut{} 0.4}}%
      \put(3961,484){\makebox(0,0){\strut{} 0.5}}%
      \put(4543,484){\makebox(0,0){\strut{} 0.6}}%
      \put(5124,484){\makebox(0,0){\strut{} }}%
      \put(5213,484){\makebox(0,0){\strut{}$z_l$}}%
      \put(5706,484){\makebox(0,0){\strut{} 0.8}}%
      \put(6287,484){\makebox(0,0){\strut{} 0.9}}%
      \put(6801,484){\makebox(0,0){\strut{}$z_r$}}%
      \put(6869,484){\makebox(0,0){\strut{} }}%
      \put(3961,154){\makebox(0,0){\strut{}$z$}}%
      \put(588,2771){\makebox(0,0)[l]{\strut{}$\frac{1}{\lo-1}$}}%
      \put(588,978){\makebox(0,0)[l]{\strut{}$\frac{1}{\hi-1}$}}%
    }%
    \gplgaddtomacro\gplfronttext{%
      \csname LTb\endcsname%
      \put(5882,4602){\makebox(0,0)[r]{\strut{}$f(z)$}}%
      \csname LTb\endcsname%
      \put(5882,4382){\makebox(0,0)[r]{\strut{}$1-z$}}%
    }%
    \gplbacktext
    \put(0,0){\includegraphics{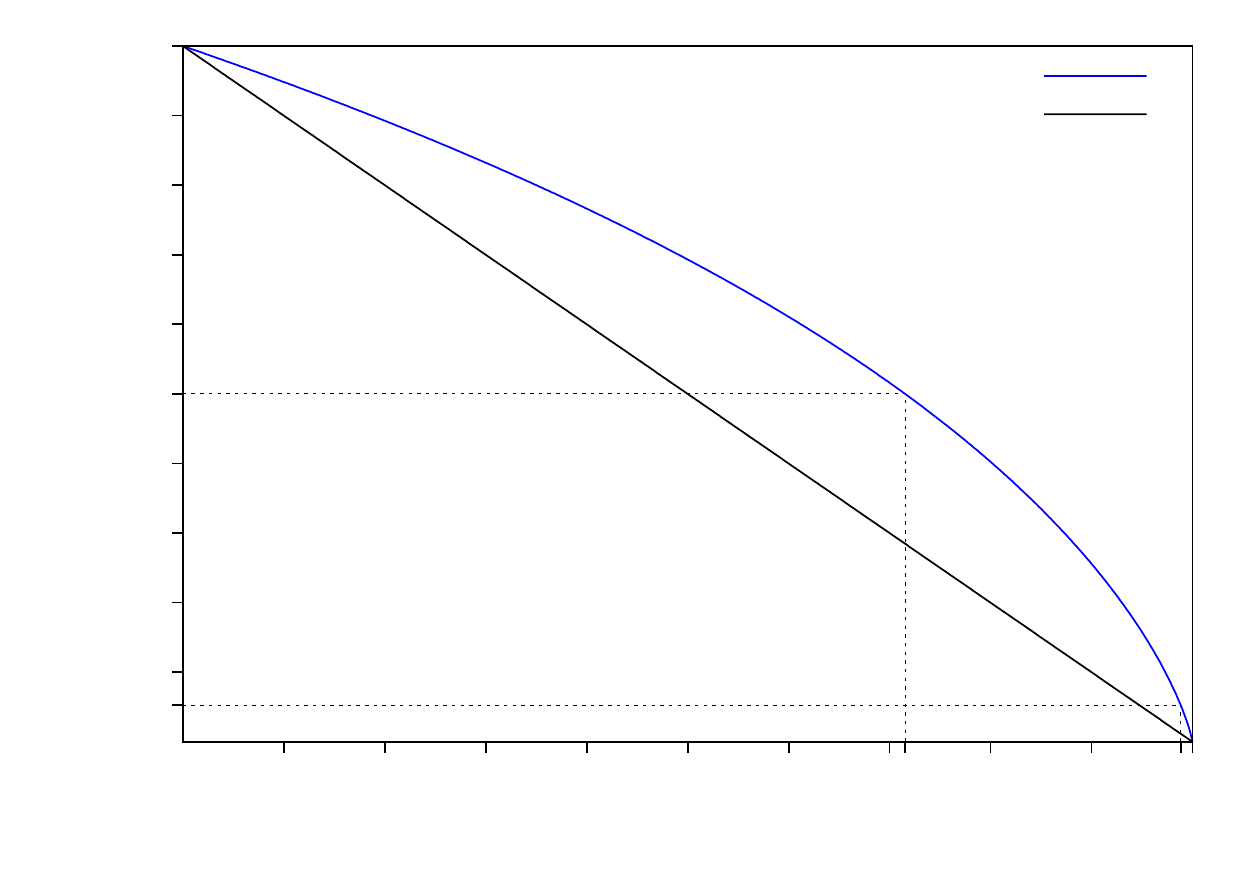}}%
    \gplfronttext
  \end{picture}%
\endgroup
}
\parbox{0.8\textwidth}{
\captionof{figure}{\label{fig:f(z)}Function $f(z)$ with $z_l$ and $z_r$ for $\lo=3, \hi=20$.}}
\end{minipage}
\end{lemma}
The proof of Lemma~\ref{lem:f(z)} is given in Appendix~\ref{app:f(z)}. A plot of $f(z)$ is shown in Figure~\ref{fig:f(z)}.

\begin{lemma}[Properties of $g(z)$]\\\hfill
\label{lem:g(z)}
Let $3 \leq \lo < \hi, \  z\in(0,1)$, then it holds
\begin{compactenum}[$(i)$]
\item \label{prop:g(z)_monotonicity} $g(z)$ is strictly decreasing, reaches a global minimum and is then strictly increasing.
The global minimum point is the only point where $\frac{\partial g(z)}{\partial z}=0$.
\item \label{prop:g(z)_left_interval}$g(z)>0$, $\forall z \in (0,z_l]$.
\item \label{prop:g(z)_right_interval}$g(z)>0$, $\forall z \in [z_r,1)$.
\item \label{prop:g(z)_zeros}If $\min g(z)<0$ then $g(z)$ has exactly two roots, say $z_1$ and $z_2$, with $z_1 <z_2$ and
$z_1,z_2\in (z_l,z_r)$.
\item \label{prop:g(z)_monotonicity_in_b}Let $z>z_l$ then it holds $g(z,\lo,\hi)>g(z,\lo,\hi+1)$.
\item \label{prop:g(z)_b'}For fixed  $\lo$ there is a threshold $\hi'$, $\hi' \geq \lo+1$, such for 
$\lo<\hi < \hi'$ it holds that $\min_z g(z,\hi)\geq 0$, and if
$\hi\geq \hi'$ then it holds $\min_z g(z,\hi)<0$. 
\end{compactenum}
\end{lemma}
The proof of Lemma~\ref{lem:g(z)} is given in Appendix~\ref{app:g(z)}. Example plots of $g(z,\lo,\hi)$ are shown in Figure~\ref{fig:g(z)}.
\begin{lemma}[Properties of $h(z)$]\hfill\\
\label{lem:h(z)}
 Let $3 \leq \lo < \hi, \  z\in(0,1)$, then it holds
\begin{compactenum}[$(i)$]
\item \label{prop:h(z)_pole}$h(z)$ has a pole at $z=z_r$.
\item $\lim_{z\to0} h(z)=-\infty$.
\item \label{prop:h(z)_limit_pole}$\lim_{z \to z_r} h(z)=+\infty$.
\item \label{prop:h(z)_first_interval} $\forall z \in (0,z_l] : h(z)\in (-\infty,1]$.
\item \label{prop:h(z)_middle_interval}$\forall z \in (z_l,z_r) : h(z)\in(1,+\infty)$, and $h(z_l)=1$.
\item \label{prop:h(z)_last_interval}$\forall z \in (z_r,1) : h(z)\in(-\infty,1)$.
\item \label{prop:h(z)_monotonicity}$\frac{\partial h(z)}{\partial z}>0 \Leftrightarrow g(z)>0$.
\item \label{prop:h(z)_increasing_first_interval}$h(z)$ is strictly increasing in $z\in (0,z_l]$. 
\item $h(z)$ is strictly increasing in $z\in (z_r,1)$. 
\item \label{prop:h(z)_monotonicity_middle_interval_strict_increasing}If $\min_z g(z)\geq 0$ then $h(z)$ is strictly increasing in $z\in [z_l,z_r)$. 
\item \label{prop:h(z)_monotonicity_middle_interval_inc_dec_inc} If $\min_z g(z)<0$ then $h(z)$ is strictly increasing to a local maximum at $z_1$,
then strictly decreasing to a local minimum at $z_2$, then strictly increasing afterwards.
\end{compactenum}
\end{lemma} 
The proof of Lemma~\ref{lem:h(z)} is given in Appendix~\ref{app:h(z)}. Example plots of $g(z,\lo,\hi)$ are shown in Figure~\ref{fig:h(z)}.

\begin{minipage}{0.5\textwidth}
 \centering
 \scalebox{0.48}{
\begingroup
  \makeatletter
  \providecommand\color[2][]{%
    \GenericError{(gnuplot) \space\space\space\@spaces}{%
      Package color not loaded in conjunction with
      terminal option `colourtext'%
    }{See the gnuplot documentation for explanation.%
    }{Either use 'blacktext' in gnuplot or load the package
      color.sty in LaTeX.}%
    \renewcommand\color[2][]{}%
  }%
  \providecommand\includegraphics[2][]{%
    \GenericError{(gnuplot) \space\space\space\@spaces}{%
      Package graphicx or graphics not loaded%
    }{See the gnuplot documentation for explanation.%
    }{The gnuplot epslatex terminal needs graphicx.sty or graphics.sty.}%
    \renewcommand\includegraphics[2][]{}%
  }%
  \providecommand\rotatebox[2]{#2}%
  \@ifundefined{ifGPcolor}{%
    \newif\ifGPcolor
    \GPcolortrue
  }{}%
  \@ifundefined{ifGPblacktext}{%
    \newif\ifGPblacktext
    \GPblacktexttrue
  }{}%
  \let\gplgaddtomacro\g@addto@macro
  \gdef\gplbacktext{}%
  \gdef\gplfronttext{}%
  \makeatother
  \ifGPblacktext
    \def\colorrgb#1{}%
    \def\colorgray#1{}%
  \else
    \ifGPcolor
      \def\colorrgb#1{\color[rgb]{#1}}%
      \def\colorgray#1{\color[gray]{#1}}%
      \expandafter\def\csname LTw\endcsname{\color{white}}%
      \expandafter\def\csname LTb\endcsname{\color{black}}%
      \expandafter\def\csname LTa\endcsname{\color{black}}%
      \expandafter\def\csname LT0\endcsname{\color[rgb]{1,0,0}}%
      \expandafter\def\csname LT1\endcsname{\color[rgb]{0,1,0}}%
      \expandafter\def\csname LT2\endcsname{\color[rgb]{0,0,1}}%
      \expandafter\def\csname LT3\endcsname{\color[rgb]{1,0,1}}%
      \expandafter\def\csname LT4\endcsname{\color[rgb]{0,1,1}}%
      \expandafter\def\csname LT5\endcsname{\color[rgb]{1,1,0}}%
      \expandafter\def\csname LT6\endcsname{\color[rgb]{0,0,0}}%
      \expandafter\def\csname LT7\endcsname{\color[rgb]{1,0.3,0}}%
      \expandafter\def\csname LT8\endcsname{\color[rgb]{0.5,0.5,0.5}}%
    \else
      \def\colorrgb#1{\color{black}}%
      \def\colorgray#1{\color[gray]{#1}}%
      \expandafter\def\csname LTw\endcsname{\color{white}}%
      \expandafter\def\csname LTb\endcsname{\color{black}}%
      \expandafter\def\csname LTa\endcsname{\color{black}}%
      \expandafter\def\csname LT0\endcsname{\color{black}}%
      \expandafter\def\csname LT1\endcsname{\color{black}}%
      \expandafter\def\csname LT2\endcsname{\color{black}}%
      \expandafter\def\csname LT3\endcsname{\color{black}}%
      \expandafter\def\csname LT4\endcsname{\color{black}}%
      \expandafter\def\csname LT5\endcsname{\color{black}}%
      \expandafter\def\csname LT6\endcsname{\color{black}}%
      \expandafter\def\csname LT7\endcsname{\color{black}}%
      \expandafter\def\csname LT8\endcsname{\color{black}}%
    \fi
  \fi
  \setlength{\unitlength}{0.0500bp}%
  \begin{picture}(7200.00,5040.00)%
    \gplgaddtomacro\gplbacktext{%
      \csname LTb\endcsname%
      \put(726,767){\makebox(0,0)[r]{\strut{}-10}}%
      \put(726,1340){\makebox(0,0)[r]{\strut{} 0}}%
      \put(726,1912){\makebox(0,0)[r]{\strut{} 10}}%
      \put(726,2485){\makebox(0,0)[r]{\strut{} 20}}%
      \put(726,3057){\makebox(0,0)[r]{\strut{} 30}}%
      \put(726,3630){\makebox(0,0)[r]{\strut{} 40}}%
      \put(726,4202){\makebox(0,0)[r]{\strut{} 50}}%
      \put(726,4775){\makebox(0,0)[r]{\strut{} 60}}%
      \put(921,484){\makebox(0,0){\strut{} 0}}%
      \put(1513,484){\makebox(0,0){\strut{} 0.1}}%
      \put(2106,484){\makebox(0,0){\strut{} 0.2}}%
      \put(2698,484){\makebox(0,0){\strut{} 0.3}}%
      \put(3291,484){\makebox(0,0){\strut{} 0.4}}%
      \put(3883,484){\makebox(0,0){\strut{} 0.5}}%
      \put(4475,484){\makebox(0,0){\strut{} 0.6}}%
      \put(5068,484){\makebox(0,0){\strut{} }}%
      \put(5159,484){\makebox(0,0){\strut{}$z_l$}}%
      \put(5340,484){\makebox(0,0){\strut{}$z_1$}}%
      \put(5660,484){\makebox(0,0){\strut{} 0.8}}%
      \put(6253,484){\makebox(0,0){\strut{} 0.9}}%
      \put(6583,484){\makebox(0,0){\strut{}$z_2$}}%
      \put(6804,484){\makebox(0,0){\strut{}$z_r$}}%
      \put(6845,484){\makebox(0,0){\strut{} 1}}%
      \put(3883,154){\makebox(0,0){\strut{}$z$}}%
    }%
    \gplgaddtomacro\gplfronttext{%
      \csname LTb\endcsname%
      \put(5858,4602){\makebox(0,0)[r]{\strut{}$g(z,3,10)$}}%
      \csname LTb\endcsname%
      \put(5858,4382){\makebox(0,0)[r]{\strut{}$g(z,3,20)$}}%
      \csname LTb\endcsname%
      \put(5858,4162){\makebox(0,0)[r]{\strut{}$g(z,3,30)$}}%
    }%
    \gplbacktext
    \put(0,0){\includegraphics{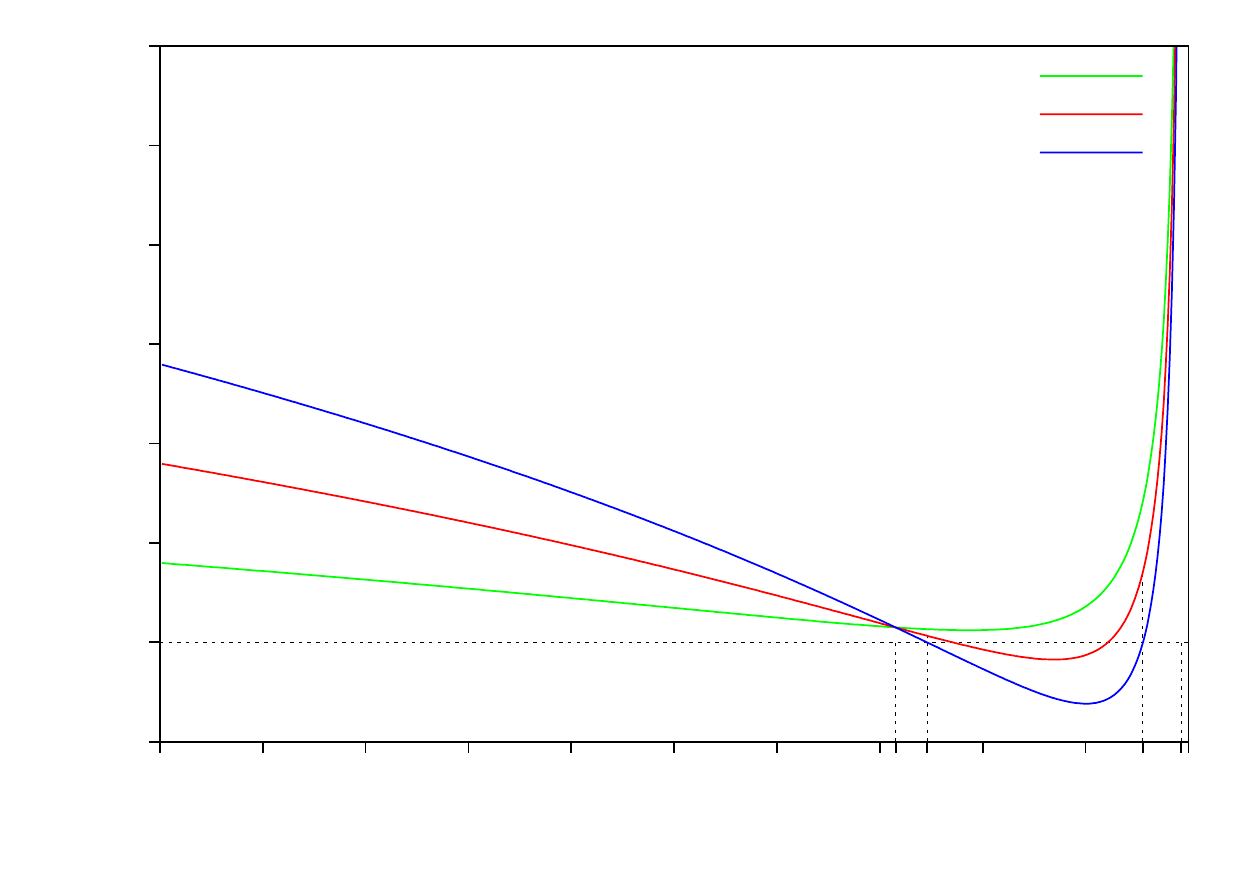}}%
    \gplfronttext
  \end{picture}%
\endgroup
}
 \parbox{0.8\textwidth}{
 \captionof{figure}{\label{fig:g(z)}Function $g(z)$, $z_1,z_2$ and $z_l, z_r$ for $\lo=3$ and $\hi=20$.}}
\end{minipage}\hfill
\begin{minipage}{0.5\textwidth}
 \centering
 \scalebox{0.48}{
\begingroup
  \makeatletter
  \providecommand\color[2][]{%
    \GenericError{(gnuplot) \space\space\space\@spaces}{%
      Package color not loaded in conjunction with
      terminal option `colourtext'%
    }{See the gnuplot documentation for explanation.%
    }{Either use 'blacktext' in gnuplot or load the package
      color.sty in LaTeX.}%
    \renewcommand\color[2][]{}%
  }%
  \providecommand\includegraphics[2][]{%
    \GenericError{(gnuplot) \space\space\space\@spaces}{%
      Package graphicx or graphics not loaded%
    }{See the gnuplot documentation for explanation.%
    }{The gnuplot epslatex terminal needs graphicx.sty or graphics.sty.}%
    \renewcommand\includegraphics[2][]{}%
  }%
  \providecommand\rotatebox[2]{#2}%
  \@ifundefined{ifGPcolor}{%
    \newif\ifGPcolor
    \GPcolortrue
  }{}%
  \@ifundefined{ifGPblacktext}{%
    \newif\ifGPblacktext
    \GPblacktexttrue
  }{}%
  \let\gplgaddtomacro\g@addto@macro
  \gdef\gplbacktext{}%
  \gdef\gplfronttext{}%
  \makeatother
  \ifGPblacktext
    \def\colorrgb#1{}%
    \def\colorgray#1{}%
  \else
    \ifGPcolor
      \def\colorrgb#1{\color[rgb]{#1}}%
      \def\colorgray#1{\color[gray]{#1}}%
      \expandafter\def\csname LTw\endcsname{\color{white}}%
      \expandafter\def\csname LTb\endcsname{\color{black}}%
      \expandafter\def\csname LTa\endcsname{\color{black}}%
      \expandafter\def\csname LT0\endcsname{\color[rgb]{1,0,0}}%
      \expandafter\def\csname LT1\endcsname{\color[rgb]{0,1,0}}%
      \expandafter\def\csname LT2\endcsname{\color[rgb]{0,0,1}}%
      \expandafter\def\csname LT3\endcsname{\color[rgb]{1,0,1}}%
      \expandafter\def\csname LT4\endcsname{\color[rgb]{0,1,1}}%
      \expandafter\def\csname LT5\endcsname{\color[rgb]{1,1,0}}%
      \expandafter\def\csname LT6\endcsname{\color[rgb]{0,0,0}}%
      \expandafter\def\csname LT7\endcsname{\color[rgb]{1,0.3,0}}%
      \expandafter\def\csname LT8\endcsname{\color[rgb]{0.5,0.5,0.5}}%
    \else
      \def\colorrgb#1{\color{black}}%
      \def\colorgray#1{\color[gray]{#1}}%
      \expandafter\def\csname LTw\endcsname{\color{white}}%
      \expandafter\def\csname LTb\endcsname{\color{black}}%
      \expandafter\def\csname LTa\endcsname{\color{black}}%
      \expandafter\def\csname LT0\endcsname{\color{black}}%
      \expandafter\def\csname LT1\endcsname{\color{black}}%
      \expandafter\def\csname LT2\endcsname{\color{black}}%
      \expandafter\def\csname LT3\endcsname{\color{black}}%
      \expandafter\def\csname LT4\endcsname{\color{black}}%
      \expandafter\def\csname LT5\endcsname{\color{black}}%
      \expandafter\def\csname LT6\endcsname{\color{black}}%
      \expandafter\def\csname LT7\endcsname{\color{black}}%
      \expandafter\def\csname LT8\endcsname{\color{black}}%
    \fi
  \fi
  \setlength{\unitlength}{0.0500bp}%
  \begin{picture}(7200.00,5040.00)%
    \gplgaddtomacro\gplbacktext{%
      \csname LTb\endcsname%
      \put(858,767){\makebox(0,0)[r]{\strut{}-0.5}}%
      \put(858,1212){\makebox(0,0)[r]{\strut{} 0}}%
      \put(858,1658){\makebox(0,0)[r]{\strut{} 0.5}}%
      \put(858,2103){\makebox(0,0)[r]{\strut{} 1}}%
      \put(858,2548){\makebox(0,0)[r]{\strut{} 1.5}}%
      \put(858,2994){\makebox(0,0)[r]{\strut{} 2}}%
      \put(858,3439){\makebox(0,0)[r]{\strut{} 2.5}}%
      \put(858,3884){\makebox(0,0)[r]{\strut{} 3}}%
      \put(858,4330){\makebox(0,0)[r]{\strut{} 3.5}}%
      \put(858,4775){\makebox(0,0)[r]{\strut{} 4}}%
      \put(1053,484){\makebox(0,0){\strut{} 0.65}}%
      \put(1884,484){\makebox(0,0){\strut{} 0.7}}%
      \put(2139,484){\makebox(0,0){\strut{}$z_l$}}%
      \put(2715,484){\makebox(0,0){\strut{} 0.75}}%
      \put(3066,484){\makebox(0,0){\strut{}$z_1$}}%
      \put(3546,484){\makebox(0,0){\strut{} 0.8}}%
      \put(4376,484){\makebox(0,0){\strut{} 0.85}}%
      \put(5207,484){\makebox(0,0){\strut{} 0.9}}%
      \put(5569,484){\makebox(0,0){\strut{}$z_2$}}%
      \put(6038,484){\makebox(0,0){\strut{} 0.95}}%
      \put(6675,484){\makebox(0,0){\strut{}$z_r$}}%
      \put(6869,484){\makebox(0,0){\strut{} 1}}%
      \put(3961,154){\makebox(0,0){\strut{}$z$}}%
    }%
    \gplgaddtomacro\gplfronttext{%
      \csname LTb\endcsname%
      \put(2637,4602){\makebox(0,0)[r]{\strut{}$h(z,3,10)$}}%
      \csname LTb\endcsname%
      \put(2637,4382){\makebox(0,0)[r]{\strut{}$h(z,3,20)$}}%
      \csname LTb\endcsname%
      \put(2637,4162){\makebox(0,0)[r]{\strut{}$h(z,3,30)$}}%
    }%
    \gplbacktext
    \put(0,0){\includegraphics{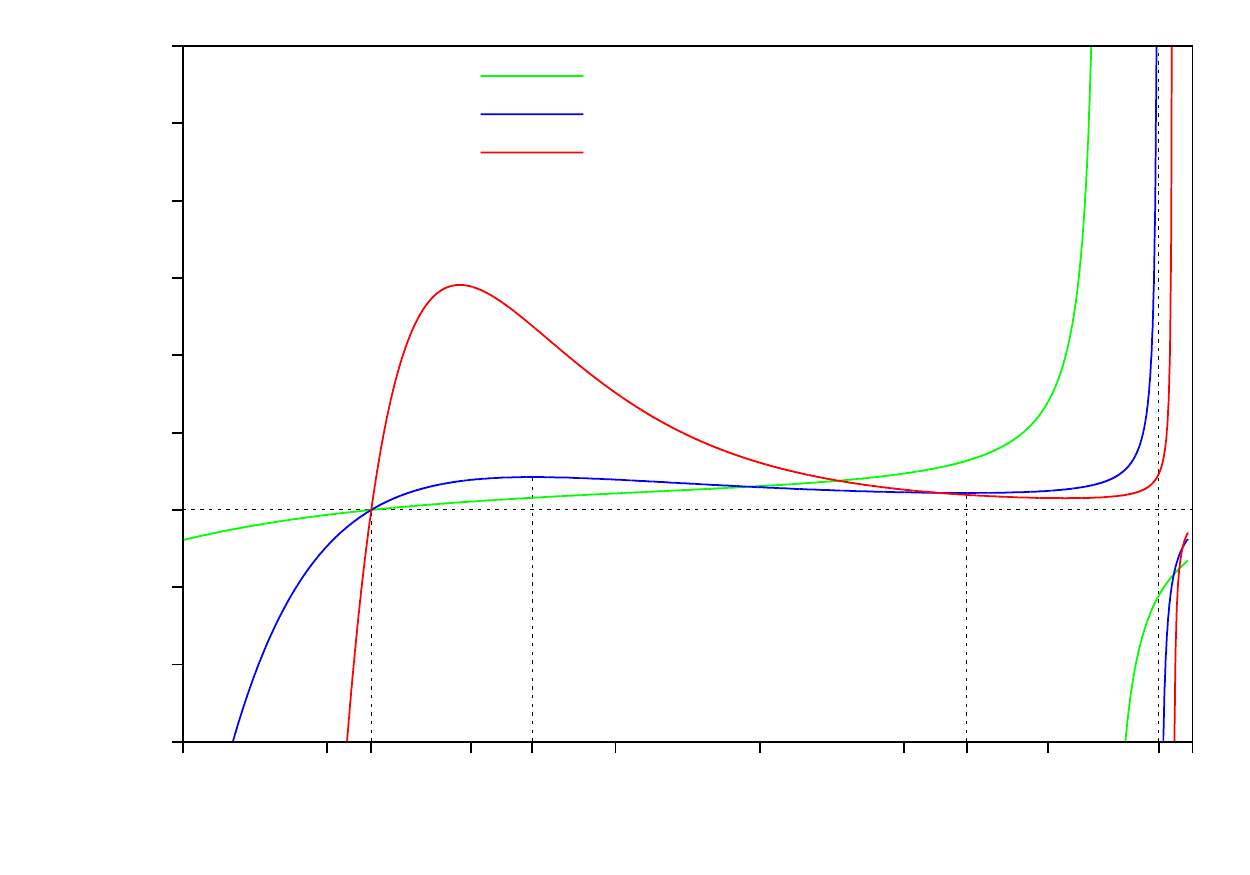}}%
    \gplfronttext
  \end{picture}%
\endgroup
}
 \parbox{0.8\textwidth}{
 \captionof{figure}{\label{fig:h(z)}Function $h(z)$, $z_1,z_2$ and $z_l, z_r$  for $\lo=3$ and $\hi=20$.}}
\end{minipage}
\vspace{1em}

\noindent Concerning the defined points, we are only interested in how they are related to each other.
\begin{lemma}
\label{lem:special_points}
Let $3 \leq \lo < \hi, \  z\in(0,1)$, then it holds
\begin{compactenum}[$(i)$]
\item $0<z_l<z_r<1$.
\item  $z_l<z_1,z_2<z_r$, if $z_1$ and $z_2$ exist.
\item  $z'\in(0,z_r)$.
\item \label{prop:z'_not_zero} $z'\neq z_1, z' \neq z_2$, if $z_1$ and $z_2$ exist.
\end{compactenum}
\end{lemma}
The proof Lemma~\ref{lem:special_points} is given in Appendix~\ref{app:points}.
Now we are ready for solving the optimization problem.

\subsection{Analysis}
Assume first that $\alpha$ is arbitrary but fixed, that is we are looking for a
global minimum of \eqref{eq:transformed_threshold_function} in $z$-direction.
Since 
\begin{equation}
\label{eq:lim_T(z)}
\lim_{z\to 0} T(z)=\lim_{z\to 1} T(z)=+\infty \ , 
\end{equation}
and $T(z)$ is continuous for $z\in(0,1)$, a global minimum must be a
point where the first derivative of $T(z)$ is zero, that is a critical point.
According to \eqref{eq:first_derivative_z} critical points in $z$-direction for \emph{unbounded} $\alpha$, i.e., $\alpha\in\mathbb{R}$,
can be described via
\begin{equation}
 \begin{split}
\label{eq:critical_h(z)}
                 \frac{\partial T(z)}{\partial z}=0 
 \Leftrightarrow & \frac{1}{1-z}\cdot D_0(z)=\frac{-\ln(1-z)}{z}\cdot D_1(z) \\
\Leftrightarrow  & \frac{D_0(z)}{D_1(z)}=f(z)
\Leftrightarrow   \alpha = 1/h(z) \ .  
 \end{split}
\end{equation}
The next lemma identifies and classifies critical points of $T(z)$ for \emph{bounded} $\alpha$ that is for $\alpha\in(0,1]$.
\begin{lemma}
\label{lem:critical_points_interval}
 Let $\alpha\in(0,1]$ be arbitrary but fixed.
 If $\frac{\partial T}{\partial z}(\tilde{z})=0$ for some $\tilde{z}\in(0,1)$ then it holds
\begin{compactenum}[$(i)$]
\item \label{prop:critical_points_interval}$\tilde{z}\in [z_l,z_r)$,
\item \label{prop:critical_points_minima}if $g(\tilde{z})>0$ then $T(\tilde{z})$ is a local minimum,
\item \label{prop:critical_points_maxima}if $g(\tilde{z})<0$ then $T(\tilde{z})$ is a local maximum.
\end{compactenum}
\end{lemma}
\begin{proof}
\begin{asparaenum}[$(i)$]
 \item According to \eqref{eq:critical_h(z)} we must have $\alpha= 1/h(\tilde{z})$ for $\alpha\in(0,1]$.
 Therefore it must hold $h(\tilde{z})\in(1,+\infty)$. Using Lemma~\ref{lem:h(z)}$(\ref{prop:h(z)_first_interval}),(\ref{prop:h(z)_middle_interval}),(\ref{prop:h(z)_last_interval})$ 
it follows that $\tilde{z}\in [z_l,z_r)$.
\item Now consider the second derivative of $T(z)$ with respect to $z$. According to \eqref{eq:second_derivative_zz} we have
\begin{align*}
  \frac{\partial^2 T(z)}{(\partial z)^2} > 0 \Leftrightarrow 
\frac{1}{(1-z)^2} -\frac{2}{z\cdot (1-z)}\cdot \frac{D_1(z)}{D_0(z)}  \\
 + \frac{\ln(1-z)}{z^2}\cdot \frac{D_2(z)-D_1(z)}{D_0(z)} -\frac{2\cdot \ln(1-z)}{z^2}\cdot \frac{D_1(z)^2}{D_0(z)^2}>0  \ .
\end{align*}
Assume that $\tilde{z}\in[z_l,z_r)$ is a critical point.
For the rest of the proof let $z=\tilde{z}$.
Utilizing that $\frac{D_0({z})}{D_1({z})}=f({z})$ it follows that
\begin{align*}
  \frac{\partial^2 T(z)}{(\partial z)^2} > 0 \Leftrightarrow &\frac{1}{(1-z)^2} -\frac{2}{z\cdot (1-z)\cdot f(z)}
 + \frac{\ln(1-z)}{z^2}\cdot \frac{D_2(z)}{D_0(z)}\\& -\frac{\ln(1-z)}{z^2\cdot f(z)} -\frac{2\cdot \ln(1-z)}{z^2\cdot f(z)^2}>0 \\
\Leftrightarrow&
\frac{1}{(1-z)^2}-\frac{f(z)}{z\cdot (1-z)}\cdot \frac{D_2(z)}{D_0(z)} +\frac{1}{z(1-z)}>0 \\
\Leftrightarrow&
\frac{D_0(z)}{D_2(z)}>f(z)\cdot (1-z)
\Leftrightarrow
\frac{D_1(z)}{D_2(z)}> (1-z) \ .
\end{align*}

Factoring out $\alpha$ from $\frac{D_1(z)}{D_2(z)}> (1-z)$ gives that $D_1(z)> (1-z)\cdot D_2(z)$ is equivalent to
\begin{equation*}
\alpha \cdot \left( Z_1(z)-(1-z)\cdot Z_2(z) \right) > -\hi \cdot (\hi-1)\cdot z^{\hi-1} +(1-z)\cdot \hi \cdot ( \hi-1)^2 \cdot z^{\hi-1} \ . 
\end{equation*}
According to the proof of Lemma~\ref{lem:critical_points_interval}$(\ref{prop:critical_points_interval})$ it holds
$\alpha=1/h(z)$, which can be written as
\begin{equation}
\label{eq:alpha}
 \alpha=\frac{\hi\cdot z^{\hi-1}\cdot((\hi-1)\cdot f(z)-1)}{Z_0(z)-f(z)\cdot Z_1(z)} \tag{$\star$}\ .
\end{equation}
Division by $\hi\cdot z^{\hi-1}$ leads to
\begin{align*}
 \frac{\partial^2 T(z)}{(\partial z)^2} > 0 \Leftrightarrow
\frac{(\hi-1)\cdot f(z)-1}{Z_0(z)-f(z)\cdot Z_1(z)}
 \cdot \left( Z_1(z)-(1-z)\cdot Z_2(z) \right) >\\ 
-(\hi-1) +(1-z)\cdot ( \hi-1)^2 \ .
\end{align*}
Consider \eqref{eq:alpha}. Given that $\alpha\in(0,1]$ and $z<z_r$ (Lemma~\ref{lem:critical_points_interval}$(\ref{prop:critical_points_interval})$)
we have, according to the definition of $z_r$ and Lemma~\ref{lem:f(z)}$(\ref{prop:f(z)_decreasing})$, that $(\hi-1)\cdot f(z)>1$,
that is the numerator of \eqref{eq:alpha} is larger than $0$. Since $\alpha>0$ it follows that
the denominator $Z_0(z)-f(z)\cdot Z_1(z)$ is larger than $0$ too.
Hence we get
\begin{align*}
     &\frac{\partial^2 T(z)}{(\partial z)^2} > 0 \\
\Leftrightarrow & ( (\hi-1)\cdot f(z)-1)\cdot (Z_1(z)-(1-z)\cdot Z_2(z)) >\\
& ((\hi-1)^2\cdot (1-z) -(\hi-1))\cdot (Z_0(z)-f(z)\cdot Z_1(z)) \\
\Leftrightarrow & 
Z_2(z) \cdot (1-z -(1-z)\cdot (\hi-1)\cdot f(z))\\ + &Z_1(z)\cdot ((\hi-1)^2\cdot (1-z)\cdot f(z)-1) >
 Z_0(z)\cdot ((1-z)\cdot (\hi-1)^2-(\hi-1))\\
\Leftrightarrow &
(\lo-1)^2\cdot (1-z-(1-z)\cdot(\hi-1)\cdot f(z))\\ +&(\lo-1)\cdot ((\hi-1)^2\cdot (1-z)\cdot f(z)-1)>
(1-z)\cdot (\hi-1)^2-(\hi-1) \ .
\end{align*}
Factoring out $f(z)\cdot (\hi-1)\cdot (\lo-1)$ gives
\begin{align*}
  \frac{\partial^2 T(z)}{(\partial z)^2} > 0
\Leftrightarrow & f(z)\cdot (\hi-1)\cdot (\lo-1) \cdot (\hi-\lo) > (\hi-1)^2 -(\lo-1)^2- \frac{\hi-\lo}{1-z} \\
\Leftrightarrow & f(z)\cdot (\hi-1)\cdot (\lo-1) + \frac{1}{1-z} + 2 - \hi - \lo > 0 \Leftrightarrow g(z)>0 \ . 
\end{align*}
\item Analogous to $(ii)$.
\end{asparaenum}
This finishes the proof of the lemma.
\end{proof}
The next lemma can be seen as the central building block for understanding the behavior of the threshold function.
Using the function $g(z)$ we decide how many and which kind of extremal points $T(z)$ has.

\begin{lemma}
\label{lem:classifying_extrema}
Let $\alpha\in(0,1]$ be arbitrary but fixed.
\begin{compactenum}[$1.$]
 \item\label{prop:classify_fst} Let $\min_z g(z)\geq 0$ then the function $T(z)$ has exactly one critical point $\tilde{z}$,
       and $\tilde{z}\in[z_l,z_r)$ is a global minimum point.
 \item\label{prop:classify_snd} Let $\min_z g(z)<0$ then there are four pairwise distinct points $z_1^{<}<z_1<z_2<z_2^{>}$
       from the interval $[z_l,\bar{z_r})$ such that the following holds:
\begin{compactenum}[$(i)$]
 \item  For all $\alpha$ with $1/\alpha \in [1,h(z_2))$ the function $T(z)$ has exactly one critical point $\tilde{z}$,
        and $\tilde{z}\in(z_l,z_1^{<})$, is a global minimum point.
 \item  For $\alpha$ with $1/\alpha=h(z_2)$ the function $T(z)$ has exactly two critical points $\tilde{z}_1<\tilde{z}_2$, 
        and $\tilde{z}_1=z_1^{<}$ is a global minimum point, and $\tilde{z}_2=z_2$  is an inflection point.
 \item \label{prop:classify_snd_three} For all $\alpha$ with $1/\alpha\in(h(z_2),h(z_1))$ the function $T(z)$ has exactly three critical points $\tilde{z}_1<\tilde{z}_3<\tilde{z}_2$,
        and $\tilde{z}_1,\tilde{z}_2$ are local minimum points and $\tilde{z}_3$ is a local maximum point.
 \item  For $\alpha$ with $1/\alpha=h(z_1)$ the function $T(z)$ has exactly two critical points $\tilde{z}_1<\tilde{z}_2$,
        and $\tilde{z}_1=z_1$ is an inflection point, and $\tilde{z}_2=z_2^{>}$  is an global minimum point.
 \item For all $\alpha$ with $1/\alpha \in (h(z_1),\infty)$ the function $T(z)$ has exactly one critical point $\tilde{z}$,
        and $\tilde{z}\in(z_2^{>},z_r)$, is a global minimum point.
\end{compactenum}
\end{compactenum}
\end{lemma}
Figure~\ref{fig:h(z)_extrema} illustrates the complete case 2 of Lemma~\ref{lem:classifying_extrema}.
The intersection points between the function $1/\alpha$ (horizontal lines) and the function $h(z)$ are the
extrema of $T(z)$. They are classified depending on the part of $h(z)$
where the intersection takes place.

\begin{minipage}{0.95\textwidth}
\upshape
\centering
\scalebox{0.8}{
\begingroup
  \makeatletter
  \providecommand\color[2][]{%
    \GenericError{(gnuplot) \space\space\space\@spaces}{%
      Package color not loaded in conjunction with
      terminal option `colourtext'%
    }{See the gnuplot documentation for explanation.%
    }{Either use 'blacktext' in gnuplot or load the package
      color.sty in LaTeX.}%
    \renewcommand\color[2][]{}%
  }%
  \providecommand\includegraphics[2][]{%
    \GenericError{(gnuplot) \space\space\space\@spaces}{%
      Package graphicx or graphics not loaded%
    }{See the gnuplot documentation for explanation.%
    }{The gnuplot epslatex terminal needs graphicx.sty or graphics.sty.}%
    \renewcommand\includegraphics[2][]{}%
  }%
  \providecommand\rotatebox[2]{#2}%
  \@ifundefined{ifGPcolor}{%
    \newif\ifGPcolor
    \GPcolortrue
  }{}%
  \@ifundefined{ifGPblacktext}{%
    \newif\ifGPblacktext
    \GPblacktexttrue
  }{}%
  \let\gplgaddtomacro\g@addto@macro
  \gdef\gplbacktext{}%
  \gdef\gplfronttext{}%
  \makeatother
  \ifGPblacktext
    \def\colorrgb#1{}%
    \def\colorgray#1{}%
  \else
    \ifGPcolor
      \def\colorrgb#1{\color[rgb]{#1}}%
      \def\colorgray#1{\color[gray]{#1}}%
      \expandafter\def\csname LTw\endcsname{\color{white}}%
      \expandafter\def\csname LTb\endcsname{\color{black}}%
      \expandafter\def\csname LTa\endcsname{\color{black}}%
      \expandafter\def\csname LT0\endcsname{\color[rgb]{1,0,0}}%
      \expandafter\def\csname LT1\endcsname{\color[rgb]{0,1,0}}%
      \expandafter\def\csname LT2\endcsname{\color[rgb]{0,0,1}}%
      \expandafter\def\csname LT3\endcsname{\color[rgb]{1,0,1}}%
      \expandafter\def\csname LT4\endcsname{\color[rgb]{0,1,1}}%
      \expandafter\def\csname LT5\endcsname{\color[rgb]{1,1,0}}%
      \expandafter\def\csname LT6\endcsname{\color[rgb]{0,0,0}}%
      \expandafter\def\csname LT7\endcsname{\color[rgb]{1,0.3,0}}%
      \expandafter\def\csname LT8\endcsname{\color[rgb]{0.5,0.5,0.5}}%
    \else
      \def\colorrgb#1{\color{black}}%
      \def\colorgray#1{\color[gray]{#1}}%
      \expandafter\def\csname LTw\endcsname{\color{white}}%
      \expandafter\def\csname LTb\endcsname{\color{black}}%
      \expandafter\def\csname LTa\endcsname{\color{black}}%
      \expandafter\def\csname LT0\endcsname{\color{black}}%
      \expandafter\def\csname LT1\endcsname{\color{black}}%
      \expandafter\def\csname LT2\endcsname{\color{black}}%
      \expandafter\def\csname LT3\endcsname{\color{black}}%
      \expandafter\def\csname LT4\endcsname{\color{black}}%
      \expandafter\def\csname LT5\endcsname{\color{black}}%
      \expandafter\def\csname LT6\endcsname{\color{black}}%
      \expandafter\def\csname LT7\endcsname{\color{black}}%
      \expandafter\def\csname LT8\endcsname{\color{black}}%
    \fi
  \fi
  \setlength{\unitlength}{0.0500bp}%
  \begin{picture}(7200.00,5040.00)%
    \gplgaddtomacro\gplbacktext{%
      \csname LTb\endcsname%
      \put(1078,767){\makebox(0,0)[r]{\strut{} 0.9}}%
      \put(1078,1435){\makebox(0,0)[r]{\strut{} 1}}%
      \put(1078,2103){\makebox(0,0)[r]{\strut{} 1.1}}%
      \put(1078,2771){\makebox(0,0)[r]{\strut{} 1.2}}%
      \put(1078,3439){\makebox(0,0)[r]{\strut{} 1.3}}%
      \put(1078,4107){\makebox(0,0)[r]{\strut{} 1.4}}%
      \put(1078,4775){\makebox(0,0)[r]{\strut{} 1.5}}%
      \put(1273,484){\makebox(0,0){\strut{} 0.7}}%
      \put(1567,484){\makebox(0,0){\strut{}$z_l$}}%
      \put(1812,484){\makebox(0,0){\strut{}$z_1^{<}$}}%
      \put(2232,484){\makebox(0,0){\strut{} 0.75}}%
      \put(2638,484){\makebox(0,0){\strut{}$z_1$}}%
      \put(3191,484){\makebox(0,0){\strut{} 0.8}}%
      \put(4151,484){\makebox(0,0){\strut{} 0.85}}%
      \put(5110,484){\makebox(0,0){\strut{} 0.9}}%
      \put(5527,484){\makebox(0,0){\strut{}$z_2$}}%
      \put(6069,484){\makebox(0,0){\strut{} 0.95}}%
      \put(6512,484){\makebox(0,0){\strut{}$z_2^{>}$}}%
      \put(6804,484){\makebox(0,0){\strut{}$z_r$}}%
      \put(308,2771){\rotatebox{-270}{\makebox(0,0){\strut{}$1/\alpha$}}}%
      \put(4038,154){\makebox(0,0){\strut{}$z$}}%
      \put(3191,1569){\makebox(0,0)[l]{\strut{}$(i)$}}%
      \put(3191,1903){\makebox(0,0)[l]{\strut{}$(ii)$}}%
      \put(5110,2637){\makebox(0,0)[l]{\strut{}$(iii)$}}%
      \put(3191,3105){\makebox(0,0)[l]{\strut{}$(iv)$}}%
      \put(3191,3773){\makebox(0,0)[l]{\strut{}$(v)$}}%
    }%
    \gplgaddtomacro\gplfronttext{%
      \csname LTb\endcsname%
      \put(4045,4602){\makebox(0,0)[r]{\strut{}$h(z)$ with $g(z)>0$}}%
      \csname LTb\endcsname%
      \put(4045,4382){\makebox(0,0)[r]{\strut{}$h(z)$ with $g(z)<0$}}%
    }%
    \gplbacktext
    \put(0,0){\includegraphics{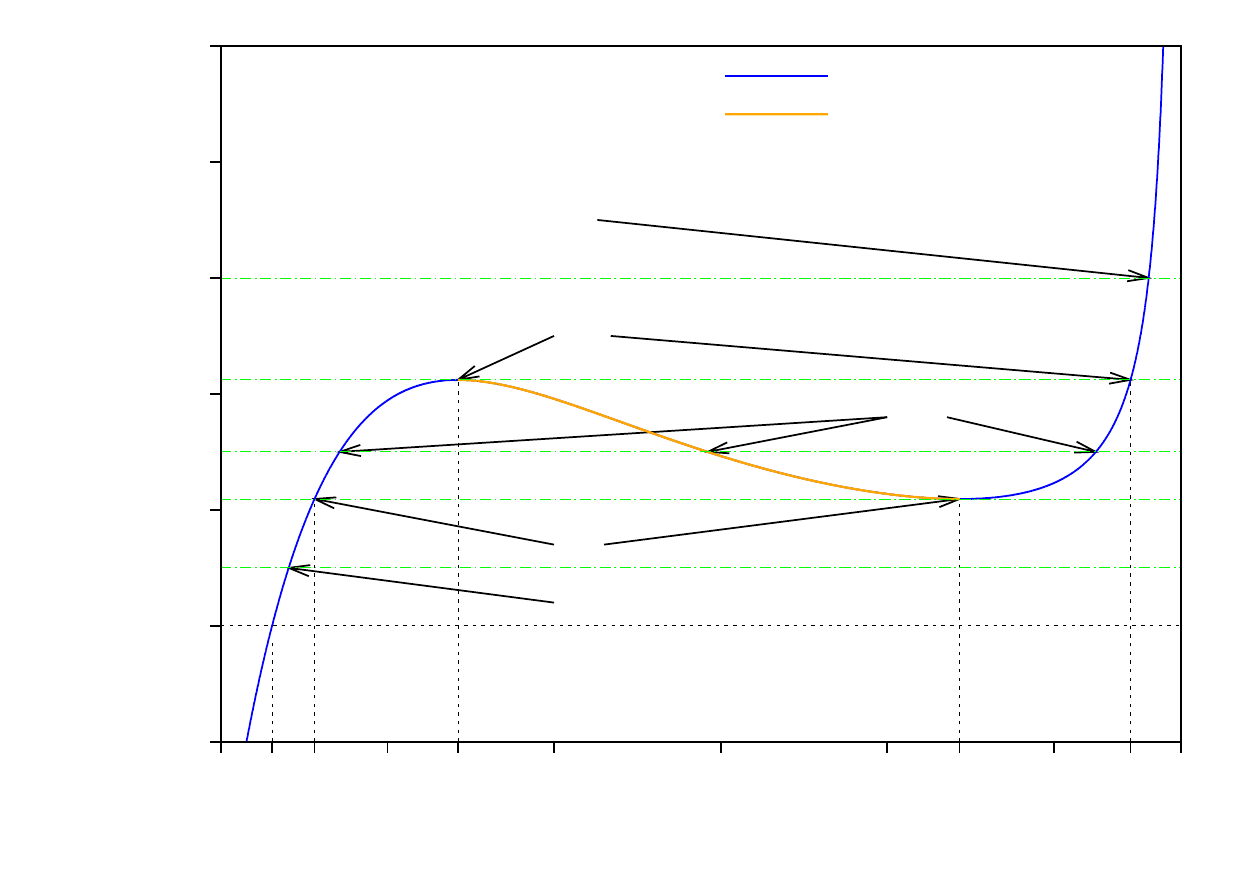}}%
    \gplfronttext
  \end{picture}%
\endgroup
}
\parbox{0.8\textwidth}{
\captionof{figure}{\label{fig:h(z)_extrema}$h(z)$ for $\lo=3, \hi=20$, $\min_z g(z)<0$.}}
\end{minipage}

\begin{proof}
\begin{compactenum}[$1.$]
\item From Lemma~\ref{lem:critical_points_interval}$(\ref{prop:critical_points_interval})$ it follows that all critical points $\tilde{z}$ must be from $[z_l,z_r)$.
Consider the function $h(z)$. According to Lemma~\ref{lem:h(z)}$(\ref{prop:h(z)_middle_interval}),(\ref{prop:h(z)_monotonicity_middle_interval_strict_increasing})$
it holds that for each $x$ from $[1,+\infty)$ there is exactly one $z$ from $[z_l,z_r)$ such that $h(z)=x$.
Furthermore, according to \eqref{eq:critical_h(z)} we have $\frac{\partial T(z)}{\partial z}=0 \Leftrightarrow \alpha= 1/h(z)$.
It follows that for each $\alpha\in(0,1]$ there is exactly one $\tilde{z}$ that is a critical point, that is it holds
$\alpha=1/h(\tilde{z})$. Let $\min_z g(z) \geq 0$ then it must hold $g(\tilde{z})\geq 0$ as well.
Since $\tilde{z}$ is the only critical point it follows with \eqref{eq:lim_T(z)} that it must be a global minimum point.
%
\item From Lemma~\ref{lem:h(z)}$(\ref{prop:h(z)_monotonicity_middle_interval_inc_dec_inc})$
we know that for $z\in [z_l,z_r)$ the function $h(z)$ is strictly increasing, reaches a local 
maximum at $z_1$, is strictly decreasing, reaches a local minimum at $z_2$ and is strictly increasing to $+\infty$ afterwards. 
Furthermore it holds $g(z)=0$ for $z\in\{z_1,z_2\}$ (definition of $z_1,z_2$),  $g(z)>0$ for $z<z_1$ and $z>z_2$, as well as $g(z)<0$ for $z\in(z_1,z_2)$ (Lemma~\ref{lem:h(z)}$(\ref{prop:h(z)_monotonicity})$).

Consider the condition \eqref{eq:critical_h(z)}.
\begin{compactenum}[$(i)$]
\item For all $\alpha$ with $1/\alpha\in [1,h(z_2))$ there is, according to Lemma~\ref{lem:h(z)},
      exactly one $z$ with $1/\alpha=h(z)$. In addition we have that $z<z_1$.
      Utilizing that $g(z)>0$, for $z<z_1$,(Lemma~\ref{lem:g(z)}$(\ref{prop:g(z)_monotonicity}),(\ref{prop:g(z)_left_interval})$)
      the claim follows by Lemma~\ref{lem:critical_points_interval}$(\ref{prop:critical_points_minima})$.
\item Let $1/\alpha=h(z_2)$ and let $\tilde{z}_2=z_2$ then according to Lemma~\ref{lem:h(z)} there is exactly one other point $\tilde{z}_1$, such that
      $\alpha=1/h(\tilde{z}_1)$. Furthermore it holds $g(\tilde{z}_1)>0$ and $g(\tilde{z}_2)=0$.
      According to Lemma~\ref{lem:h(z)}$(\ref{prop:h(z)_monotonicity})$ $\tilde{z}_1$ must be a local minimum point.
      Because of the monotonicity of $T(z)$ \eqref{eq:lim_T(z)} the other critical point must be an inflection point.
      Hence $\tilde{z}_1$ is also a global minimum point.
\item According to Lemma~\ref{lem:h(z)} there are exactly $3$ different points $\tilde{z}_i$, $1\leq i\leq 3$, such that 
      $1/\alpha=h(\tilde{z}_i)$ and $\frac{\partial T}{\partial z}(\tilde{z_i})=0$, respectively.
      Furthermore it holds $\tilde{z}_1< z_1< \tilde{z}_3< z_2<\tilde{z}_2$ and $g(\tilde{z}_1)>0, g(\tilde{z}_2)>0, g(\tilde{z}_3)<0$.
      From Lemma \ref{lem:critical_points_interval}$(\ref{prop:critical_points_minima}),(\ref{prop:critical_points_maxima})$ it follows that $\tilde{z}_1$ and $\tilde{z}_2$ are local minimum points of $T(z)$ and
      $\tilde{z}_3$ is a local maximum point of $T(z)$.
\item The case  $1/\alpha=1/h(z_1)$ is analogous to the case $(ii)$. 
\item The case  $1/\alpha \in (h(z_1),\infty)$ is analogous to the case $(i)$.
\end{compactenum}
\end{compactenum}
This finishes the proof of the lemma.
\end{proof}
The last lemma gives a complete characterization of the local extrema of \eqref{eq:transformed_threshold_function}
in $z$-direction including the global minimum for arbitrary but fixed $\alpha$.
It remains to find a value $\alpha^*$ that maximizes the threshold function at the corresponding global minimum in $z$-direction.
So the point we are looking for could be a saddle point of $T(z,\alpha)$. Indeed the following lemma shows
that $T(z,\alpha)$ has exactly one saddle point for \emph{unbounded} $\alpha$, i.e. $\alpha\in\mathbb{R}$,
and Theorem~\ref{theo:main} finally shows under which conditions 
this point is the optimum we are looking for.
\begin{lemma}
\label{lem:saddle}
Let $\alpha \in \mathbb{R}$. Then $T(z,\alpha)$ has exactly one saddle point 
\begin{displaymath}
(\tilde{z},\tilde{\alpha}) =\Big(\left(\tfrac{\lo}{\hi}\right)^{\tfrac{1}{\hi-\lo}} ,\tfrac{\hi-1}{\hi-\lo}-\tfrac{1}{f(\tilde{z})\cdot (\hi-\lo)}\Big). 
\end{displaymath}
\end{lemma}
\begin{proof}
Solving the linear system $\{ \frac{\partial T(z,\alpha)}{\partial z}=0 , \frac{\partial T(z,\alpha)}{\partial \alpha}=0\}$
gives 
\begin{align*}
\frac{\partial T(z,\alpha)}{\partial z}=0 &\Leftrightarrow \alpha= 1/h(z)\\
\frac{\partial T(z,\alpha)}{\partial \alpha}= 0& \Leftrightarrow   \frac{\ln(1-z)\cdot Z_0(z)}{D_0(z,\alpha)^2}=0 \Leftrightarrow Z_0(z)=0\\
  &\Leftrightarrow 
\lo\cdot z^{\lo-1}=\hi\cdot z^{\hi-1}\Leftrightarrow z=\left(\frac{\lo}{\hi}\right)^{\frac{1}{\hi-\lo}} \ .
\end{align*}
There is only one solution of $\frac{\partial T(z,\alpha)}{\partial \alpha}= 0$ and 
according to Lemma~\ref{lem:h(z)}$(\ref{prop:h(z)_pole})$ and Lemma~\ref{lem:f(z)}$(\ref{prop:f(z)_z'})$
$h(z)$ is defined at $z'=\left(\frac{\lo}{\hi}\right)^{\frac{1}{\hi-\lo}}$. Hence we get a unique critical point
$(\tilde{z},\tilde{\alpha})$ where $\tilde{z}=z'$ and 
\begin{align*}
\tilde{\alpha}=1/h(\tilde{z})&= \frac{\hi\cdot \left(\frac{\lo}{\hi}\right)^{\frac{\hi-1}{\hi-\lo}} \cdot ( f(\tilde{z})\cdot (\hi-1)-1) }
{-f(\tilde{z})\cdot \left(  \lo^2\cdot  \left(\frac{\lo}{\hi}\right)^{\frac{\lo-1}{\hi-\lo}} -\hi^2\cdot \left(\frac{\lo}{\hi}\right)^{\frac{\hi-1}{\hi-\lo}} \right)}\\
&= 
\frac{\hi\cdot \left(\frac{\lo}{\hi}\right)^{\frac{\hi-1}{\hi-\lo}} \cdot ( f(\tilde{z})\cdot (\hi-1)-1) }
{-f(\tilde{z})\cdot \hi^2 \cdot \left(\frac{\lo}{\hi}\right)^{\frac{\hi-1}{\hi-\lo}} \cdot (\frac{\lo}{\hi}-1)}\\
&=\frac{f(\tilde{z})\cdot(\hi-1)-1}{f(\tilde{z})\cdot (\hi-\lo)}=\frac{\hi-1}{\hi-\lo}-\frac{1}{f(\tilde{z})\cdot (\hi-\lo)} \ .
\end{align*}

To classify this critical point we consider the second partial derivatives of $T(z,\alpha)$.
We have $Z_0(\tilde{z})=0$ and $Z_1(\tilde{z})>0$, since
\begin{align*}
Z_1(\tilde{z})=&\lo\cdot (\lo-1)\cdot  \left(\frac{\lo}{\hi}\right)^{(\lo-1)/(\hi-\lo)}-\hi\cdot (\hi-1)\cdot \left(\frac{\lo}{\hi}\right)^{(\hi-1)/(\hi-\lo)}<0\\
&\Leftrightarrow \frac{\lo\cdot (\lo-1)}{\hi\cdot (\hi-1)}<\left( \frac{\lo}{\hi}\right)^{(\hi-\lo)/(\hi-\lo)} \Leftrightarrow \frac{\lo-1}{\hi-1}<1\checkmark
\end{align*}
It follows that $\frac{\partial^2 }{(\partial \alpha)^2} T(\tilde{z},\tilde{\alpha})=0$ as well as $\frac{\partial^2}{\partial z \partial \alpha} T(\tilde{z},\tilde{\alpha})>0$.
Therefore the Hessian matrix $H$ with
\begin{equation*}
H= \begin{pmatrix}
\frac{\partial^2 }{(\partial z)^2} T(\tilde{z},\tilde{\alpha})            & \frac{\partial^2}{\partial z \partial \alpha} T(\tilde{z},\tilde{\alpha})\\
\frac{\partial^2}{\partial z \partial \alpha} T(\tilde{z},\tilde{\alpha}) & \frac{\partial^2 }{(\partial \alpha)^2} T(\tilde{z},\tilde{\alpha}) 
   \end{pmatrix} \hat{=}
\begin{pmatrix}
=0 & >0 \\>0 &\frac{\partial^2 }{(\partial z)^2} T(\tilde{z},\tilde{\alpha})
\end{pmatrix}
\end{equation*}
has determinant $\det(H)< 0$, that is $(\tilde{z},\tilde{\alpha})$ is a saddle point.
\end{proof}
\subsection{Putting It All Together}
Now we prove Theorem~\ref{theo:main}.
\begin{proof}
Using \eqref{eq:critical_h(z)} we can define
a function of critical points $\tilde{T}(z)$ of $T(z,\alpha)$ as follows
\begin{align*}
 \tilde{T}(z):=&T(z,1/h(z))=\frac{-\ln(1-z)}{1/h(z)\cdot Z_0(z)+\hi\cdot z^{\hi-1}}\\
              =&\frac{\frac{z}{1-z}\cdot Z_0(z)+\ln(1-z)\cdot Z_1(z)}{\hi\cdot\lo\cdot(\hi-\lo)\cdot z^{\hi+\lo-2}} \ .
\end{align*}
The first derivative of $\tilde{T}(z)$ is
\begin{align*}
\frac{\partial \tilde{T}(z)}{\partial z} =& \frac{1}{\hi\cdot\lo\cdot(\hi-\lo)\cdot z^{\hi+\lo-2}}\cdot\left( \frac{Z_0(z)}{(1-z)^2}+\frac{\ln(1-z)\cdot Z_2(z)}{z} \right)\\
-&\frac{\hi+\lo-2}{\hi\cdot\lo\cdot(\hi-\lo)\cdot z^{\hi+\lo-2}} \cdot \left(\frac{Z_0(z)}{1-z} +\frac{\ln(1-z)\cdot Z_1(z)}{z}\right) \ .
\end{align*}
We are interested in the monotonicity of $\tilde{T}(z)$.
\begin{align*}
 &\frac{\partial \tilde{T}(z)}{\partial z} \overset{!}{>} 0 \\
&\Leftrightarrow
\frac{Z_0(z)}{(1-z)^2}+\frac{\ln(1-z)\cdot Z_2(z)}{z}\\
&\hspace{1.9cm}-(\hi+\lo-2) \cdot \left(\frac{Z_0(z)}{1-z} +\frac{\ln(1-z)\cdot Z_1(z)}{z}\right)>0 \\
&\Leftrightarrow 
\frac{Z_0(z)}{1-z}-(\hi+\lo-2)\cdot Z_0(z)-f(z)\cdot \left( Z_2(z)-(\hi+\lo-2)\cdot Z_1(z)\right)>0\\
&\Leftrightarrow
\frac{Z_0(z)}{1-z}-(\hi+\lo-2)\cdot Z_0(z)+f(z)\cdot Z_0(z)\cdot(\hi-1)\cdot(\lo-1)>0 \ .
\end{align*}
Note that $Z_0(z)>0 \Leftrightarrow z<\left(\frac{\lo}{\hi}\right)^{\frac{1}{\hi-\lo}}=z'$.
Division by $Z_0(z)$ gives, by definition of $g(z)$
\begin{equation*}
\label{eq:function_critical_points}
\tag{$\star$}
\begin{split}
 &\forall z<\left(\frac{\lo}{\hi}\right)^{\frac{1}{\hi-\lo}} \colon \frac{\partial \tilde{T}(z)}{\partial z} > 0 \Leftrightarrow g(z)>0 \\
 &\forall z>\left(\frac{\lo}{\hi}\right)^{\frac{1}{\hi-\lo}} \colon \frac{\partial \tilde{T}(z)}{\partial z} < 0 \Leftrightarrow g(z)>0 \ .
\end{split}
\end{equation*}

\begin{compactenum}[1.]
 \item If $\min_z g(z)>0$ then according to \eqref{eq:function_critical_points} we have $\frac{\partial \tilde{T}(z)}{\partial z}>0$ 
for all $z<z'$ and we have $\frac{\partial \tilde{T}(z)}{\partial z}<0$ for all $z>z'$. Hence the function of critical points has a global maximum 
in $\alpha$-direction at $z'$. Consider the special case $\min_z g(z)=0$ with $z_{\min}=\arg\min_z g(z)$. According to Lemma~\ref{lem:g(z)}$(\ref{prop:g(z)_monotonicity})$
and the definition of $z_1$ and $z_2$ we have $z_1=z_2=z_{\min}$. From Lemma~\ref{lem:special_points}$(\ref{prop:z'_not_zero})$
it follows that $z_{\min}\neq z'$. Hence $z_{\min}$ must be an inflection point of
$\tilde{T}(z)$ since before and after $z_{\min}$ the monotonicity is the same.
Hence the function of critical points has a global maximum 
in $\alpha$-direction at $z'$ also in this case. 
\begin{compactenum}[$(i)$]
\item If $h(z')>1$ then according to Lemma~\ref{lem:h(z)}$(\ref{prop:h(z)_middle_interval})$ we have $z'\in(z_l,z_r)$.
It follows from Lemma~\ref{lem:classifying_extrema}$(\ref{prop:classify_fst})$ that $\tilde{T}(z')$ is a global minimum in $z$-direction.
Hence, $(z',1/h(z'))$ is the optimum point, which is according to Lemma~\ref{lem:saddle} the saddle point.

\item If $h(z')\leq 1$ then $z'$ must be from the interval $(0,z_l]$ (Lemma~\ref{lem:h(z)}$(\ref{prop:h(z)_first_interval})$)
 and not from the interval $(z_r,1)$, (Lemma~\ref{lem:h(z)}$(\ref{prop:h(z)_last_interval})$), since we have 
$f(z')>f(z_r)$ (Lemma~\ref{lem:f(z)}$(\ref{prop:f(z)_z'})$)  and $f(z)$ is monotonically decreasing (Lemma~\ref{lem:f(z)}$(\ref{prop:f(z)_decreasing})$).
But if $z'\leq z_l$ then because of the monotonicity of $\tilde{T}(z)$ the optimal
$z$ value is the nearest feasible critical point. That is the optimum point is the (degenerated) solution $(z_l,1)$.
\end{compactenum}

\item Since $\min_z g(z)<0$ it follows from Lemma~\ref{lem:h(z)}$(\ref{prop:h(z)_monotonicity_middle_interval_inc_dec_inc})$
that for $z\in[z_l,z_r)$ the function $h(z)$ is strictly increasing, reaches a maximum at $z_1$,
is strictly decreasing, reaches a minimum at $z_2$, and is strictly increasing afterwards.
Furthermore we have $g(z)>0$ for $z\in [z_l,z_1)$, $g(z)<0$ for $z\in(z_1,z_2)$, $g(z)>0$ for $z\in(z_2,z_r)$, 
and $g(z)=0$ for $z\in\{z_1,z_2\}$.
An optimal $z$ must be global minimum point in $z$-direction. According to 
Lemma~\ref{lem:critical_points_interval}$(\ref{prop:critical_points_minima})$ and 
Lemma~\ref{lem:classifying_extrema}$(\ref{prop:classify_snd})$
global minimum points are the points from $[z_l,z_1) \cup (z_2,z_r)$.
\begin{compactenum}[$(i)$]
\item An optimal $z$ cannot be from $(z_2,z_r)$ since for each $z\in(z_2,z_r)$ there is an $\eps>0$
such that $z-\eps \in (z_2,z_r)$ and  $\tilde{T}(z)<\tilde{T}(z-\eps)$. This descent converges to $z_2$.
But according to Lemma~\ref{lem:classifying_extrema}$(\ref{prop:classify_snd})$  $z_2$ is an inflection point and not a global minimum point.
Hence the optimal $z$ must be from $[z_l,z_1)$. For each $z\in(z_l,z_1)$ there is an $\eps>0$
such that $z-\eps \in [z_l,z_1)$ and  $\tilde{T}(z)<\tilde{T}(z-\eps)$. This descent converges to $z_l$.
\item According to Lemma~\ref{lem:special_points}$(\ref{prop:z'_not_zero})$ $z'\neq z_2$. It follows that $z'\in (z_l,z^{<}_1]$,
see also Lemma~\ref{lem:classifying_extrema}$(\ref{prop:classify_snd})$. An optimal $z$ cannot be from $(z_2,z_r)$
for the same reasons as in case $2(i)$.
\item Consider an arbitrary but fixed $\alpha$ with 
$1/\alpha \in (h(z_2),h(z_1))$. According to Lemma~\ref{lem:classifying_extrema}$(\ref{prop:classify_snd_three})$ we have two different points $\tilde{z}_1,\tilde{z}_2$, with 
$\tilde{z}_1<z_1<z_2<\tilde{z}_2$, that are local minimum points of the threshold function $T(z,\alpha)$ in $z$-direction.
\begin{compactitem}[$\bullet$]
 \item Let $z_1<z'<z_2$. Decreasing $\alpha$ (increasing $1/\alpha$) by an arbitrary small but fixed positive value 
gives two new 
local minimum points in $z$-direction, $\tilde{z}_1+\varepsilon$, $\tilde{z}_2+\delta$, where $\varepsilon,\delta>0$.
According to \eqref{eq:function_critical_points} it holds that 
$\tilde{T}(\tilde{z}_1)<\tilde{T}(\tilde{z}_1+\varepsilon)$
and
$\tilde{T}(\tilde{z}_2)>\tilde{T}(\tilde{z}_2+\delta)$.
Hence for the left critical point the local minimum in $z$-direction becomes smaller while the potential threshold becomes larger
and for the right critical point the local minimum in $z$-direction becomes larger while the potential threshold becomes smaller.
Increasing $\alpha$ by an arbitrary small but fixed positive value reverses the behavior.
Assume we have found an optimal $\alpha$, that is $\alpha=\alpha^*$. Decreasing $\alpha$ 
by some small fixed positive value increases the threshold for the left critical point
but because of the optimality of $\alpha$ we have no global minimum for the left critical point but only a local minimum.
Increasing $\alpha$ increases the threshold for the right critical point but because of the optimality of
$\alpha$ we have no global minimum for the right critical point but only a local minimum. 
Hence for $\alpha^*$ both critical points $z^*$ and $z^{**}$, with $1/\alpha^*=h(z^*)=h(z^{**})$,
lead to the same minimum in $z$-direction, that is both local minimum points are also global minimum points and it holds
$T(z^*,\alpha^*)=T(z^{**},\alpha^*)$ is the optimal threshold.
\item Let $z'<z_1$. Assume that $1/\alpha\in (h(z'),h(z_1))$, then $\alpha$ cannot be optimal
since increasing $\alpha$ by an arbitrary small but fixed positive value increases $\tilde{T}(\tilde{z}_1)$ as well as
$\tilde{T}(\tilde{z}_2)$ and one of the critical points must be the global minimum point in $z$-direction.
Hence the optimum $1/\alpha$ must be in the interval $[h(z_2),h(z')]$.
\item The case $z'>z_2$ is analogous to the case $z'<z_1$.
\end{compactitem}
\item According to Lemma~\ref{lem:special_points}$(\ref{prop:z'_not_zero})$ $z'\neq z_1$. It follows that $z'\in [z^{>}_2,z_r)$.
An optimal $z$ cannot be from $[z_l,z_1)$.
\end{compactenum}
\end{compactenum}
\end{proof}
For given $\vec{\dg}=(\lo,\hi)$, Algorithm~\ref{algo:opt_thresh} calculates $\alpha^*,z^*$ and $c^*=T(z^*,\alpha^*)$
of Theorem~\ref{theo:main} and optimization problem \eqref{eq:optimization_problem}, respectively. 
If one wants to determine the optimal values for fixed $\lo$ but increasing $\hi$
one can make use of the following observation. 
According to Lemma~\ref{lem:g(z)}$(\ref{prop:g(z)_b'})$
there is a threshold $\hi'$, such that for $\lo <\hi < \hi'$ 
it holds $\min_z g(z,\hi)\geq 0$ and for $\hi\geq \hi'$ it holds $g(z,\hi)<0$.
That is after reaching $b'$ we don't need to further calculate the minimum of $g(z)$.
The following table lists some values for $\hi'$.
\begin{center}
\begin{tabular}{l|rrrrrrrr}
 $\lo$ & 3  &  4  &  5  &  6 &  7  &  8 &   9  &  10 \\\hline
 $\hi'$ & 16 & 29  & 45  & 62 & 79  & 98 & 117  & 137
 \end{tabular}
\end{center}
\vfill\pagebreak
\begin{algorithm}[H]
\caption{\label{algo:opt_thresh}Optimal Thresholds}
\SetKwInput{KwPrerequisite}{Prerequisite}
\SetKwInput{KwPurpose}{Purpose}
\SetKw{KwTrue}{true}
\SetKw{KwOr}{or}
\SetKw{KwBreak}{break}

\newcommand{\ns}[2]{ \texttt{numSolve(}#1,#2\texttt{)}} 
\KwIn{$\lo,\hi$, $\eps$ (stopping criterion for binary search)}
 \KwPurpose{finds optimal thresholds for parameters $\lo$ and $\hi$. }
\KwPrerequisite{subroutine $\ns{ equation}{interval}$ that returns numerical solution of $equation$ within the given $interval$}
\emph{Initialization}:

$z_l \leftarrow \ns{ f(z)=\frac{1}{\lo-1}}{ z \in (0,1) }$\;
$z_r \leftarrow \ns{ f(z)=\frac{1}{\hi-1}}{ z \in (z_l,1)}$\;
$z_g \leftarrow \ns{ \frac{\partial g(z)}{\partial z}=0}{ z \in (0,1)}$\;
$z'        \leftarrow \left( \frac{\lo}{\hi} \right)^{\frac{1}{\hi-\lo}}$; $z_1 \leftarrow z'$; $z_2 \leftarrow z'$\;
\If{ $g(z_g,\lo,\hi)<0$ }
{
 $z_1       \leftarrow \ns{ g(z,\lo,\hi) = 0}{z \in (z_l,z_g)}$\;
 $z_2       \leftarrow \ns{ g(z,\lo,\hi) = 0}{z \in (z_g,z_r)}$\;
}
\smallskip
\emph{Optimization}:

\eIf{ $h(z',\lo,\hi)\leq 1$ }
{
 $z^* \leftarrow z_l$; $\alpha^* \leftarrow 1$; $T^* \leftarrow \frac{-\ln( 1-z_l)}{\lo\cdot z_l^{\lo-1}}$\;
}
{
  \eIf{$h(z',\lo,\hi)\leq h(z_2,\lo,\hi)$  \KwOr $h(z',\lo,\hi)\geq h(z_1,\lo,\hi)$}
  {
    $z^*      \leftarrow z'$; $\alpha^* \leftarrow \tfrac{\hi-1}{\hi-\lo}-\tfrac{1}{f({z}^*)\cdot (\hi-\lo)}$ ;
    $T^*      \leftarrow \ln\left( 1-\left(\frac{\lo}{\hi} \right)^\frac{1}{\hi-\lo}  \right)\cdot \left(\frac{\hi^{\lo-1}}{\lo^{\hi-1}}\right)^{\frac{1}{\hi-\lo}}$\;
  }
  {
    $u\leftarrow z_1$; $l \leftarrow z_2$\;
    \lIf{$z'<z_1$}{$u \leftarrow z'$}\;
    \lIf{$z'>z_2$}{$l \leftarrow z'$}\;
    $\alpha_{\mathrm{min}}\leftarrow \frac{1}{h(u,\lo,\hi)}$;
    $\alpha_{\mathrm{max}}\leftarrow \frac{1}{h(l,\lo,\hi)}$\;
    \While{\KwTrue}
    {
        $\alpha^*\leftarrow \frac{\alpha_{\mathrm{max}}-\alpha_{\mathrm{min}}}{2}+\alpha_{\mathrm{min}}$\;
        $z^{**}   \leftarrow \ns{h(z,\lo,\hi)-\frac{1}{\alpha^*}=0}{z\in(z_l, u)}$\;
        $z^{*\ }\leftarrow     \ns{h(z,\lo,\hi)-\frac{1}{\alpha^*}=0}{z\in(l      ,z_r)}$\; 
        $t^{**} \leftarrow T(z^{**\ },\lo,\hi,\alpha^*) $\;
        $t^{*\ }  \leftarrow T(z^{*},\lo,\hi,\alpha^*) $\;
        \eIf{ $|t^{*\ }-t^{**}|<\eps$}
        {
           \KwBreak
        }
        {

            \lIf{ $t^{*\ }>t^{**}$ }
            {
              $\alpha_{\mathrm{min}}\leftarrow \alpha^*$\;
            }
            \lElse
            {
              $\alpha_{\mathrm{max}} \leftarrow \alpha^*$\;
            }
        }

    }
  }
}

$\Return(z^*,\alpha^*,T^*)$
\end{algorithm}

\vspace{-1cm}
\section{Optimal Thresholds}
\label{app:optimal_values}
The following four tables list optimal thresholds for different edge sizes  $\lo=\dg_1$ and $\hi=\dg_2$, with $\lo\in\{3,4,5,6\}$ and $\lo\leq\hi\leq 50$.\hfill\\
\begin{minipage}{0.5\textwidth}
  \centering
\scalebox{0.8}{
\begin{tabular}{c|ccccc}
$\hi$ & $z^*$ & $\lambda^*$ & $\alpha^*$ & $\Adg$ & $c^*$\\\hline
3 & 0.71533 & 1.25643 & 1.00000 & 3.00000 & 0.81847 \\\hline
4 & 0.75000 & 1.38629 & 0.83596 & 3.16404 & 0.82151 \\ 
5 & 0.77460 & 1.48986 & 0.84671 & 3.30658 & 0.82770 \\ 
6 & 0.79370 & 1.57843 & 0.85419 & 3.43744 & 0.83520 \\ 
7 & 0.80911 & 1.65604 & 0.86014 & 3.55944 & 0.84321 \\ 
8 & 0.82188 & 1.72527 & 0.86512 & 3.67439 & 0.85138 \\ 
9 & 0.83268 & 1.78787 & 0.86940 & 3.78359 & 0.85952 \\ 
10 & 0.84198 & 1.84505 & 0.87315 & 3.88795 & 0.86752 \\ 
11 & 0.85009 & 1.89774 & 0.87648 & 3.98818 & 0.87535 \\ 
12 & 0.85724 & 1.94662 & 0.87946 & 4.08482 & 0.88298 \\ 
13 & 0.86361 & 1.99224 & 0.88217 & 4.17830 & 0.89040 \\ 
14 & 0.86932 & 2.03503 & 0.88464 & 4.26898 & 0.89761 \\ 
15 & 0.87449 & 2.07533 & 0.88690 & 4.35715 & 0.90461 \\ 
16 & 0.90263 & 2.32922 & 0.88684 & 4.47102 & 0.91089 \\ 
17 & 0.92384 & 2.57487 & 0.88616 & 4.59372 & 0.91510 \\ 
18 & 0.93703 & 2.76508 & 0.88599 & 4.71015 & 0.91772 \\ 
19 & 0.94632 & 2.92464 & 0.88620 & 4.82077 & 0.91922 \\ 
20 & 0.95328 & 3.06354 & 0.88671 & 4.92601 & 0.91992 \\ 
\rowcolor{lightgray}21 & 0.95871 & 3.18715 & 0.88743 & 5.02626 & 0.92004 \\ 
22 & 0.96307 & 3.29883 & 0.88832 & 5.12190 & 0.91974 \\ 
23 & 0.96666 & 3.40086 & 0.88934 & 5.21328 & 0.91914 \\ 
24 & 0.96965 & 3.49488 & 0.89044 & 5.30070 & 0.91832 \\ 
25 & 0.97218 & 3.58213 & 0.89162 & 5.38446 & 0.91734 \\ 
26 & 0.97436 & 3.66355 & 0.89283 & 5.46482 & 0.91626 \\ 
27 & 0.97624 & 3.73993 & 0.89408 & 5.54202 & 0.91510 \\ 
28 & 0.97789 & 3.81185 & 0.89535 & 5.61628 & 0.91390 \\ 
29 & 0.97935 & 3.87983 & 0.89662 & 5.68780 & 0.91266 \\ 
30 & 0.98064 & 3.94430 & 0.89790 & 5.75675 & 0.91141 \\ 
31 & 0.98179 & 4.00560 & 0.89917 & 5.82330 & 0.91016 \\ 
32 & 0.98282 & 4.06404 & 0.90043 & 5.88761 & 0.90891 \\ 
33 & 0.98375 & 4.11988 & 0.90167 & 5.94981 & 0.90768 \\ 
34 & 0.98460 & 4.17334 & 0.90290 & 6.01003 & 0.90645 \\ 
35 & 0.98537 & 4.22462 & 0.90411 & 6.06839 & 0.90525 \\ 
36 & 0.98607 & 4.27390 & 0.90530 & 6.12498 & 0.90406 \\ 
37 & 0.98672 & 4.32132 & 0.90647 & 6.17992 & 0.90290 \\ 
38 & 0.98731 & 4.36702 & 0.90762 & 6.23328 & 0.90177 \\ 
39 & 0.98786 & 4.41113 & 0.90875 & 6.28516 & 0.90065 \\ 
40 & 0.98837 & 4.45375 & 0.90985 & 6.33562 & 0.89956 \\ 
41 & 0.98884 & 4.49497 & 0.91093 & 6.38475 & 0.89850 \\ 
42 & 0.98927 & 4.53489 & 0.91198 & 6.43260 & 0.89746 \\ 
43 & 0.98968 & 4.57359 & 0.91302 & 6.47925 & 0.89644 \\ 
44 & 0.99006 & 4.61114 & 0.91403 & 6.52474 & 0.89545 \\ 
45 & 0.99042 & 4.64761 & 0.91502 & 6.56913 & 0.89448 \\ 
46 & 0.99075 & 4.68305 & 0.91599 & 6.61246 & 0.89354 \\ 
47 & 0.99106 & 4.71752 & 0.91694 & 6.65480 & 0.89261 \\ 
48 & 0.99136 & 4.75108 & 0.91786 & 6.69617 & 0.89171 \\ 
49 & 0.99164 & 4.78376 & 0.91877 & 6.73662 & 0.89083 \\ 
50 & 0.99190 & 4.81563 & 0.91966 & 6.77619 & 0.88997 \\
\end{tabular}}

\parbox{0.9\textwidth}{
\captionof{table}{Optimal values for $\lo=3$ and $\lo\leq\hi\leq 50$. The maximum threshold in this range is about $0.92004$.}}
\end{minipage}\hfill
\begin{minipage}{0.5\textwidth}
\centering
\scalebox{0.8}{
\begin{tabular}{c|ccccc}
$\hi$ & $z^*$ & $\lambda^*$ & $\alpha^*$ & $\Adg$ & $c^*$\\\hline
\\
4 & 0.85100 & 1.90381 & 1.00000 & 4.00000 & 0.77228 \\\hline
5 & 0.85100 & 1.90381 & 1.00000 & 4.00000 & 0.77228 \\ 
6 & 0.85100 & 1.90381 & 1.00000 & 4.00000 & 0.77228 \\ 
7 & 0.85100 & 1.90381 & 1.00000 & 4.00000 & 0.77228 \\ 
8 & 0.85100 & 1.90381 & 1.00000 & 4.00000 & 0.77228 \\ 
9 & 0.85100 & 1.90381 & 1.00000 & 4.00000 & 0.77228 \\ 
10 & 0.85837 & 1.95457 & 0.98319 & 4.10087 & 0.77261 \\ 
11 & 0.86544 & 2.00576 & 0.97048 & 4.20664 & 0.77358 \\ 
12 & 0.87169 & 2.05327 & 0.96143 & 4.30855 & 0.77501 \\ 
13 & 0.87725 & 2.09762 & 0.95477 & 4.40707 & 0.77677 \\ 
14 & 0.88225 & 2.13922 & 0.94974 & 4.50259 & 0.77878 \\ 
15 & 0.88678 & 2.17841 & 0.94587 & 4.59540 & 0.78097 \\ 
16 & 0.89090 & 2.21548 & 0.94285 & 4.68579 & 0.78329 \\ 
17 & 0.89467 & 2.25065 & 0.94046 & 4.77397 & 0.78571 \\ 
18 & 0.89814 & 2.28411 & 0.93856 & 4.86013 & 0.78819 \\ 
19 & 0.90134 & 2.31604 & 0.93704 & 4.94444 & 0.79072 \\ 
20 & 0.90430 & 2.34658 & 0.93581 & 5.02703 & 0.79329 \\ 
21 & 0.90706 & 2.37584 & 0.93482 & 5.10804 & 0.79587 \\ 
22 & 0.90964 & 2.40393 & 0.93402 & 5.18757 & 0.79847 \\ 
23 & 0.91205 & 2.43096 & 0.93338 & 5.26572 & 0.80106 \\ 
24 & 0.91431 & 2.45699 & 0.93287 & 5.34257 & 0.80365 \\ 
25 & 0.91643 & 2.48211 & 0.93247 & 5.41822 & 0.80623 \\ 
26 & 0.91844 & 2.50638 & 0.93215 & 5.49272 & 0.80880 \\ 
27 & 0.92033 & 2.52986 & 0.93191 & 5.56615 & 0.81135 \\ 
28 & 0.92212 & 2.55259 & 0.93173 & 5.63855 & 0.81388 \\ 
29 & 0.93047 & 2.66596 & 0.93157 & 5.71069 & 0.81638 \\ 
30 & 0.94616 & 2.92176 & 0.93133 & 5.78542 & 0.81858 \\ 
31 & 0.95404 & 3.08007 & 0.93119 & 5.85792 & 0.82036 \\ 
32 & 0.95955 & 3.20757 & 0.93113 & 5.92826 & 0.82179 \\ 
33 & 0.96375 & 3.31740 & 0.93115 & 5.99653 & 0.82293 \\ 
34 & 0.96713 & 3.41513 & 0.93124 & 6.06282 & 0.82383 \\ 
35 & 0.96992 & 3.50380 & 0.93138 & 6.12721 & 0.82453 \\ 
36 & 0.97227 & 3.58530 & 0.93157 & 6.18979 & 0.82505 \\ 
37 & 0.97429 & 3.66093 & 0.93180 & 6.25063 & 0.82544 \\ 
38 & 0.97605 & 3.73160 & 0.93206 & 6.30980 & 0.82570 \\ 
39 & 0.97758 & 3.79802 & 0.93236 & 6.36738 & 0.82586 \\ 
\rowcolor{lightgray}40 & 0.97895 & 3.86074 & 0.93268 & 6.42343 & 0.82593 \\ 
41 & 0.98016 & 3.92020 & 0.93303 & 6.47802 & 0.82593 \\ 
42 & 0.98125 & 3.97674 & 0.93339 & 6.53122 & 0.82587 \\ 
43 & 0.98224 & 4.03068 & 0.93377 & 6.58308 & 0.82576 \\ 
44 & 0.98313 & 4.08225 & 0.93416 & 6.63365 & 0.82560 \\ 
45 & 0.98394 & 4.13168 & 0.93456 & 6.68299 & 0.82540 \\ 
46 & 0.98469 & 4.17914 & 0.93497 & 6.73115 & 0.82518 \\ 
47 & 0.98537 & 4.22479 & 0.93539 & 6.77818 & 0.82492 \\ 
48 & 0.98600 & 4.26878 & 0.93582 & 6.82413 & 0.82465 \\ 
49 & 0.98658 & 4.31123 & 0.93624 & 6.86903 & 0.82436 \\ 
50 & 0.98712 & 4.35225 & 0.93668 & 6.91294 & 0.82405 \\
\end{tabular}}
\parbox{0.9\textwidth}{
\captionof{table}{Optimal values for $\lo=4$ and $\lo\leq\hi\leq 50$. The maximum threshold in this range is about $0.82593$.}}
\end{minipage}

\begin{minipage}{0.5\textwidth}
\centering
\scalebox{0.8}{
\begin{tabular}{c|ccccc}
$\hi$ & $z^*$ & $\lambda^*$ & $\alpha^*$ & $\Adg$ & $c^*$\\\hline
5 & 0.90335 & 2.33666 & 1.00000 & 5.00000 & 0.70178 \\\hline
6 & 0.90335 & 2.33666 & 1.00000 & 5.00000 & 0.70178 \\ 
7 & 0.90335 & 2.33666 & 1.00000 & 5.00000 & 0.70178 \\ 
8 & 0.90335 & 2.33666 & 1.00000 & 5.00000 & 0.70178 \\ 
9 & 0.90335 & 2.33666 & 1.00000 & 5.00000 & 0.70178 \\ 
10 & 0.90335 & 2.33666 & 1.00000 & 5.00000 & 0.70178 \\ 
11 & 0.90335 & 2.33666 & 1.00000 & 5.00000 & 0.70178 \\ 
12 & 0.90335 & 2.33666 & 1.00000 & 5.00000 & 0.70178 \\ 
13 & 0.90335 & 2.33666 & 1.00000 & 5.00000 & 0.70178 \\ 
14 & 0.90335 & 2.33666 & 1.00000 & 5.00000 & 0.70178 \\ 
15 & 0.90335 & 2.33666 & 1.00000 & 5.00000 & 0.70178 \\ 
16 & 0.90335 & 2.33666 & 1.00000 & 5.00000 & 0.70178 \\ 
17 & 0.90335 & 2.33666 & 1.00000 & 5.00000 & 0.70178 \\ 
18 & 0.90617 & 2.36622 & 0.99375 & 5.08121 & 0.70187 \\ 
19 & 0.90905 & 2.39743 & 0.98793 & 5.16898 & 0.70215 \\ 
20 & 0.91172 & 2.42727 & 0.98300 & 5.25495 & 0.70258 \\ 
21 & 0.91421 & 2.45588 & 0.97880 & 5.33924 & 0.70315 \\ 
22 & 0.91654 & 2.48335 & 0.97518 & 5.42198 & 0.70383 \\ 
23 & 0.91871 & 2.50978 & 0.97204 & 5.50326 & 0.70460 \\ 
24 & 0.92076 & 2.53524 & 0.96931 & 5.58318 & 0.70545 \\ 
25 & 0.92268 & 2.55981 & 0.96691 & 5.66183 & 0.70637 \\ 
26 & 0.92450 & 2.58356 & 0.96480 & 5.73927 & 0.70734 \\ 
27 & 0.92621 & 2.60653 & 0.96293 & 5.81557 & 0.70836 \\ 
28 & 0.92783 & 2.62878 & 0.96127 & 5.89081 & 0.70942 \\ 
29 & 0.92937 & 2.65036 & 0.95979 & 5.96502 & 0.71051 \\ 
30 & 0.93084 & 2.67130 & 0.95847 & 6.03827 & 0.71163 \\ 
31 & 0.93223 & 2.69165 & 0.95728 & 6.11061 & 0.71278 \\ 
32 & 0.93356 & 2.71143 & 0.95622 & 6.18206 & 0.71394 \\ 
33 & 0.93483 & 2.73069 & 0.95526 & 6.25268 & 0.71512 \\ 
34 & 0.93604 & 2.74944 & 0.95440 & 6.32250 & 0.71631 \\ 
35 & 0.93720 & 2.76772 & 0.95361 & 6.39156 & 0.71752 \\ 
36 & 0.93831 & 2.78556 & 0.95291 & 6.45989 & 0.71873 \\ 
37 & 0.93937 & 2.80296 & 0.95227 & 6.52751 & 0.71995 \\ 
38 & 0.94039 & 2.81996 & 0.95168 & 6.59446 & 0.72117 \\ 
39 & 0.94137 & 2.83657 & 0.95115 & 6.66075 & 0.72240 \\ 
40 & 0.94232 & 2.85281 & 0.95067 & 6.72643 & 0.72362 \\ 
41 & 0.94323 & 2.86870 & 0.95024 & 6.79150 & 0.72485 \\ 
42 & 0.94410 & 2.88425 & 0.94984 & 6.85599 & 0.72608 \\ 
43 & 0.94495 & 2.89948 & 0.94948 & 6.91991 & 0.72731 \\ 
44 & 0.94576 & 2.91440 & 0.94915 & 6.98330 & 0.72853 \\ 
45 & 0.95990 & 3.21627 & 0.94897 & 7.04133 & 0.72973 \\ 
46 & 0.96531 & 3.36136 & 0.94889 & 7.09535 & 0.73078 \\ 
47 & 0.96890 & 3.47066 & 0.94885 & 7.14826 & 0.73171 \\ 
48 & 0.97164 & 3.56291 & 0.94883 & 7.20010 & 0.73252 \\ 
49 & 0.97386 & 3.64437 & 0.94884 & 7.25089 & 0.73322 \\ 
\rowcolor{lightgray}50 & 0.97572 & 3.71813 & 0.94887 & 7.30068 & 0.73384 \\ 
\end{tabular}}

\parbox{0.9\textwidth}{
\captionof{table}{Optimal values for $\lo=5$ and $\lo\leq\hi\leq 50$. The maximum threshold in this range is about $0.73384$.}}
\end{minipage}\hfill
\begin{minipage}{0.5\textwidth}
\centering
\scalebox{0.8}{
\begin{tabular}{c|ccccc}
$\hi$ & $z^*$ & $\lambda^*$ & $\alpha^*$ & $\Adg$ & $c^*$\\\hline
\\
6 & 0.93008 & 2.66040 & 1.00000 & 6.00000 & 0.63708 \\\hline
7 & 0.93008 & 2.66040 & 1.00000 & 6.00000 & 0.63708 \\ 
8 & 0.93008 & 2.66040 & 1.00000 & 6.00000 & 0.63708 \\ 
9 & 0.93008 & 2.66040 & 1.00000 & 6.00000 & 0.63708 \\ 
10 & 0.93008 & 2.66040 & 1.00000 & 6.00000 & 0.63708 \\ 
11 & 0.93008 & 2.66040 & 1.00000 & 6.00000 & 0.63708 \\ 
12 & 0.93008 & 2.66040 & 1.00000 & 6.00000 & 0.63708 \\ 
13 & 0.93008 & 2.66040 & 1.00000 & 6.00000 & 0.63708 \\ 
14 & 0.93008 & 2.66040 & 1.00000 & 6.00000 & 0.63708 \\ 
15 & 0.93008 & 2.66040 & 1.00000 & 6.00000 & 0.63708 \\ 
16 & 0.93008 & 2.66040 & 1.00000 & 6.00000 & 0.63708 \\ 
17 & 0.93008 & 2.66040 & 1.00000 & 6.00000 & 0.63708 \\ 
18 & 0.93008 & 2.66040 & 1.00000 & 6.00000 & 0.63708 \\ 
19 & 0.93008 & 2.66040 & 1.00000 & 6.00000 & 0.63708 \\ 
20 & 0.93008 & 2.66040 & 1.00000 & 6.00000 & 0.63708 \\ 
21 & 0.93008 & 2.66040 & 1.00000 & 6.00000 & 0.63708 \\ 
22 & 0.93008 & 2.66040 & 1.00000 & 6.00000 & 0.63708 \\ 
23 & 0.93008 & 2.66040 & 1.00000 & 6.00000 & 0.63708 \\ 
24 & 0.93008 & 2.66040 & 1.00000 & 6.00000 & 0.63708 \\ 
25 & 0.93008 & 2.66040 & 1.00000 & 6.00000 & 0.63708 \\ 
26 & 0.93008 & 2.66040 & 1.00000 & 6.00000 & 0.63708 \\ 
27 & 0.93088 & 2.67194 & 0.99807 & 6.04054 & 0.63709 \\ 
28 & 0.93237 & 2.69378 & 0.99463 & 6.11825 & 0.63717 \\ 
29 & 0.93379 & 2.71495 & 0.99153 & 6.19490 & 0.63732 \\ 
30 & 0.93514 & 2.73551 & 0.98873 & 6.27054 & 0.63754 \\ 
31 & 0.93642 & 2.75549 & 0.98619 & 6.34522 & 0.63781 \\ 
32 & 0.93765 & 2.77491 & 0.98388 & 6.41900 & 0.63812 \\ 
33 & 0.93881 & 2.79382 & 0.98178 & 6.49189 & 0.63849 \\ 
34 & 0.93993 & 2.81224 & 0.97986 & 6.56396 & 0.63889 \\ 
35 & 0.94100 & 2.83020 & 0.97810 & 6.63523 & 0.63932 \\ 
36 & 0.94202 & 2.84771 & 0.97648 & 6.70573 & 0.63979 \\ 
37 & 0.94301 & 2.86481 & 0.97498 & 6.77551 & 0.64028 \\ 
38 & 0.94395 & 2.88151 & 0.97361 & 6.84458 & 0.64080 \\ 
39 & 0.94486 & 2.89783 & 0.97233 & 6.91297 & 0.64134 \\ 
40 & 0.94573 & 2.91379 & 0.97116 & 6.98071 & 0.64190 \\ 
41 & 0.94657 & 2.92941 & 0.97006 & 7.04783 & 0.64248 \\ 
42 & 0.94738 & 2.94469 & 0.96905 & 7.11433 & 0.64308 \\ 
43 & 0.94816 & 2.95966 & 0.96810 & 7.18026 & 0.64369 \\ 
44 & 0.94892 & 2.97433 & 0.96722 & 7.24562 & 0.64431 \\ 
45 & 0.94965 & 2.98871 & 0.96640 & 7.31044 & 0.64494 \\ 
46 & 0.95035 & 3.00281 & 0.96563 & 7.37472 & 0.64558 \\ 
47 & 0.95103 & 3.01665 & 0.96491 & 7.43850 & 0.64624 \\ 
48 & 0.95170 & 3.03022 & 0.96424 & 7.50178 & 0.64690 \\ 
49 & 0.95233 & 3.04355 & 0.96361 & 7.56458 & 0.64756 \\ 
\rowcolor{lightgray}50 & 0.95295 & 3.05665 & 0.96302 & 7.62692 & 0.64823 \\
\end{tabular}}

\parbox{0.9\textwidth}{
\captionof{table}{Optimal values for $\lo=6$ and $\lo\leq\hi\leq 50$. The maximum threshold in this range is about $0.64823$.}}
\end{minipage}

\pagebreak
\section{Properties of $f(z)$}
\label{app:f(z)}
In this section we prove Lemma~\ref{lem:f(z)}.

\begin{asparaenum}[($i$)]
\Item 
\begin{align*}
 &f(z)=\frac{-\ln(1-z)\cdot (1-z)}{z} \overset{!}{>} 1-z \\
 &\Leftrightarrow -\ln(1-z)>z 
 \Leftrightarrow \frac{1}{1-z} > e^z\Leftrightarrow e^{-z} >1-z  \ \checkmark
\end{align*}
\item Applying L'Hôpital's rule it follows that
\begin{equation*}
 \lim_{z\to 0} f(z)=\lim_{z \to 0}\frac{-\ln(1-z)\cdot (1-z)}{z}=\lim_{z\to 0} \frac{\frac{1}{|1-z|}\cdot (1-z)+\ln(1-z)}{1}=1 \ . 
\end{equation*}
\item Applying L'Hôpital's rule it follows that
\begin{align*}
 \lim_{z\to 1} f(z)&=\lim_{z \to 1}\frac{-\ln(1-z)\cdot (1-z)}{z}=\lim_{z \to 1}\frac{-\ln(1-z)}{\frac{z}{1-z}}
=\lim_{z\to 1} \frac{\frac{1}{|1-z|}}{\frac{1-z+z}{(1-z)^2}}\\
&=\lim_{z\to 1}1-z=0 \ .
\end{align*}
\Item
\begin{equation*}
 \frac{\de f(z)}{\de z}=\frac{z+\ln(1-z)}{z^2} \overset{!}{<} 0 \Leftrightarrow \ln(1-z)<-z \Leftrightarrow 1-z < e^{-z} \ \checkmark
\end{equation*}
\Item
\begin{align*}
 &\frac{\de^2 f(z)}{(\de z)^2}=\frac{-2z+z^2-2\ln(1-z)\cdot(1-z)}{(1-z)\cdot z^3} \overset{!}{<}0\\
 \Leftrightarrow& \underbrace{\frac{z^2-2\cdot z}{1-z}}_{f_1(z)} < \underbrace{2\ln(1-z)}_{f_2(z)}
\end{align*}
which is true since it holds 
\begin{itemize}
 \item $\lim_{z\to 0} f_1(z)=\lim_{z\to 0}f_2(z)=0$ and 
 \item $\frac{\de f_1(z)}{\de z}= \frac{-z^2+2\cdot z-2}{(1-z)^2}< \frac{\de f_2(z)}{\de z}=\frac{-2}{1-z}<0$.
\end{itemize}
\item First we show that $z'$ is strictly increasing for growing $\lo$.
Utilizing Lemma~\ref{lem:f(z)}$(\ref{prop:f(z)_decreasing})$ this implies that $f(z')$ is striclty monotonically decreasing for growing $\lo$. 
\begin{align*}
&\frac{\partial z'}{\partial \lo} = z' \cdot \left( \frac{\ln(\lo/\hi)}{(\hi-\lo)^2} + \frac{1}{\lo\cdot (\hi-\lo)} \right) \overset{!}{>}0\\
\Leftrightarrow & \frac{\hi-\lo}{\lo}>\ln(\hi/\lo)
\Leftrightarrow \exp\left( \frac{\hi-\lo}{\lo} \right) > \frac{\hi}{\lo} \\
\Leftrightarrow &\sum_{i=0}^\infty \left(\frac{\hi-\lo}{\lo}\right)^i \cdot \frac{1}{i!} > \frac{\hi}{\lo} \\
\Leftrightarrow & 1 + \frac{\hi-\lo}{\lo}+\underbrace{\left(\frac{\hi-\lo}{\lo}\right)^2\cdot \frac{1}{2} +\left(\frac{\hi-\lo}{\lo}\right)^3\cdot \frac{1}{6}+ \ldots}_{>0} >  \frac{\hi}{\lo} \ \checkmark
\end{align*}
Now all we need to show is that our assumption holds for the maximum value of $\lo$, that is $\lo = \hi-1$.
\begin{align*}
 f\left( \left(\frac{\hi-1}{\hi}\right)^{\frac{1}{\hi-(\hi-1)}} \right)\overset{!}{>} \frac{1}{\hi-1}
&\Leftrightarrow \frac{-\ln\left(1-\frac{\hi-1}{\hi}\right) \cdot \left(1-\frac{\hi-1}{\hi}\right)}{\frac{\hi-1}{\hi}}>\frac{1}{\hi-1}\\
&\Leftrightarrow  \frac{-\ln(1/\hi)}{\hi-1}>\frac{1}{\hi-1} \Leftrightarrow \ln(\hi)>1
\end{align*}
which is true since $\hi\geq 4$.
\item For $z\in(0,1)$ we have $f(z)\in(0,1)$. Consider $\varphi(\lo,\hi)$ with
\begin{equation*}
\varphi(\lo,\hi)=-\frac{1}{1-z'}-2+\hi+\lo=-\frac{1}{1-\left(\frac{\lo}{\hi}\right)^{\frac{1}{\hi-\lo}}}-2+\hi+\lo \ .
\end{equation*}
Let $\lo$ be fixed. We will show that $z'=z'(\hi)$ is strictly increasing for growing $\hi$.
\begin{align*}
 &\frac{\partial z'}{\partial \hi}= z' \cdot \left(- \frac{\ln(\lo/\hi)}{(\hi-\lo)^2} - \frac{1}{\hi\cdot (\hi-\lo)} \right) \overset{!}{>}0\\
&\Leftrightarrow 
 \frac{-\ln(\frac{\lo}{\hi})}{(\hi-\lo)^2} >\frac{1}{\hi\cdot(\hi-\lo)} \Leftrightarrow \frac{-\hi+\lo}{\hi}>\ln\left(\frac{\lo}{\hi}\right )\\
& \Leftrightarrow x>\ln(1+x) \ ,
\end{align*}
for $\frac{\lo}{\hi}=1+x$  and $-1<x<0$. Hence
\begin{align*}
  \frac{\partial z'}{\partial \hi}\overset{!}{<}0&\Leftrightarrow x>\ln(1+x) \Leftrightarrow x> \sum_{i=0}^{\infty}(-1)^i\cdot \frac{x^{i+1}}{(i+1)}\\
 &\Leftrightarrow x > x \underbrace{-x^2/2+x^3/3-x^4/4}_{<0} \checkmark
\end{align*}
It follows directly that the function $\varphi(\hi)$ is strictly increasing for growing $\hi$.
Consider the minimum of 
$
 \varphi(\lo+1,\lo)=-\frac{1}{1-\left(\frac{\lo}{\lo+1}\right)}-1+2\cdot \lo=\lo-2 \ .
$
Since $\lo\geq 3$ we have $\varphi(\lo+1,\lo)\geq 1$ which is not in the range of $f(z)$  for $z\in (0,1)$.
\end{asparaenum}

\section{Properties of $g(z,\lo,\hi)$}
\label{app:g(z)}
In this section we prove Lemma~\ref{lem:g(z)}.
\begin{asparaenum}[$(i)$]
 \item Consider the first derivative of $g(z)$.
 \begin{align*}
 &\frac{\partial g(z)}{\partial z}=\frac{\ln(1-z)\cdot(\hi-1)\cdot(\lo-1)}{z^2}+\frac{1}{(1-z)^2}+\frac{(\hi-1)\cdot(\lo-1)}{z}\overset{!}{<}0 \\
\Leftrightarrow & \frac{-f(z)\cdot(\hi-1)\cdot(\lo-1)}{z\cdot(1-z)} +\frac{1}{(1-z)^2} + \frac{(\hi-1)\cdot(\lo-1)}{z} < 0 \\
\Leftrightarrow &  1-z + \frac{z}{1-z}\cdot \frac{1}{(\hi-1)\cdot (\lo-1)}<f(z) \\
\Leftrightarrow & \frac{1}{(\hi-1)\cdot (\lo-1)} <\underbrace{\frac{1-z}{z} \cdot(f(z)-(1-z))}_{g_1(z)}
\end{align*}
Hence we have $\frac{\partial g(z)}{\partial z}\ R \ 0  \Leftrightarrow \frac{1}{(\hi-1)\cdot (\lo-1)} \ R \ g_1(z)$ for $R\in \{<,>,=\}$.
Now consider $g_1(z)$. It holds that 
\begin{itemize}
\item  $\lim\limits_{z\to 0} g_1(z)=0.5$, since 
\begin{align*}
 \lim_{z\to 0} g_1(z)&=\lim_{z\to 0} \frac{(1-z)\cdot(f(z)-1+z)}{z}\\
&=\lim_{z\to 0} \frac{-1\cdot(f(z)-1+z)+(1-z)\cdot \left(\frac{\de f(z)}{\de z}+1\right) }{1}=1+\lim_{z\to0}\frac{\de f(z)}{\de z} \\
&=1+\lim_{z\to 0}\frac{z+\ln(1-z)}{z^2}=1+\lim_{z\to 0}\frac{1+\frac{-1}{1-z}}{2\cdot z}\\
&=1+\lim_{z\to 0}\frac{-z}{2\cdot z\cdot(1-z)}= 1+\lim_{z\to 0}\frac{-1}{2\cdot (1-z)-2\cdot z}= 0.5 \ ,
\end{align*}
using L'Hôpital's rule three times.
\item $\lim\limits_{z\to 1} g_1(z)= \frac{1-1}{1}\cdot (0-1+1)=0$.
\item $g_1(z)$ is strictly decreasing for growing $z\in(0,1)$, since
\begin{align*}
  &\frac{\de g_1(z)}{\de z}=-\frac{-2\cdot \ln(1-z) \cdot (1-z)+ z^3+z^2-2\cdot z}{z^3} \overset{!}{<}0 \\
 \Leftrightarrow &2\cdot \ln(1-z) \cdot (1-z)<z^3+z^2-2\cdot z \Leftrightarrow \ln(1-z)<\underbrace{\frac{z^3+z^2-2\cdot z}{2\cdot(1-z)}}_{g_2(z)} \ ,
\end{align*}
which is true because
\begin{itemize}
\item $\lim_{z\to 0}\ln(1-z)=0 = \lim_{z\to 0}g_2(z)=0$ and 
\item $\frac{\de \ln(1-z)}{\de z}=\frac{-1}{1-z}< \frac{\de g_2(z)}{\de z}=-1-z<0$.
\end{itemize}
\end{itemize}
Using that $0<\frac{1}{(\hi-1)\cdot (\lo-1)}<0.5$ it follows that
for growing $z$ there is a first phase with $g_1(z)> \frac{1}{(\hi-1)\cdot (\lo-1)}$ which implies $\frac{\de g(z)}{\de z}{<}0$. Then there is
exactly one $z$ where $g_1(z)= \frac{1}{(\hi-1)\cdot (\lo-1)}$, which is a local minimum.
After this point we have $g_1(z) < \frac{1}{(\hi-1)\cdot (\lo-1)}$ which impilies $\frac{\de g(z)}{\de z}{>}0$.
It follows that the local minimum is actual a global minimum. 

\item If  $z \in (0,z_l]$ then it holds $f(z)\geq \frac{1}{\lo-1}>\frac{1}{\hi-1}$. 
Furthermore, according to Lemma~\ref{lem:f(z)}$(\ref{prop:f(z)_1-z})$ we have $\frac{1}{1-z}>\frac{1}{f(z)}$.
Let $f(z)=\frac{1+\eps}{\lo-1}$ and $f(z)=\frac{1+\delta}{\hi-1}$
as well as $\frac{1}{1-z}=\frac{1+\gamma}{f(z)}$ with $\eps\geq0$ and $\delta,\gamma>0$. Using that $f(z)>0$, for $z\in(0,1)$, it follows that
\begin{align*}
 g(z)\overset{!}{>}0 &\Leftrightarrow f(z)\cdot(\hi-1)\cdot(\lo-1)+\tfrac{1+\gamma}{f(z)}-(\lo-1)-(\hi-1)>0 \\
&\Leftrightarrow f(z)^2\cdot(\hi-1)\cdot(\lo-1)+(1+\gamma)-f(z)\cdot (\lo-1)-f(z)\cdot (\hi-1)>0\\
&\Leftrightarrow (1+\eps)\cdot(1+\delta)+(1+\gamma)-(1+\eps)-(1+\delta)>0\\
&\Leftrightarrow \eps\cdot \delta+\gamma>0 \ \checkmark
\end{align*}

\item If  $z \in [z_r,1)$ then it holds $f(z)\leq \frac{1}{\hi-1}<\frac{1}{\lo-1}$.
Furthermore, according to Lemma~\ref{lem:f(z)}$(\ref{prop:f(z)_1-z})$ we have $\frac{1}{1-z}>\frac{1}{f(z)}$.
Let $f(z)=\frac{1-\eps}{\lo-1}$ and $f(z)=\frac{1-\delta}{\hi-1}$
as well as $\frac{1}{1-z}=\frac{1+\gamma}{f(z)}$ with $\eps\geq0$ and $\delta,\gamma>0$.
Following the proof of Lemma~\ref{lem:g(z)}$(\ref{prop:g(z)_left_interval})$ we get
\begin{align*}
 g(z)\overset{!}{>}0 
&\Leftrightarrow (1-\eps)\cdot(1-\delta)+(1+\gamma)-(1-\eps)-(1-\delta)>0\\
&\Leftrightarrow \eps\cdot \delta + \gamma>0 \ \checkmark
\end{align*}
\item The existance of the roots $z_1$ and $z_2$ follows directly from Lemma~\ref{lem:g(z)}~$(\ref{prop:g(z)_monotonicity})$.
Moreover, from Lemma~\ref{lem:g(z)}~$(\ref{prop:g(z)_left_interval})$,~$(\ref{prop:g(z)_right_interval})$ it follows that if $g(z)\leq 0$ for $z\in(0,1)$ then
it holds $z\in(z_l,z_r)$.
\item Let $z>z_l$, that is $f(z)=\frac{1-\eps}{\lo-1}$ for $\eps>0$. If follows
\begin{align*}
g(z,\lo,\hi)
&=f(z)\cdot(\hi-1)\cdot(\lo-1)+\frac{1}{1-z}+2-\hi-\lo\\
&=f(z)\cdot \hi\cdot(\lo-1)+\frac{1}{1-z} +2-(\hi+1)-\lo -f(z)\cdot(\lo-1)+1 \\
&=g(z,\lo,\hi+1)-f(z)\cdot(\lo-1)+1=g(z,\lo,\hi+1)-(1-\eps)+1\\
&>g(z,\lo,\hi+1) \ .
\end{align*}

\item Assume that there is some $\hi'$ such that $\check{z}=\min_z g(z,\lo,\hi')<0$. 
Then from Lemma~\ref{lem:g(z)}$(\ref{prop:g(z)_left_interval}),(\ref{prop:g(z)_right_interval})$ it follows that
$\check{z}>z_l$. Using Lemma~\ref{lem:g(z)}$(\ref{prop:g(z)_monotonicity_in_b})$ we conclude that 
for all $b\geq \hi'$ it holds that $ g(\check{z},\lo,\hi)<0$ and therefore $\min_zg(z,\lo,\hi)<0$ as well.
It remains to find one such $\hi'$.   

Consider the inequality $g(z',\lo,\hi)\geq 0$ which is equivalent to
\begin{equation*}
 (\hi-1)\cdot \Big( \underbrace{f(z')\cdot(\lo-1)}_{g_1(\lo,\hi)}+\underbrace{\frac{1}{1-z'}\cdot \frac{1}{\hi-1}}_{g_1(\lo,\hi)}-1\Big) -(\lo-1)\geq 0 \ .
\end{equation*}
Assume that $\lim_{\hi\to\infty} g_1(\hi)=0$ and $\lim_{\hi\to\infty} g_2(\hi)\leq 0$.
It follows that there must be a $\hi'$ with $g(z'(\lo,\hi'),\lo,\hi')< 0$ and thus $\min_z g(z'(\lo,\hi'),\lo,\hi')<0$.

\begin{itemize}
 \item $\lim_{\hi\to\infty} g_1(\hi)= 0$: Assume that it holds $\lim_{\hi\to\infty} z'=1$.
Lemma~\ref{lem:f(z)}~$(\ref{prop:f(z)_lim_1})$ gives that $\lim_{\hi\to\infty}f(z'(\hi))=0=\lim_{\hi\to\infty} g_1(\hi)$.
\begin{align*}
\lim_{\hi\to\infty} z'=&\lim_{\hi\to\infty}\exp\left( \ln\left(\frac{\lo}{\hi}\right)^{\frac{1}{\hi-\lo}} \right)
=\lim_{\hi\to\infty}\exp\left( \frac{\ln(\lo)-\ln(\hi)}{\hi-\lo}  \right) \\
=&\exp\left( \lim_{\hi\to\infty} \frac{\ln(\lo)}{\hi-\lo}  -\lim_{\hi\to\infty}\frac{\ln(\hi)}{\hi-\lo}  \right) 
=\exp\left(0-\lim_{\hi\to\infty} \frac{1/\hi}{1}\right)=1 \ .
\end{align*}
 \item  $\lim_{\hi\to\infty} g_2(\hi)\leq 0$:
Since $\frac{2-z}{1-z}>\frac{1}{1-z}$, for $z\in(0,1)$, it is sufficient to show that
$\lim_{\hi \to \infty} \frac{2-z'}{1-z'}\cdot  \frac{1}{\hi-1}=0$. Hence
\begin{align*}
 &\lim_{\hi \to \infty} \frac{(2-z')/(\hi-1)}{1-z'} 
=  \lim_{\hi \to \infty}\frac{ - \frac{1}{\hi-1}\cdot \frac{\partial z'}{\partial \hi}  +   \frac{-2+z'}{(\hi-1)^2}}{-\frac{\partial z'}{\partial \hi}}\\
=&  \lim_{\hi \to \infty} \left( \frac{1}{\hi-1} + \frac{\frac{1}{(\hi-1)^2}\cdot (2-z')}{\frac{\partial z'}{\partial \hi}} \right) \ .
\end{align*}
Using that $\frac{\partial z'}{\partial \hi}=z'\cdot \left( \frac{-\ln\left(\lo/\hi \right)}{(\hi-\lo)^2}-\frac{1}{\hi \cdot (\hi-\lo)} \right)$ we get
\begin{align*}
&\lim_{\hi \to \infty} \frac{(2-z')/(\hi-1)}{1-z'} =  \lim_{\hi \to \infty} \frac{2-z'}{(\hi-1)^2}\cdot\frac{1}{z'}\cdot  \left(  \frac{-\ln(\frac{\lo}{\hi})}{(\hi-\lo)^2} -\frac{1}{\hi\cdot(\hi-\lo)} \right)^{-1}\\
=&  \lim_{\hi \to \infty} \underbrace{\frac{2-z'}{z'}}_{\to 1}\cdot \Big( \underbrace{(\hi-1)^2\cdot \frac{-\ln(\frac{\lo}{\hi})}{(\hi-\lo)^2}}_{\to \infty} -\underbrace{\frac{(\hi-1)^2}{\hi\cdot(\hi-\lo)}}_{\to 1} \Big)^{-1}=0 \ .
\end{align*}
\end{itemize}
\end{asparaenum}

\section{Properties of $h(z,\lo,\hi)$}
\label{app:h(z)}
In this section we prove Lemma~\ref{lem:h(z)}.
\begin{asparaenum}[$(i)$]
\item  Consider the denominator of $h(z)$.
\begin{equation*}
\hi\cdot ((\hi-1)\cdot f(z)-1) \overset{!}{=}0 \Leftrightarrow f(z)=\frac{1}{\hi-1} \ ,
\end{equation*}
which is true for exactly one $z$ from $(0,1)$, which is per definition $z = z_r$.
\item With $\lim_{z\to 0}f(z)=1$, Lemma \ref{lem:f(z)}$(\ref{prop:f(z)_lim_0})$, and 
$\lim_{z\to 0}(1-f(z)\cdot(\lo-1))=2-a\leq -1$ we get
\begin{equation*}
\lim_{z\to 0}h(z)
=\lim_{z\to 0} \frac{  \lo\cdot \frac{1}{z^{\hi -\lo}}\cdot(1-f(z)\cdot(\lo-1)) -\hi +f(z)\cdot \hi\cdot(\hi-1)}{\hi\cdot ((\hi-1)\cdot f(z)-1)}
{=}-\infty \ .
\end{equation*}

\item It holds $f(z)>\frac{1}{\hi-1}, \forall z\in (0,z_r)$. Let
$f(z)=\frac{1+\eps}{\hi-1}$. Consider the limit of the numerator of $h(z)$.
\begin{align*}
 &\lim_{\eps\to 0} \lo\cdot \tfrac{1}{z_r^{\hi -\lo}}\cdot(1-\tfrac{1+\eps}{\hi-1}\cdot(\lo-1)) -\hi +\tfrac{1+\eps}{\hi-1}\cdot \hi\cdot(\hi-1) \\
=&\lim_{\eps\to 0} \lo\cdot \tfrac{1}{z_r^{\hi -\lo}}\cdot(1-\tfrac{\lo-1}{\hi-1})=K \ ,
\end{align*}
for some positive constant $K$ (depending on $\lo$ and $\hi$).
For the denominator of $h(z)$ it holds
\begin{equation*}
\lim_{\eps\to +0} \hi \cdot((\hi-1)\cdot \tfrac{1+\eps}{\hi-1}-1)=\hi \cdot \eps=+0 \ .
\end{equation*}
Hence $\lim_{z\to z_r}h(z)=+\infty$.
\item It holds $f(z)\geq\frac{1}{\lo-1}, \forall z\in (0,z_l]$. Let $\eps \geq 0$ and let $f(z)=\frac{1+\eps}{\lo-1}$. Hence
\begin{align*}
   h(z)=&\frac{ \lo\cdot z^{\lo -\hi} -\hi -\frac{1+\eps}{\lo-1} \cdot (\lo\cdot(\lo-1)\cdot z^{\lo-\hi}-\hi\cdot(\hi-1)) }{\hi\cdot ((\hi-1)\cdot \frac{1+\eps}{\lo-1}-1)}\overset{!}{\leq} 1 \\
\Leftrightarrow & \lo \cdot z^{\lo-\hi}-(1+\eps)\cdot \lo \cdot z^{\lo-\hi} -\hi + (1+\eps)\cdot\hi\cdot \tfrac{\hi-1}{\lo-1}\leq -\hi + (1+\eps)\cdot\hi\cdot \tfrac{\hi-1}{\lo-1}\\
\Leftrightarrow & \lo \cdot z^{\lo-\hi} -(1+\eps)\cdot \lo \cdot z^{\lo-\hi}\leq 0 \Leftrightarrow \eps \geq 0 \ .
\end{align*}
It holds $\frac{1}{\lo-1}\geq f(z)>\frac{1}{\hi-1},\forall z \in (z_l,z_r)$. 
Let $\eps> 0$, and let $\delta>0$ with $\frac{1-\eps}{\lo-1}=f(z)=\frac{1+\delta}{\hi-1}$.
Hence 
\begin{align*}
h(z)=&\frac{ \lo\cdot z^{\lo -\hi} -\hi -\frac{1-\eps}{\lo-1} \cdot \lo\cdot(\lo-1)\cdot z^{\lo-\hi} + \frac{1+\delta}{\hi-1}\cdot \hi\cdot(\hi-1)) }{\hi\cdot ((\hi-1)\cdot \frac{1+\delta}{\hi-1}-1)}\overset{!}{>}1\\
\Leftrightarrow &\lo\cdot z^{\lo -\hi}  -(1-\eps)\cdot \lo\cdot z^{\lo-\hi} + (1+\delta)\cdot \hi-\hi > \delta\cdot \hi\\
\Leftrightarrow &\eps\cdot \lo\cdot z^{\lo-\hi} > 0 \Leftrightarrow \eps > 0 \ .
\end{align*}
Note that for $z=z_l$, that is $\eps=0$, we have $h(z)=1$.

\item It holds $f(z)<\frac{1}{\hi-1}, \forall z\in (z_r,1)$. Let $\eps \in (0,1)$ and let $f(z)=\frac{1-\eps}{\hi-1}$. Hence
\begin{align*}
h(z)=&\frac{ \lo\cdot z^{\lo -\hi} -\hi -\frac{1-\eps}{\hi-1} \cdot (\lo\cdot(\lo-1)\cdot z^{\lo-\hi}-\hi\cdot(\hi-1)) }{\hi\cdot ((\hi-1)\cdot \frac{1-\eps}{\hi-1}-1)} \overset{!}{<} 1 \\
\Leftrightarrow & \lo \cdot z^{\lo-\hi} \cdot (1-(1-\eps)\cdot \tfrac{\lo-1}{\hi-1}) -\hi + \hi\cdot (1-\eps) > -\eps\cdot \hi \\
\Leftrightarrow &  1-(1-\eps)\cdot \tfrac{\lo-1}{\hi-1} >0  \Leftarrow \eps \in (0,1) \ .
 \end{align*}

\item We will use the following representation of $h(z)$
\begin{equation*}
h(z) =\frac{Z_0(z)-f(z) \cdot Z_1(z)}{\hi \cdot z^{\hi-1} \cdot \left((\hi-1)\cdot f(z)-1\right)} \ .
\end{equation*}
Let $h_1(z)=(\hi-1)\cdot f(z)-1$. Note that $h_1(z)\neq 0$, if $z\neq z_r$.

The first derivative of  $h(z)$ is
\begin{align*}
\frac{\partial h(z)}{\partial z}=&\frac{f(z)}{\hi \cdot z^\hi \cdot {h_1(z)}} \\
&\cdot \Bigg( \frac{Z_1(z)}{1-z}-Z_2(z) -(Z_0(z)-f(z)\cdot Z_1(z))\cdot  \frac{(\hi-1)\cdot (\hi-1-\frac{1}{1-z})}{{h_1(z)}} \Bigg) \ .
\end{align*}
The function $h(z)$ is strictly increasing  if and only if 
\begin{align*}
  &\frac{\partial h(z)}{\partial z}\overset{!}{>}0 \\
 \Leftrightarrow& \frac{\partial h(z)}{\partial z}\cdot h_1(z)^2
=f(z) \cdot \Big[( h_1(z) \cdot \big( \tfrac{Z_1(z)}{1-z}-Z_2(z)\big) \\
 &\hspace{3.4cm}-\big(Z_0(z)-f(z)\cdot Z_1(z)\big)\cdot  (\hi-1)\cdot (\hi-1-\tfrac{1}{1-z}) \Big] > 0 \ .
\end{align*}
Note that $f(z)$ is positive for $z\in(0,1)$. So we get
\begin{align*}
   \frac{\partial h(z)}{\partial z}\overset{!}{>}0
\Leftrightarrow (\tfrac{Z_1(z)}{1-z}-Z_2(z))\cdot& ( (\hi-1)\cdot f(z)-1) > 
\\&(Z_0(z)-f(z)\cdot Z_1(z))\cdot ( (\hi-1)^2- \tfrac{\hi-1}{1-z}) 
\end{align*}
This inequality is equivalent to
\begin{align*}
&\frac{Z_1(z)}{1-z}\cdot (\hi-1)\cdot f(z) - \frac{Z_1(z)}{1-z} -Z_2(z)\cdot (\hi-1)\cdot f(z) + Z_2(z)> \\
&\frac{Z_1(z)}{1-z}\cdot (\hi-1)\cdot f(z)+ Z_0(z)\cdot (\hi-1)^2\\ -&Z_0(z)\cdot \frac{(\hi-1) }{1-z} -f(z)\cdot Z_1(z)\cdot (\hi-1)^2\\
\Leftrightarrow 
&   f(z)\cdot (\hi-1) \cdot (Z_1(z)\cdot (\hi-1)-Z_2(z))
 + \frac{1}{1-z}\cdot(Z_0\cdot (\hi-1)-Z_1(z)) \\
& + Z_2(z)-Z_0(z)\cdot (\hi-1)^2 >0 \ .
\end{align*}
Expanding the functions $Z_j(z)$ gives
\begin{align*}
  \frac{\partial h(z)}{\partial z}\overset{!}{>}0
\Leftrightarrow&f(z)\cdot (\hi-1)\cdot \lo\cdot z^{\lo-1} \cdot( (\lo-1)\cdot(\hi-1)-(\lo-1)^2)\\
&+\tfrac{1}{1-z}\cdot \lo\cdot z^{\lo-1}\cdot( (\hi-1)-(\lo-1))\\
&+\lo\cdot z^{\lo-1}\cdot( (\lo-1)^2-(\hi-1)^2)>0\\
\Leftrightarrow
&f(z)\cdot(\hi-1)\cdot (\lo-1)\cdot(\hi-\lo)\\
&+ \frac{1}{1-z}\cdot(\hi-\lo) +(\lo-1)^2 -(\hi-1)^2 >0 \\
\Leftrightarrow
& f(z)\cdot (\hi-1)\cdot(\lo-1) + \frac{1}{1-z} +2 -\hi -\lo>0\\
 \Leftrightarrow & g(z)>0 \ ,
\end{align*}
where we divided by $\hi-\lo$ which is larger than $0$ by definition.

\item This follows directly from Lemma~\ref{lem:g(z)}$(\ref{prop:g(z)_left_interval})$ and Lemma~\ref{lem:h(z)}$(\ref{prop:h(z)_monotonicity})$.
\item This follows directly from Lemma~\ref{lem:g(z)}$(\ref{prop:g(z)_right_interval})$ and Lemma~\ref{lem:h(z)}$(\ref{prop:h(z)_monotonicity})$
\item For $\min_z g(z)>0$ this follows directly from Lemma~\ref{lem:h(z)}$(\ref{prop:h(z)_monotonicity})$.
      Let $\min_z g(z)=0$. According to Lemma~\ref{lem:g(z)}$(\ref{prop:g(z)_monotonicity})$ the point $z_{\min}=\arg\min_z g(z)$ is the only point
      where $g(z)=0$. It follows that $z_{\min}$ is the only inflection point of $h(z)$. Therefore, 
      according to Lemma~\ref{lem:h(z)}$(\ref{prop:h(z)_monotonicity})$,
      $h(z)$ is strictly increasing.
\item According to Lemma~\ref{lem:g(z)}$(\ref{prop:g(z)_zeros})$ the function $g(z)$ has exactly two different roots $z_1, z_2$
in the interval $(z_l,z_r)$ and according to Lemma~\ref{lem:h(z)}$(\ref{prop:h(z)_monotonicity})$ and Lemma~\ref{lem:g(z)}$(\ref{prop:g(z)_monotonicity})$
it follows that $h(z)$ is strictly increasing for $z<z_1$, strictly decreasing for $z$ with $z_1<z<z_2$, and strictly increasing 
for $z>z_2$. Hence the claim follows.
\end{asparaenum}

\section{Special Points}
\label{app:points}
In this section we prove Lemma~\ref{lem:special_points}.
\begin{compactenum}[$(i)$]
 \item follows from Lemma~\ref{lem:f(z)}$(\ref{prop:f(z)_decreasing})$,$(\ref{prop:f(z)_lim_0})$,$(\ref{prop:f(z)_lim_1})$.
 \item follows from Lemma~\ref{lem:g(z)}$(\ref{prop:g(z)_zeros})$.
 \item follows from Lemma~\ref{lem:f(z)}$(\ref{prop:f(z)_z'})$.
 \item According to the definition of $g(z)$ we have $f(z')=g(z')-\frac{1}{1-z'}-2+\hi+\lo$. 
Assume that $z'=z_1$ or $z'=z_2$, then $g(z')=0$ and $f(z')=-\frac{1}{1-z'}-2+\hi+\lo$ which is contradiction to Lemma~\ref{lem:f(z)}$(\ref{prop:f(z)_phi})$.
\end{compactenum}

\end{document}